\documentclass[11pt]{article}
\usepackage{times}
\usepackage[top=1in,bottom=1in,right=1in,left=1in]{geometry}
\usepackage{amssymb,amsmath,amsthm}
\usepackage{arydshln}
\usepackage{verbatim}
\usepackage{titlesec}
\usepackage{subfigure}
\usepackage[linesnumbered,ruled,noend]{algorithm2e}
\usepackage{bm}
\usepackage{pdfpages}
\usepackage{bbm}
\usepackage{hyperref} % must be loaded before cleveref
\usepackage[nameinlink, noabbrev, capitalize]{cleveref}
\usepackage{tikz}
\usepackage{booktabs}
\usepackage{makecell}

\usepackage[numbers]{natbib}

\numberwithin{equation}{section}

\usepackage{slivkins-setup}
\usepackage{slivkins-theorems}

\definecolor{forestgreen}{rgb}{0.13, 0.55, 0.13}
\usepackage[T1]{fontenc}

\usepackage[normalem]{ulem}

\newcommand{\reals}{\mathbb{R}}

\crefname{equation}{}{}
\crefrangeformat{equation}{(#3#1#4) --~(#5#2#6)}
\crefmultiformat{equation}{(#2#1#3)}{, (#2#1#3)}{, (#2#1#3)}{, (#2#1#3)}
\crefname{lem}{Lemma}{Lemmas}
\crefname{section}{Section}{Sections}
\crefname{subsubsubsection}{Section}{Sections}
\crefname{rem}{Remark}{Remarks}
\crefname{figure}{Figure}{Figures}
\crefname{table}{Table}{Tables}
\Crefname{lem}{Lemma}{Lemmas}
\crefname{thm}{Theorem}{Theorems}
\Crefname{thm}{Theorem}{Theorems}
\crefname{assn}{Assumption}{Assumptions}
\Crefname{assn}{Assumption}{Assumptions}

\newcommand{\term}[1]{\ensuremath{\mathtt{#1}}\xspace}
\newcommand{\GPD}{\term{GPD}} % generalized pacing dynamics

\newcommand{\PaceReg}{R_{\term{pace}}} % pacing regret
\newcommand{\uniObj}[1][]{\Phi^{\term{uni}}_{#1}} % unified objective
\newcommand{\uniObjFixed}{\uniObj[\term{fix}]} % ... for a fixed multiplier
\newcommand{\overmu}{\overline{\mu}}

\newcommand{\PacingSGD}{\term{PacingSGD}}
\newcommand{\PacingOGD}{\term{PacingOGD}}
\newcommand{\PacingOMD}{\term{PacingOMD}}
\newcommand{\PacingOFTRL}{\term{PacingOFTRL}}

\newcommand{\myGrad}[1][k,t]{\widetilde{\nabla}_{#1}}

\usepackage{commath}

\usepackage{color-edits}
\addauthor{as}{blue}   % as for Alex
\addauthor{bjl}{red}   % bjl for Brendan
\addauthor{bl}{violet} % bl for Bar
\addauthor{newbl}{orange}
\addauthor{yl}{cyan}
\addauthor{jg}{forestgreen}
\newcommand{\bjl}[1]{{\color{red}  [BJL: #1]}}

\newcommand{\boldsym}[1]{\ensuremath{\boldsymbol{#1}}}  % Use \boldsymbol

\makeatletter
\def\nomarkfootnote{\gdef\@thefnmark{}\@footnotetext}
\makeatother

\addtocontents{toc}{\protect\setcounter{tocdepth}{0}}

\title{Budget Pacing in Repeated Auctions:\\Regret and Efficiency without Convergence%
\footnote{The authors are grateful to Bach Ha (Microsoft Bing Ads) for many conversations that informed our perspective, and to the anonymous conference and journal referees for many helpful comments.\vspace{2mm}\newline
A preliminary version of this paper~\citep{Gaitonde-itcs23} appeared in \emph{ITCS 2023}: the 14th Conf. on Innovations in Theoretical Computer Science.\vspace{2mm}\newline
Compared to the initial working paper and the conference version, the current version features revised presentation and expanded discussions, as well as \textbf{several major updates}: our aggregate and individual guarantees are extended to several other gradient-based pacing algorithms (\Cref{sec:other}, since Dec'25), the regret analysis is extended to the utility-maximization objective (since Dec'25), and numerical experiments are added (\Cref{sec:expts}, since Aug'24, with additional baselines added in Dec'25).
\vspace{2mm}
}}
\author{Jason Gaitonde%
\footnote{Duke University, Durham NC, USA. Email: jason.gaitonde@duke.edu.
Many results were obtained while J. Gaitonde was a research intern at Microsoft Research New England and a graduate student at Cornell University, supported in part by NSF Award CCF-1408673 and AFOSR Award FA9550-19-1-0183.}
\and Yingkai Li%
\footnote{Department of Economics, National University of Singapore, Singapore.
Email: yk.li@nus.edu.sg. Research carried out while Y. Li was a graduate student at Northwestern University and a research intern at Microsoft Research NYC.}
\and Bar Light%
\footnote{Business School and Institute of Operations Research and Analytics, National University of Singapore, Singapore. Email: barlight@nus.edu.sg.
Most of the research was conducted while Bar Light was a postdoc at Microsoft Research NYC.}
\and Brendan Lucier%
\footnote{Microsoft Research, Cambridge MA, USA. Email: brlucier@microsoft.com.}
\and Aleksandrs Slivkins%
\footnote{Microsoft Research, New York NY, USA. Email: slivkins@microsoft.com.}
}

\date{%
    Initial version: February 2022\\
    This version: December 2025
}

\begin{document}

\begin{titlepage}
\maketitle
\thispagestyle{empty}

\begin{abstract}
We study the aggregate welfare and individual regret guarantees of dynamic \emph{pacing algorithms} in the context of repeated auctions with budgets.
Such algorithms are commonly used as bidding agents in Internet advertising platforms, adaptively learning to shade bids by a tunable linear multiplier in order to match a specified budget.
We show that when agents simultaneously apply a natural form of
gradient-based pacing, the liquid welfare obtained over the course of the learning dynamics is at least half the optimal expected liquid welfare obtainable by any allocation rule.
Crucially, this result holds \emph{without requiring convergence of the dynamics},
allowing us to circumvent known complexity-theoretic obstacles of finding equilibria.
This result is also robust to the correlation structure between agent valuations and holds for any \emph{core auction}, a broad class of auctions that includes first-price, second-price, and generalized second-price auctions as special cases.  For individual guarantees, we further show such pacing algorithms enjoy \emph{dynamic regret} bounds for individual utility- and value-maximization, with respect to the sequence of budget-pacing bids, for any auction satisfying a monotone bang-for-buck property. To complement our theoretical findings, we provide semi-synthetic numerical simulations based on auction data from the Bing Advertising platform.
%\asdelete{This generalizes known regret guarantees for bidders facing stochastic bidding environments in two ways: it applies to a wider class of auctions than previously known, and it allows the environment to change over time. \BLnew{Our results show that for a wide range of auctions, when bidders simultaneously employ   budget pacing algorithms they achieve  individual regret bounds, and the resulting allocation is approximately welfare optimal even when the corresponding dynamics do not necessarily converge.}}

%\ascomment{AE Q (4a) explicitly asked to mention linear policies in the abstract. Seems easier to oblige than to argue.} 
\end{abstract}

\end{titlepage}

\pagenumbering{arabic}

\iffalse
\section{Outline}

\begin{itemize}
    \item Model: single-dimensional allocation problem each round; value profile drawn from some distribution each round (independent across time); agent goal is to maximize total value subject to aggregate budget constraint.
    \item Fix any core auction.  If agents simultaneously run BG dynamics then LPoA is at most $2$.
    \item In a stochastic environment, single-agent BG dynamics achieves low regret, as long as auction satisfies monotone bang-per-buck.
    \item ?? In an adversarial environment, single-agent BG dynamics achieves low regret versus budget-pace-every-round benchmark.  (Only under a path-length assumption??)
    \item ?? In a simultaneous learning environment, as long as number of bidders is large and valuation distributions satisfy condition xyz, then achieve low dynamic regret ??
\end{itemize}
\fi

\iffalse
Todo list:
\begin{itemize}
    \item Abstract
    %\item Better title
    %\item Literature search on SGD methods in budget pacing.
    %
    %\item Remove assumption that $Z(0) \geq \rho$ for the regret result.
\end{itemize}
\fi

\section{Introduction}

Online advertising increasingly dominates the marketing landscape, accounting for $54.2\%$ of total media ad spending in the US in 2019 (~$\$129$ billion)  \citep{enberg_2019}.
Such ads are predominantly allocated by auction: advertisers submit bids to an Internet platform to determine whether they will be displayed as part of a given page view and at what price.  A typical advertiser participates in many thousands of auctions each day, across a variety of possible ad sizes and formats, payment options (pay per impression, per click, per conversion, etc.), and bids tailored to a variety of signals about user demographics and intent.  To further complicate the decision-making process from the  advertiser's perspective, these many auction instances are strategically linked through a \emph{budget constraint}, the total amount of money that can be allocated to advertising.
An advertiser therefore faces the daunting task of choosing how to appropriately allocate a global budget across a complex landscape of advertising opportunities, and then convert that intent into a bidding strategy.

To help address this difficulty,
all major online platforms provide automated budget management services that adjust campaign parameters on an advertiser's behalf.
This is commonly achieved via
%{a simple but flexible approach called}
\emph{budget-pacing}:
%rather than bidding directly,
an advertiser specifies a global budget target and a maximum willingness to pay (or ``value") for different {advertising opportunities},
%\bjledit{which}
and these
values
%maximums
%parameters
are then scaled down (or ``paced'') by a multiplier
into bids such
%in response to a variety of factors, including the competing bids,
that the realized daily spend matches the target budget.
%The campaign management system
%therefore
%uses
An algorithmic bidding agent learns, online, how best to pace the advertiser's bids as it observes auction outcomes.
This campaign management service
%\asdelete{has numerous advantages.  For one, it}
lowers the barrier to entry into the online advertising ecosystem and removes the need for the advertiser to constantly adjust their campaign in the face of changing market conditions.  Moreover, the platform is often better positioned to manage the budget since they have direct access to detailed market statistics.

Bidding agents are now near-universally adopted across all mature advertising platforms, but this success raises some pressing questions about the whole-market view.  What can we say about the aggregate market outcomes when nearly all advertiser spend is controlled by automated bidding agents that are simultaneously learning to pace their bids?  And to what extent does this depend on the details of the underlying auction?

%\asdelete{We address these questions in a dynamic environment where bidding agents are simultaneously learning to pace their spend in a repeated auction.}
Central to our question is the interplay between individual learning and aggregate market efficiency, each of which has been studied on its own.  For example, when each advertising opportunity is sold by a second-price auction, it is known that linear bidding strategies {(i.e., mappings from maximum willingness to pay to a bid)}
are in fact optimal over all possible
{bidding strategies} % bidding functions
 for both utility-maximizing and value-maximizing agents, and that gradient-based methods can be used by an agent to achieve vanishing regret relative to the best bidding strategy in hindsight~\citep{Borgs-www07,BalseiroGur19}.  On the other hand, when multiple bidding agents participating in second-price auctions choose pacing factors that form a pure Nash equilibrium, the resulting outcome is known to be approximately efficient (in the sense of maximizing expected liquid welfare; more on this below)~\citep{Gagan-wine19,Babaioff-itcs21}.
At first glance this combination of results seems to address the question of aggregate performance of bidding agents in second-price auctions.  %However, it is well-known that
But convergence of online learning algorithms to a Nash equilibrium,
{let alone a pure Nash equilibrium, is notoriously difficult to guarantee and}
%, let alone a pure one,
should not be taken for granted. Moreover,
%\asdelete{a recent result by \citet{Chen-ec21} shows that}
finding a pure Nash equilibrium of the pacing game is PPAD-hard for second-price auctions \citep{Chen-ec21}, and hence we should not assume that bidding agents employing polynomial-time online learning strategies will efficiently converge to a {pure Nash equilibrium in the full generality of second-price auctions}.  So the question remains: \emph{if bidding agents do not converge, what happens to overall market performance?}

\subsection{Our Contributions}

%
%Our high-level contribution is as follows.
%\bjl{ITCS Suggestion: further emphasize the innovation of bounding outcomes without convergence.}
%
%
%
We provide (classes of) bidding algorithms that simultaneously admit good \emph{aggregate guarantees}
in terms of overall market efficiency, \emph{without relying on convergence,}
%performance of the system,
while still providing good \emph{individual guarantees} as online learning algorithms that benefit a particular advertiser. Closely related are three literatures: (i) on aggregate outcomes in single-shot budget-constrained ad auctions, without regard to bidding dynamics, (ii) on online learning with budget constraints, without regard to the aggregate performance, and (iii) conditions under which learning agents converge to equilibrium in repeated games.  With this perspective in mind, we match a state-of-art aggregate guarantee from (i), while being qualitatively on par with state-of-art individual guarantees in (ii), without relying on convergence and thereby side-stepping
conditions from
%the associated strong assumptions from
(iii).
% \bjl{For example, ``first-price auction'' is an example of a condition.}
We accomplish this with bidding algorithms that are arguably quite natural and
for a broad class of auctions.
%we handle the general class of auctions described above.}
%

In our model there are $T$ rounds, each corresponding to an auction instance.
%All of our results apply to multiple auction rules, including first-price, second-price, and GSP auctions.%
%\footnote{\label{fn:intro-non-truthful} For non-truthful auctions the restriction to linear bidding strategies is not without loss. \bjledit{Our individual regret guarantees are therefore with respect to a \emph{perfect pacing benchmark} that corresponds to an optimal linear pacing strategy in stochastic environments.}
%%to the best {linear} bidding strategy.
%See our description of the individual regret guarantees later in this section for further discussion.}
%
%of a divisible good.%
%\footnote{\asedit{That is, the good can, in principle, be divided fractionally among the agents. An  inherently indivisible item can be interpreted as divisible by allocating it probabilistically.}}
%specific advertising opportunity can be we can interpret the good as click probability, which can be divided among the advertisers by choosing different ad placements.}
 The {private} values (i.e., maximum willingness to pay) observed by the agents are randomly drawn in each round and can be arbitrarily correlated with each other, capturing scenarios where the willingness to pay of different advertisers is correlated through characteristics of the impression.%
\footnote{{An \emph{impression} is the industry term for an ``atomic" advertising opportunity: a  specific slot on a specific webpage when this webpage is rendered for a specific user. Impression characteristics depend on the slot, the page, and the user.}}
In each round the bidding agents place bids on behalf of their respective advertisers.  The agents operate independently of each other, interacting only through the feedback they receive from the auction.

\xhdr{Algorithm(s).}
We focus on a gradient-based pacing algorithm (\Cref{alg:bg}) that was first introduced by \citet{BalseiroGur19} in the context of utility maximization in second-price auctions. They derive this algorithm via the Lagrangian dual of the (quasilinear) utility maximization problem and establish optimality of linear pacing for utility maximization.
In practice, gradient-based pacing also sees widespread use in richer allocation problems (such as for multiple ad slots) and auction rules (such as first-price auctions)~\cite{aggarwal2024auto}. Even though the utility optimality guarantees do not extend to all such settings, the algorithm can nevertheless be interpreted as maximizing value subject to linearity, budget, and maximum bid constraints.
%
%All of our results extend to apply to multiple auction rules, including first-price, second-price, and GSP auctions.%
Motivated by this broader usage, we extend the definition and analysis of \Cref{alg:bg} to a richer class of allocation problems and auction formats, including first-price, second-price, and GSP auctions, to which all of our results apply.\footnote{\label{fn:intro-non-truthful} For non-truthful auctions the restriction to linear bidding strategies is not without loss. Our individual regret guarantees are therefore with respect to optimal \emph{linear} pacing strategies,
%\ascomment{Brendan: trying to side-step the stochastic vs adversarial distinction, which is "too soon" and irrelevant to this point.}
%a \emph{perfect pacing benchmark} that corresponds to an optimal linear pacing strategy in stochastic environments.
%to the best {linear} bidding strategy.
see our description of the individual regret guarantees later in this section.}
%for further discussion.}
%, and show that it achieves good individual regret guarantees for second-price auctions.
%The gradient-based pacing algorithm we study is a direct extension of their method
%We directly extend their algorithm to a richer class of allocation problems and auction formats.
%Our main result above shows that this algorithm attains good aggregate performance even beyond second-price auctions.
Underlying this extension is a modified interpretation of this algorithm as stochastic gradient descent on a certain artificial objective that applies even beyond second-price auctions.

While our exposition focuses on \Cref{alg:bg}, all our guarantees extend to several other gradient-based algorithms. Essentially, \Cref{alg:bg} invokes on an update step akin to stochastic gradient descent, and this step can be replaced with several other well-known algorithms from online convex optimization:
optimistic gradient descent
\citep{Rakhlin-colt13,Rakhlin-nips13},
optimistic mirror descent
\citep{Chiang-colt12,Rakhlin-colt13,Rakhlin-nips13},
and optimistic FTRL \cite{Rakhlin-colt13}.

\xhdr{First Result: Aggregate Market Performance.}
We prove that when the bidding agents employ {\Cref{alg:bg}}
%a simple form of gradient descent
to tune their pacing multipliers,
%they achieve individual regret bounds (more on this below),} and
the resulting market outcome over the full time horizon achieves at least half of the optimal expected \emph{liquid welfare}. {Crucially, this guarantee  does not depend on the convergence of the algorithms' actions to an equilibrium of the bidding game. Nevertheless, it matches the best possible guarantee even for a pure Nash equilibrium in a static truthful auction \citep{Gagan-wine19,Babaioff-itcs21}.}

%\ascomment{LW discussion got bloated and a bit disorganized. Moved it to a separate para (from the middle of the previous para), and revised.}

Liquid welfare is the maximum amount that the agents are jointly willing to pay for a given allocation. Put differently, it is the maximum revenue that can be extracted for this allocation by an all-knowing auctioneer. Liquid welfare coincides with \emph{compensating variation} when specialized to our setting.%
\footnote{\emph{Compensating variation}, a standard notion in economics, is the amount that agents would need to pay (or be paid) to return to their original utility levels after some change, such as a change in prices~\cite{perloff2009microeconomics}. When interpreting liquid welfare as compensating variation, the change being considered is the allocation itself. See \Cref{app:cv} for further discussion.}
Conveniently, it is a welfare measure that applies even when agents seek to maximize \emph{value} (e.g., number of clicks or impressions received) subject to constraints, rather than a monetary utility objective.%
\footnote{When the agents' objective \emph{can} be expressed in dollars, such as the objective of maximizing advertiser utility subject to the budget constraints, utilitarian welfare would be reasonable aggregate objective. However, strong impossibility results are known even in a single-shot (non-repeated) setting for a single good \citep{Shahar-icalp14}.
Thus, liquid welfare is a meaningful notion of social surplus in budgeted environments, and it specializes to utilitarian welfare when budgets are infinite.}
It has become a standard objective in the analysis of budget-constrained auctions
 \citep{Shahar-icalp14,Azar-wine17,Gagan-wine19,Babaioff-itcs21}.

While our discussion so far has focused mainly on second-price single-item auctions,
%\asdelete{which have received the lion's share of attention in the literature.  But}
our approximation result actually holds for a far richer set of allocation problems and auction formats, including those used in real-world advertising platforms.
%Importantly, our model is not limited to a specific auction format.
%
% Removed ``combinatorial'' since that might suggest multi-dimensional types, which we don't support.
We allow arbitrary downward-closed constraints on the set of feasible allocations of a single divisible good,%
\footnote{\Ie a good that can be divided fractionally among the agents.
% An inherently indivisible
Any item can be interpreted as divisible via probabilistic allocation.}
which captures single-item auctions as well as complex settings such as sponsored search auctions with multiple slots and separable click rates.
Further, our result applies even when the underlying mechanism is not %a second-price auction or even
truthful.  We accommodate any \emph{core auction}: an auction
%of a divisible good
that generates outcomes in the core, meaning that no coalition of agents could improve their joint utility by renegotiating the outcome with the auctioneer \citep{ausubel2002ascending}. This {is a well-studied class of auctions that} includes both first and second-price auctions, as well as the generalized second price (GSP) auction, and has previously been studied in the context of advertising auctions \citep{goel2015core,hartline2018fast}.
We emphasize, however, that the problem remains non-trivial (and almost as challenging) in a much simpler model with a repeated single-item second-price
auction and constant private values.

\xhdr{Second Result: Individual Regret Guarantees.}
%\citet{BalseiroGur19} showed that gradient-based pacing achieves good individual regret guarantees for second-price auctions, and \citet{Balseiro-BestOfMany-Opre} extend that result to all truthful auctions.
We have analyzed
%shown bounds on
aggregate market performance for a broad class of (possibly non-truthful) auctions including first-price and GSP auctions, but is gradient-based pacing an effective learning method in those settings?
%Prior work has established individual
Regret guarantees are known when the underlying auction is truthful and the environment is stochastic~\citep{BalseiroGur19,Balseiro-BestOfMany-Opre}, but what about non-truthful auctions and non-stationary environments?
To address this question,
%we next turn our attention to individual regret guarantees.}
%\citet{BalseiroGur19} showed that gradient-based pacing achieves good individual regret guarantees for second-price auctions, and \citet{Balseiro-BestOfMany-Opre} extend that analysis to all truthful auctions.
%Since our aggregate welfare results extend beyond truthful auctions, we next show that gradient-based pacing achieve good individual regret guarantees beyond truthful auctions as well.
we bound the regret obtained by an individual bidding agent participating in any auction format that satisfies a \emph{monotone bang-per-buck} condition, which implies that the marginal value obtained per dollar spent weakly decreases in an agent's bid.  For example, first and second-price auctions satisfy this condition, as does the GSP auction.
{We first consider the \emph{stochastic}  environment where the profile of opposing bids is drawn independently from the same distribution in each round. Our benchmark is the pacing multiplier that optimally realizes the desired budget in hindsight.\footnote{{Pacing multipliers are defined to only adjust bids downward relative to the maximum willingness to pay.  Thus, if the budget is not exhausted even when setting the bid equal to the value each round, this benchmark corresponds to a pacing multiplier of $1$.}}
We obtain regret $O(T^{3/4})$ relative to this benchmark.}%
\footnote{A typical goal in regret minimization is regret  $\tilde{O}(T^\gamma)$ for some constant $\gamma\in [\tfrac12,1)$. As a baseline, regret $O(\sqrt{T})$ is the best possible in the worst case, even in a stochastic environment with only two possible actions \citep{bandits-exp3}.}
%
%\asmargincomment{reworded to separate the points.}
{This regret bound applies for the objective of utility maximization, value maximization, or any convex combination thereof.}

When the underlying auction is truthful, our benchmark is known to be optimal over the class of all possible bidding strategies (i.e., mappings from value to bid, which may not be linear), for utility-maximizing agents~\cite{Conitzer-wine18,BalseiroGur19,balseiro2021budget}.  This means that, for truthful auctions such as the second-price auction, our benchmark for regret is actually the utility-optimal bidding policy in hindsight.
%Our regret bounds generalize to non-truthful auction formats as well, in which case the benchmark is the best linear policy.
%
But even beyond truthful auctions, we argue that an appropriate benchmark is
the best linear policy (i.e., best pacing multiplier).
%\bjledit{the best linear policy that expends the target budget without overbidding (or, if exhausting the budget is not possible, the linear policy that maximizes spend).}
%for multiple reasons.
First, while we abstractly model agent values as a willingness to pay per impression, in practice the variation in values is primarily driven by click rate estimates that are internal to the platform.
In such an environment, an advertiser {whose bidding algorithm is external to the platform}
%bidding manually
would necessarily be limited to a linear policy.
{Indeed, if the algorithm cannot access}
%it is not possible to directly condition one's bid on
the platform's click rate estimates,
then a bidding ``strategy" simply reduces to a single real-valued bid that would be (linearly) multiplied by the click rate; see \Cref{app:external} for a more formal discussion.
%an impression with half the click rate would necessarily receive half the effective per-impression bid.
The benchmark therefore tracks the best performance that one could achieve in hindsight with {an externally provided bid.}
%given variation in click rates.
Second, linear pacing is commonly used in practice as an algorithmic bidding policy even for non-truthful auctions~\citep{Conitzer-ec19}, so from a practical perspective it is useful to focus attention on linear pacing
%as a goal unto itself.
policies.
Third, linear pacing is reasonable from the online learning perspective: it is typical to choose a subset of policies as a hypothesis class {(even if this class is not known to contain an optimal policy)}, and the set of linear policies is a common and natural class to consider.
%focus on a particular subclass of policies in the face of infeasibility, and linear policies is a common and natural class to consider.

While the discussion above is framed in a stochastic environment, we prove an even stronger individual guarantee by permitting the opposing bids to change adaptively and adversarially based on the auction history. (Indeed, realistic auction environments are not necessarily stochastic from the individual bidder's perspective, because the other agents' bidding algorithms may be revising their bids.) In such an environment, we show that gradient-based pacing achieves vanishing regret relative to the \emph{perfect pacing sequence}, which is the sequence of pacing multipliers such that the expected spend in each round is precisely the per-round budget.%
\footnote{This guarantee is parameterized by the total round-to-round change in the ``perfect" pacing multipliers.}
In a stochastic environment, this perfect pacing sequence is precisely the single best fixed pacing multiplier in hindsight. More generally, this sequence may not be uniform and is not necessarily the sequence that maximizes expected utility or value subject to the budget constraint. Achieving low regret against this stronger benchmark
%the latter objective
in an adversarial environment is essentially hopeless (more on this in Section~\ref{sec:related}).
%in budget-constrained bidding, as well as other budget-constrained online learning problems, one is doomed to approximation ratios even in relatively simple examples \citep{BalseiroGur19,AdvBwK-focs19}.
Therefore, we suggest the perfect pacing sequence as a reasonable and tractable benchmark for this problem, and one particularly suitable for our algorithm (see \Cref{rem:regret-benchmark}). In fact, following this sequence (i.e., matching a target spend rate as closely as possible across time) can be a natural and desirable goal for a budget management system.

\xhdr{Numerical Simulations.}
To complement our theoretical findings, we provide semi-synthetic  numerical simulations of \cref{alg:bg} based on auction and campaign data from {Microsoft's} Bing Advertising platform. {Motivated by the impossibility of {sublinear-}regret guarantees against this benchmark in adversarial environments, we simulate {\emph{self-play}}: the progression of a multi-player environment in which the competing bidding agents engage in simultaneous learning.
%\asmargincomment{moved up the ``self-play" sentence}
We consider the utility-maximization objective for second-price payment rules, and the value-maximization objective for both first- and second-price payment rules.} We focus on regret relative to the standard benchmark: the best fixed pacing multiplier in hindsight. We find numerically that simultaneous execution of \cref{alg:bg} yields vanishing regret in our simulations, with regret rate less than $T^{3/5}$.

%\asmargincomment{Para break, and split into 2 sentences.}
Furthermore, we compare {\cref{alg:bg}} to other common online learning methods by evaluating both the (individual) empirical regret rate and (aggregate) liquid welfare obtained during self-play. {Our comparators are based on} Adaptive Moment Estimation (Adam)~\cite{kingma2014adam}, Multiplicative Updating (MU)~\cite{Borgs-www07}, and Optimistic Gradient Descent (OGD)~\cite{rakhlin2013online}. We find that Adam has an improved regret rate at the expense of lower aggregate liquid welfare, MU is dominated by \cref{alg:bg} on both liquid welfare and regret rate, and OGD has comparable performance to \cref{alg:bg} on both measures.

\xhdr{Map of the paper.}
We detail our model in \Cref{sec:model}. The main results --- liquid-welfare and regret guarantees for \Cref{alg:bg} --- are presented, resp., in   \Cref{sec:liquid.welfare,sec:regret}. \Cref{sec:other} extends them to several other algorithms. In \Cref{sec:expts} we complement these theoretical findings with numerical simulations. Additional discussions and some details of the proofs and the simulations are relegated to the appendices.

\subsection{Related Work}
\label{sec:related}

\iffalse
\begin{itemize}
    \item Prior work on pacing dynamics: Balseiro-Gur, Borgs-Chayes-Immorlica, et al first-price dynamics etc, Conitzer pacing papers
    \item Papers on convergence in general auctions without budgets (from email threads: Alex's mentee, Nisan papers if it's okay to reference?)
    \item PPAD-hardness results from Columbia group (EC 2021)
    \item Welfare: Brendan's ITCS paper, which implies it in special cases like Balseiro-Gur monotone games. Not sure if there are others?  (Yes -- a 2019 WINE paper ``Autobidding with constraints'')
    \item Learning side: special case of adversarial BwK, stochastic BwK papers for auctions, convergence in monotone games (Daskalakis and co.)
\end{itemize}
\fi

\vspace{-2mm}\xhdr{Pacing in Ad auctions.}
Budget-pacing is a popular approach for repeated bidding under budget {constraints}, both in practice and in theory \citep{balseiro2021budget}. %This approach is especially natural for second-price auctions (and truthful auctions more generally) because, given the valuations and competing bids, a fixed pacing multiplier is an optimal-in-hindsight bidding strategy \citep{BalseiroGur19}.  The most relevant results on budget-pacing algorithms are surveyed below.
%
%
% there always exists an optimal-in-hindsight bidding strategy corresponding to a single pacing multiplier  which exactly paces the budget \citep{BalseiroGur19}.
%
%
 %provide a comprehensive study and
%provide a comprehensive study and comparison of various budget management strategies both from the theoretical and empirical sides.
%
%
%\ascomment{next para is only about convergence}
%
%
\citet{BalseiroGur19} attain convergence guarantees in repeated second-price auctions, under strong convexity-like assumptions.
%\citet{BalseiroGur19} and \citet{Borgs-www07} study the aggregate outcomes of budget-pacing algorithms, focusing, resp., on first and second-price auctions.
%Prior work has also studied the outcomes of \emph{pacing dynamics}, especially in the context of second-price auctions.
%The work closest to ours is \citet{BalseiroGur19}, which derives a pacing algorithm by applying stochastic gradient descent to an appropriate Lagrangian relaxation of the utility maximization problem in second-price auctions .
%\citet{BalseiroGur19} establish several results for both stochastic and adversarial environments, as well as convergence to an equilibrium, under strong convexity assumptions.
%For simultaneous learning, they further establish convergence to equilibrium under a \emph{strong monotonicity} assumption on the utilities of the agents.
%Similar results for convergence in smooth monotone games under natural learning dynamics have also been established in prior work (for recent work on this, as well as more discussion, see Golowich, Pattathil, and Daskalakis \citep{DBLP:conf/nips/GolowichPD20}).
\citet{Borgs-www07} attain a similar result for first-price auctions, without convexity assumptions (via a different algorithm).
%using a similar (but technically different) algorithm.
% \citet{Borgs-www07} show that a similar (but technically different) budget-pacing algorithm converges to a Nash equilibrium
%%in the setting of second-price and first-price auctions, and they show that a similar pacing dynamic provably converges to an equilibrium
%for first-price auctions, without further convexity assumptions.
%%However, their results do not have any bearing on any efficiency metric such as (liquid) welfare.
Our emphasis on welfare guarantees
% in repeated auctions with budgets
\emph{without requiring convergence} appears novel,
%, to the best of our knowledge; indeed, this appears
and possibly necessary given the aforementioned $\mathsf{PPAD}$-hardness result \citep{Chen-ec21}.

%\ascomment{and this para is only about individual guarantees, and discusses \citet{Balseiro-BestOfMany-Opre}}

\citet{BalseiroGur19} establish individual guarantees {for utility-maximization in}
%for a budget-pacing algorithm
for repeated second-price auctions:
$\sqrt{T}$ regret rates for the stochastic environment and approximation guarantee for the adversarial environment,
% both stochastic and adversarial environments ,
under various convexity assumptions. \citet{Balseiro-BestOfMany-Opre} extend similar guarantees to repeated \emph{truthful} auctions, without convexity assumptions.
%\asdelete{In particular, both papers obtain regret-optimal guarantees for the stochastic environment.}
They also obtain regret bounds for several (specific types of)  non-stationary stochastic environments.
%obtain bounds on dynamic regret for an adversarial environment that does not deviate too much from a stochastic one.
Our individual guarantee is different from theirs in several respects: (i) it applies to a much wider family of auctions, (ii) it applies to the adversarial environment,
%the problem instance can deviate arbitrarily far from any fixed stochastic instance,
%\footnote{Our guarantee bounds the ``path-length" of the perfect pacing sequence.}
but (iii) the benchmark is the perfect pacing sequence rather than the best outcome in hindsight.
%
%Furthermore, our individual regret guarantees hold under a bandit feedback model, where the agent learns only the allocation and payment obtained by their chosen bid and no further information about the other bids.

%; in contrast, \citet{Balseiro-BestOfMany-Opre} assume the

%\bjlcomment{Removed the note about prior work; this is redundant given that we also describe subsequent work later on.}
%\bjldelete{We are not aware of any work that appeared prior to the initial publication of this work that bounds individual regret for \emph{non-truthful auctions}, even in the stochastic environment and even under the more relaxed full-feedback model, but later in this section we describe a relevant literature that appeared subsequently.}
%\asmargincomment{Brendan: your new cites are subsequent work!}

%\bjldelete{In particular, the algorithms in \citep{Balseiro-BestOfMany-Opre} do not apply, as they require full feedback before each round.}
%Generic algorithms from contextual bandits with knapsacks might be applicable for the stochastic environment, but might not improve beyond $T^{4/5}$ regret (see Footnote~\ref{fn:discretization} below).
%On the other hand, regret-optimal algorithms for the stochastic environment are known for truthful auctions \citep{BalseiroGur19,Balseiro-BestOfMany-Opre}.

A static (single-shot) game between budget-constrained bidders who tune their pacing multipliers, a.k.a. the \emph{pacing game}, along with the appropriate equilibrium concept was studied in \citet{Conitzer-wine18,Conitzer-ec19}, in the context of first- and second-price auctions with quasilinear utilities, and then extended to more general payment constraints \citep{Gagan-wine19} and utility measures \citep{Babaioff-itcs21}.
%In that work, the authors establish the existence of equilibria as well as a number of welfare and (computational) complexity properties.
%These results were later extended by \citet{Conitzer-ec19} to first-price auctions, although the restriction to pacing strategies is not quite as well-motivated theoretically as in the second-price case.
%Equilibria efficiency guarantees for such games were established in \citet{Gagan-wine19} under more general payment constraints, as well as in  \citet{Babaioff-itcs21} for more general utility measures.
In particular, any pure Nash equilibrium of this game achieves at least half of the optimal liquid welfare when the underlying auction is truthful.
In a related contextual auction setting and simultaneously with our work, \citet{StandardAuctions-ec22} establishes a similar bound on liquid welfare at any (possibly non-linear) Bayes-Nash equilibrium for i.i.d.\ bidders, for a class of standard auctions that includes first-price and second-price auctions.
In contrast to these results, our efficiency result does not rely on convergence to equilibrium,
% applies to learning dynamics that do not necessarily converge to equilibrium,
and applies to all core auctions. %the richer class of core auctions.

%For a general survey of price of anarchy guarantees in auctions without budgets, we refer to \citet{DBLP:journals/jair/RoughgardenST17}.

%% Include work on non-stationary analysis

A growing line of work in mechanism design targets bidding agents that maximize value or utility
under spending constraints and are assumed to reach equilibrium \citep{pai2014optimal,Balseiro-ec21,deng2021towards}.
%\asdelete{The revenue-optimal mechanism for value-maximizing bidders subject to a budget constraint was found by \citet{pai2014optimal}. \citet{Balseiro-ec21} compare revenue-optimal mechanisms for different agents' objectives and constraints, noting in particular that whether first-best revenue for value-maximizing agents is achievable depends on whether the agent values and spend constraints are private or public information. \citet{deng2021towards} note that common auction heuristics such as boosting can increase efficiency in the presence of bidding agents.}
In contrast, our emphasis is not on mechanism design: we take the auction specification as exogenous and focus on the learning dynamics.
%\asdelete{of bidding agents and the resulting efficiency outcomes.}

%\ascomment{revised this para}
%Among the other approaches to repeated budgeted auctions studied by  is
\emph{Throttling} (a.k.a. \emph{probabilistic pacing}) is an alternative approach:
instead of pacing their bids, agents participate in only a fraction of the auctions \citep{balseiro2021budget}.
%another approach to repeated bidding under budget \citep{balseiro2021budget}.
%Rather than maintaining pacing multipliers, agents instead elect to not participate in some fraction of the auctions.
Very recently,  \citet{Chen-wine21}
%studies \emph{throttling games}: for  first- and second-price auctions.
proved that throttling converges to a Nash equilibrium in the first-price auctions (albeit without any stated implications on welfare or liquid welfare). On the other hand, no such convergence is possible for second-price auctions under throttling dynamics for the same reason as pacing: a Nash equilibrium is $\mathsf{PPAD}$-hard to compute \citep{Chen-ec21}.
%while equilibria in the first-price throttling games can be efficiently achieved via decentralized dynamics, it is $\mathsf{PPAD}$-hard to compute an equilibrium in second-price throttling games.
%\citet{Chen-wine21} also compare
The equilibria obtained by throttling and pacing dynamics in the first-price auction can differ in revenue by at most a factor of $2$ \citep{Chen-wine21}.
%the \jg{revenue} obtained in throttling equilibria and pacing equilibria differ by at most a factor of two.

\xhdr{Learning theory.}
Repeated bidding with a budget is a special case of multi-armed bandit problems with global constraints, a.k.a. \emph{bandits with knapsacks} (BwK)
\citep{BwK-focs13,AgrawalDevanur-ec14-OpRe,AdvBwK-focs19}
(see Chapter 10 in \citep{slivkins-MABbook} for a survey).
BwK problems in adversarial environments do not admit regret bounds: instead,
%One noteworthy insight from this work is that one should not hope to minimize regret in the adversarial environments:
one is doomed to approximation ratios, even against a time-invariant benchmark and even in relatively simple examples \citep{AdvBwK-focs19}. A similar impossibility result is derived in
\citep{BalseiroGur19} specifically for repeated budget-constrained bidding in second-price auctions. {The essential reason for this impossibility is the \emph{spend-or-save dilemma} \citep{AdvBwK-focs19}, whereby the algorithm does not know whether to spend the budget now or save it for the future.}

Several known algorithms for BwK are potentially applicable to our problem, and the corresponding individual guarantees on regret may be within reach for the stochastic environment, but haven't yet been published.%
\footnote{\label{fn:discretization}One could run a BwK algorithm on bids discretized as multiples of some $\epsilon>0$. Individual guarantees would depend on bounding the discretization error which is a known challenge in BwK \cite{BwK-focs13,cBwK-colt14}. We are only aware of one such result for BwK when one has contexts and a ``continuous" action set (in our case, resp., private values and bids). This result concerns a different problem: \emph{dynamic pricing}, where actions correspond to posted prices, and only achieves regret  $\tilde{O}(T^{4/5})$ in the stochastic environment \citep{cBwK-colt14}.}
However, it is unclear how to derive aggregate guarantees for such algorithms.  The algorithm analyzed in the present work is based upon stochastic gradient descent, which is a standard, well-understood algorithm in online convex optimization \citep{Hazan-OCO-book}.

%in budget-constrained bidding, as well as other budget-constrained online learning problems, one is doomed to approximation ratios even in relatively simple examples \citep{BalseiroGur19,AdvBwK-focs19}.

A long line of work targets convergence of learning algorithms in repeated games (not specifically focusing on ad auctions).
%\bjldelete{In this work, all algorithms have low regret in terms of their cumulative payoffs in the repeated game; in particular, there are no budget constraints w.r.t. these payoffs. Guarantees are the strongest for two-player zero-sum games (which ad auctions are definitely not), and for smooth games (which ad auctions are not known to be).}
%Beyond that, the game dynamics is only known to converge in a rather weak sense: in \emph{average play} (time-averaged distribution over chosen actions) and to a much wider class of equilibria: (coarse) correlated equilibria rather than Nash equilibria \citep{Aumann-74,Moulin-78,HartMasCollel-econometrica00}.
When algorithms achieve low regret (in terms of cumulative payoffs), the \emph{average play} (time-averaged distribution over chosen actions) converges to a (coarse) correlated equilibrium~\citep{Aumann-74,Moulin-78,HartMasCollel-econometrica00}, and this implies welfare bounds for various auction formats in the absence of budget constraints~\cite{DBLP:journals/jair/RoughgardenST17}.  In contrast, we show in \Cref{app:regret.example} that for repeated auctions with budgets, low individual regret on its own does not imply any bounded approximation for liquid welfare.
Convergence in the last iterate is more challenging: strong negative results are known even for two-player zero-sum games \citep{Piliouras-ec18,Piliouras-soda18,Piliouras-colt19}. Yet, a recent line of work (starting from \citet{Daskalakis-iclr18}, see \citep{Daskalakis-itcs19-lastIterate,Daskalakis-neurips20-lastIterate,Haipeng-iclr21-lastIterate} and references therein) achieves last-iterate convergence under full feedback and substantial convexity-like assumptions, using two specific regret-minimizing algorithms.
To the best of our understanding, these positive results do not apply to the setting of repeated auctions with budgets.

\xhdr{Subsequent work.} Since the preliminary conference version of this work was made publicly available~\cite{Gaitonde-itcs23}, %this work on {\tt arxiv.org},
several relevant papers have appeared in the literature.

{The follow-up papers \citet{Fikioris-LPoA23,Autobidding-colt24} are directly related, combining liquid-welfare and individual guarantees.} \citet{Fikioris-LPoA23} focus on the special case of repeated first-price auctions and achieve liquid welfare guarantees as long as each bidding algorithm satisfies a particular individual guarantee. Specifically, they assume multiplicative $\gamma\geq 1$ approximation relative to the best fixed pacing multiplier, and obtain multiplicative approximation ratio $R = \gamma+O(1)$ on liquid welfare. They achieve $\gamma=T/B$ and $R \approx T/B+1/2$ plugging in a recent result on bandits with knapsacks \citep{Castiglioni-icml22}. Further, they achieve $R \approx 2.4$ if $\gamma=1$. However, it is currently not known how to achieve $\gamma<T/B$, let alone $\gamma=1$, with non-trivial budget-constrained bidding algorithms, e.g., such as those guaranteed to achieve vanishing regret in a stochastic environment.
%Their guarantees apply even when the value profile $\boldsym{v}_t$ changes over time.

\citet{Autobidding-colt24} extend our results to bidders that face return-on-investment constraints in addition to budget constraints. They achieve the same 2-approximation guarantee on liquid welfare and similar individual guarantees. Their algorithm coincides with ours when specialized to budget constraints. Their results hold for single-item allocation problems and any auction format in which the single item is sold to the highest bidder and the payment lies between the highest and second-highest bids.

Two other papers focus on individual guarantees for online bidding under budget, obtaining $\tilde{O}(\sqrt{T})$ regret under various assumptions. \citet{wang2023learning} consider utility maximization in a repeated first-price auction, provided that the maximum bid is revealed when an auction is lost. \citet{balseiro2023robust} consider repeated second-price auction, but in a non-stationary environment: they use time-varying target expenditures driven by additional data samples.

%\bjledit{Under a one-sided feedback model where the maximum bid is revealed when an auction is lost,~\citet{wang2023learning} establish regret guarantees of $\tilde{O}(\sqrt{T})$ for utility maximization in a repeated first-price auction subject to a budget constraint and stationary environment.~\citet{jiang2025online} and~\citet{balseiro2023robust} study the problem of achieving $O(\sqrt{T})$ regret in repeated second-price auctions with budget constraints, but in a non-stationary environment: they us+e time-varying target expenditures driven by data samples and bound the amount of data required to achieve the optimal $O(\sqrt{T})$ regret rate.  In contrast, our individual regret analysis handles non-stationarity by using path length to quantify the degree of non-stationarity and incorporating this metric into the regret guarantee.}

%\ascomment{Brendan: \citet{jiang2025online} is not really relevant: it assumes that the reward and consumption functions are revealed to the DM before the DM chooses an action. No need to cite, I think!}

Finally, the idea to measure regret against the perfect pacing sequence, introduced in this paper to bypass the spend-or-save dilemma in the adversarial environment, has been fruitfully used (for the same purpose) in general BwK problems \citep{LagCBwK-jmlr24,BwK-braverman2025}.

%\yledit{\xhdr{Welfare Maximization.} In this paper we have considered the objective of liquid welfare maximization. An alternative benchmark is to directly consider the welfare maximization problem.
%In the static setting, when the agents are symmetric, the all-pay auction is the Bayesian optimal mechanism \citep{pai2014optimal}. Moreover, the implementation of the all-pay auction does not require seller's knowledge about the valuation distributions.
%In the dynamic environment, even if payments are not allowed, \citet{JS-07,BJK-22}
%show that by imposing quotas on the reported values over the whole dynamic process,
%the seller can obtain no-regret compared to the optimal welfare in equilibrium.
%The design of optimal quotas rely on the distribution over valuations for each agents.
%It is an interesting open question for designing dynamic mechanisms that approximate the optimal welfare when the seller is ignorant of the valuation distributions.
%} 
\section{Our Model and Preliminaries}
\label{sec:model}

%\subsection{The Allocation Problem}
%\label{sec:model.allocation}

\vspace{-2mm}\xhdr{The Allocation Problem.}
Our setting is a repeated auction with budgets.
%In this section we describe the main components of the model and discuss preliminary definitions and results that are needed in order to prove our welfare and regret results.
%
%
%
%\subsection{Repeated Auctions with Budgets}
%In this section we describe the repeated auction setting that we study.
%
There is one seller (the platform) and $n$ bidding agents. In each time $t=1,\ldots,T$, the seller has a single unit of a good available for sale.  We will sometimes refer to the good as a (divisible) \emph{item}.  An \textit{allocation profile} is a vector $\boldsym{x} = (x_{1}, \dotsc, x_{n}) \in [0,1]^n$ where $x_{k}$ is the quantity of the good allocated to agent $k$. There is a convex and closed set $X \subseteq [0,1]^n$ of feasible allocation profiles, which is assumed to be downward closed.\footnote{The set $X \subseteq [0,1]^n$ is downward closed if, for any $\boldsym{x}, \boldsym{x}' \in [0,1]^n$ such that $\boldsym{x} \in X$ and $x'_k \leq x_k$ for all $k$, we have $\boldsym{x}' \in X$.} An allocation profile is feasible if $\boldsym{x} \in X$.  An allocation sequence is a sequence of allocation profiles $(\boldsym{x}_{1},\ldots,\boldsym{x}_{T})$ where $\boldsym{x_{t}} = (x_{1t}, \dotsc, x_{nt})$ is the allocation profile at time $t$.

At each time $t$, each agent $k$ has a value $v_{k,t} \in [0, \overline{v}]$ per unit of the item received. By scaling values, we will assume that $\overline{v}\geq 1$.  The value profile $\boldsym{v}_t = (v_{1,t}, \dotsc, v_{n,t})$ at time $t$ is drawn from a distribution function $F$ independently across different time periods.  We emphasize that $F$ is not necessarily a product distribution, so the values held by different agents can be arbitrarily correlated within each round.  %We will sometimes write $F_i$ for the marginal distribution of agent $i$'s value.

%Let us now pause to describe some special cases of our model that are of particular interest in the context of advertising auctions.

Two special cases of our model are of particular interest in the context of ad auctions. In a \emph{single-slot} ad auction, the ``item" for sale is an opportunity to display one ad to the current user. A \emph{Multi-slot} ad auction is a natural generalization in which multiple ads can be displayed in each round. The formal details are standard, we provide them in \Cref{app:single-multiple-slot}.

\xhdr{Auctions and Budgets.}
At each time $t$ the good is allocated using an auction mechanism that we now describe.  In round $t$, each agent $k \in [n]$ first observes her value $v_{k,t} \geq 0$.  After all agents have observed their values, each agent $k$ then submits a bid $b_{k,t} \geq 0$ to the auction.  All agents submit bids simultaneously.  The auction is defined by an allocation rule $\boldsym{x}$ and a payment rule $\boldsym{p}$, where $\boldsym{x}(\boldsym{b}) \in X$ is the allocation profile generated under a bid profile $\boldsym{b}$, and $p_k(\boldsym{b}) \geq 0$ is the payment made by agent $k$.
All of the auction formats we consider will have weakly monotone allocation and payment rules, so we will assume this for the remainder of the paper.  That is, for any $k$ and any bids of the other agents $\boldsym{b}_{-k}$, both $x_k(b_k, \boldsym{b}_{-k})$ and $p_k(b_k, \boldsym{b}_{-k})$ are weakly increasing in $b_k$.

Each agent $k$ has a fixed budget $B_k$ that can be spent over the $T$ time periods.  Once an agent has run out of budget she can no longer bid in future rounds.  We emphasize that this budget constraint binds ex post, and must be satisfied on every realization of value sequences.\footnote{All of the auctions we consider satisfy the property that each agent's payment will be at most her bid, so it is always possible to bid in such a way that does not overspend one's remaining budget.}  An important quantity in our analysis is the \emph{target expenditure rate} for agent $k$ which is defined by $\rho_k\triangleq B_k/T$.

%\asmargincomment{we never defined "history"!}
The \emph{auction history} up to round $t$ consists of all realized values, bids, allocations, and payments for all agents and all rounds prior to $t$. An agent $k$ will observe some (but typically not all) of this history. All of our results hold under \emph{bandit feedback} where each agent can see her own values, bids, allocations, and payments from previous rounds, but is not required to observe any other aspect of the history including the bids of other agents.

Each agent $k$ applies a dynamic bidding strategy that, in each round $t$, maps  (the observed part of) auction history and the realized valuation $v_{k,t}$ to a bid $b_{k,t}$.  In a slight abuse of notation we'll write $\boldsym{b}(\boldsym{v})$ for the sequence of bid profiles that are generated when the sequence of value profiles is $\boldsym{v}$.  In this paper we will focus on a particular class of bidding strategies (gradient-based pacing) described
%in Section~\ref{sec:pacing}.
later in this section.

Given an execution of the auction over $T$ rounds, we will typically write $x_{k,t}$ for the realized allocation obtained by agent $k$ in time $t$, and $z_{k,t}$ for the realized payment of agent $k$ in time $t$.  This notation omits the dependency on the agents' bids when this dependency is clear from the context.

%\subsection{Core Auctions}
%\label{sec:model.core}

\xhdr{Core Auctions.}
Our results apply to \emph{Core Auctions} \citep{ausubel2002ascending}, a wide class of auction mechanisms which includes standard auctions such as first-price auctions, second-price auctions for single-slot allocations, and generalized second-price (GSP) auctions for multi-slot allocations.%
\footnote{For the sake of completeness, these special cases are defined in \Cref{app:CoreAuctions}.}
%Recall that the allocation and payments in each round are determined by an auction mechanism.  Our analysis will apply to a class of auction formats known as \emph{Core Auctions}.
%Viewing an auction instance as a game played by the seller and the $n$ agents, a core auction is a mechanism that generates outcomes that lie in the core of the game.
Roughly speaking, no subset of players (which may or may not include the seller) could jointly benefit by renegotiating the auction outcome among themselves.
%\citep{ausubel2002ascending}.
Stated more formally in our context, a core auction satisfies the following two properties:
\begin{enumerate}
\item {The auction is \emph{individually rational} (IR):} each agent's payment does not exceed her declared welfare for the allocation received.  That is,
            $p_k(\boldsym{b}) \leq b_k\; x_k(\boldsym{b})$
    for every agent $k$ and every bidding profile $\boldsym{b}$.

    \item The seller and any subset of agents $S \subseteq [n]$ could not strictly benefit by jointly abandoning the auction and deviating to another outcome. That is, for any bidding profile $\boldsym{b}$ and any allocation profile $\boldsym{y} \in X$,
\[
    \textstyle \sum_{k}\; p_k(\boldsym{b})
    + \sum_{k \in S} \left(b_k x_k(\boldsym{b}) - p_k(\boldsym{b}) \right)
    \geq \sum_{k \in S}\; b_k y_k
\]
which simplifies to
    \begin{align}\label{eq:model-core}
    \textstyle \sum_{k \not\in S}\; p_k(\boldsym{b})
    + \sum_{k \in S}\; b_k x_k(\boldsym{b})
    \geq \sum_{k \in S}\; b_k y_k.
    \end{align}
    The left-hand (resp., right-hand) side of \refeq{eq:model-core} is the total welfare obtained by the seller and the agents in $S$ under the auction (resp., under a deviation to allocation $\boldsym{y}$). On both sides,
    the summation over agents in $S$ accounts not only for the agent utilities, but also those agents' payments to the seller.
    %payments from $S$ to the seller cancel with the utility loss from the deviating agents, so that only the left side contains the welfare of the seller via payments from non-deviating bidders.

     % \newbl{I am slightly confused with this definition. The definition is over the auction (allocation,payment profile) -- is the bidding profile exogenous here? so this should be true for every bidding profile? Or the bids must be constrained somehow?
%    }\bjl{True for all bidding profiles.  Agree this should be more clear.}
\end{enumerate}

The second property implies that a core auction must always maximize declared welfare.  Indeed, taking $S = [n]$, \refeq{eq:model-core} states that $\sum_k b_k x_k(\boldsym{b}) \geq \sum_k b_k y_k$ for all feasible allocation profiles $\boldsym{y}$.  The first property is simply a restatement of the core condition that no subset of buyers could jointly benefit by renegotiating an auction outcome that does not include the seller, i.e. no buyer (and hence no set of buyers) would strictly prefer to switch to the null outcome in which no goods are allocated and no payments are made.
%\BLcomment{Did you mean the second?}\bjl{No, the first.  The payments being at most the welfare is the same as ``no bidder wants to deviate to getting nothing,'' which is equivalent to having a deviating coalition that doesn't include the auctioneer.  We should probably go through all this in more detail.}
%can be interpreted as saying that no subset of the agents can strictly benefit by jointly deviating to a different outcome that does \emph{not} include the seller (which must be the outcome where no goods are obtained).

%\ascomment{Move the examples below to an appendix?}

%For example, when selling a single indivisible good, the first-price auction and second-price auction are both core auctions (as is any auction where the highest bid wins and the winner pays some amount between the highest and second-highest bid).  Sponsored search auctions can be modeled as allocating a divisible good subject to a polymatroid constraint \citep{}, and the generalized second price auction is a core auction \citep{}. \BLcomment{citations to be inserted }

\xhdr{Monotone Bang-Per-Buck.}
%\ascomment{to be revised once we converge on the new version of this property}
Our guarantees on individual regret require an additional property of the auction mechanism called  \emph{monotone bang-per-buck} (MBB).
Roughly speaking, if increasing one's bid leads to an increased allocation, then the increase in payment per unit of new allocation is at least the minimum bid needed for the increase.
%Roughly, each buyer's expected payment per expected unit of allocation received is weakly increasing in their bid.
%amount of allocation obtained.
Formally, an auction mechanism satisfies MBB if for any agent $k$, any profile of bids $\boldsym{b}_{-k}$ of the other agents, and any bids $b_k \leq b'_k$ of agent $k$, we have
\begin{align}
\label{eq.mbb.1}
 p_k(b'_k, \boldsym{b}_{-k}) - p_k(b_k, \boldsym{b}_{-k}) 
    \geq b_k \cdot \rbr{x_k(b'_k, \boldsym{b}_{-k}) - x_k(b_k, \boldsym{b}_{-k})}.
\end{align}
One can readily check that the MBB property is satisfied by many common auction formats, such as first-price and second-price auctions as well as generalized second-price auctions (see \Cref{app:ExamplesofAuctions} for further discussion).
The name MBB is motivated by a simple observation that in a MBB auction each buyer's expected payment per expected unit of allocation received is weakly increasing in their bid, for any distribution of competing bids.  See, for example, inequality \eqref{eq:mbb} in Section~\ref{sec:proof.regmain}.

\xhdr{Agent Objective.}
%In order to evaluate bidding strategies,
{To state individual regret guarantees}, we must define an objective for each individual agent.  Two natural, well-motivated candidates are (constrained) value maximization and utility maximization.
We focus on utility maximization as a primary scenario in this paper, but also discuss value maximization below. Our \emph{aggregate} guarantees are independent of the agent objective.

%{We focus on the value maximization in this paper (which is increasingly standard in online advertising markets with automated bidding algorithms \citep{balseiro2021landscape}), but also discuss utility maximization below. Our \emph{aggregate} guarantees are independent of the agent objective.}

A utility-maximizing agent aims to maximize total quasi-linear utility,
$\sum_t v_{k,t} x_{k,t} - z_{k,t}$, subject to the budget constraint $\sum_t z_{k,t} \leq B_k$. Implicitly, this formulation assumes that values are expressible in units of money. This objective is especially relevant in scenarios where the agents are the advertisers themselves.

For any auction that is truthful for quasi-linear bidders {within a single round}, such as the second-price auction, the ex-post optimal strategy for a utility-maximizing agent is to bid $v_{k,t} \times \alpha$ in each round for the largest $\alpha \in [0,1]$ such that the budget constraint is satisfied~\cite{BalseiroGur19}.  Notably, this strategy also maximizes the total value received subject to the budget constraint and a no-overbidding condition.
%So, our individual guarantees apply also to the objective of utility maximization whenever the underlying {single-round} auction mechanism is truthful.

Motivated by this observation, we also consider value-maximizing agents whose objective
%A value-maximizing agent's objective
is to maximize the total value received subject to budget and maximum bid constraints.
More specifically, a value-maximizing agent in our model aims to maximize the total value of the allocation obtained over all rounds, $\sum_t v_{k,t} x_{k,t}$, subject to the following two constraints.  First, the total payment across all rounds must not exceed the budget: $\sum_t z_{k,t} \leq B_k$ for all $k$.
%Second, the payment in each round $t$ must not exceed the value received in round $t$: $z_{k,t} \leq v_{k,t} x_{k,t}$ for all $k$ and $t$.
Second, the bid placed in round $t$ must not exceed the value of the good in round $t$: $b_{k,t} \leq v_{k,t}$ for all $k$ and $t$.
The latter can be interpreted as a no-overbidding constraint, consistent with a pacing scenario where the agent is empowered to scale bids downward but not upward.  Alternatively, for all the auction formats we consider, the second constraint could equivalently be interpreted as a bound on marginal ROI: i.e., that $z_{k,t} \leq v_{k,t} x_{k,t}$ for all $k$ and $t$, in which case $v_{k,t}$ is viewed as an upper bound on the maximum amount the agent is allowed to pay per unit in any round.\footnote{We emphasize that a bound on marginal ROI is not equivalent to an average ROI constraint that binds over multiple rounds.  Maximizing value subject to an average ROI constraint is a common but different model of algorithmic bidding.  See~\citet{aggarwal2024auto} for a survey and~\citet{Autobidding-colt24} for follow-up work extending our analysis to scenarios with average ROI constraints.}

\xhdr{Liquid Welfare.}
Liquid welfare is a measure of welfare that accounts for non-quasi-linear agent utilities.
%is amenable to settings where agents have non-quasi-linear utilities.
Intuitively, an agent's liquid welfare for an allocation sequence is the agent's \emph{maximum willingness to pay for the allocation}.  This generalizes the common notion of welfare in quasi-linear environments, and motivates the choice to measure welfare in transferable units of money.  In our setting with a budget constraint that binds across rounds, liquid welfare is defined as follows.

\begin{definition}
 Given a sequence of value profiles $\boldsym{v} = (v_{k,t}) \in [0,\overline{v}]^{nT}$ and any sequence of feasible allocations $\boldsym{x}=(x_{k,t}) \in X^T$, the \emph{liquid value} obtained by agent $k$ is
%\begin{equation*}
%\textstyle
     $W_{k}(\boldsym{x})
     =\min\cbr{B_k,\,\sum_{t=1}^T\; x_{k,t}\,v_{k,t}}.$
 %\end{equation*}
 The \emph{liquid welfare} of allocation sequence $\boldsym{x} \in X^T$ is $W(\boldsym{x})=\sum_{k=1}^n W_{k}(\boldsym{x}).$
 \end{definition}

We emphasize that liquid welfare
% is a property of the allocations, and does not depend on the payments made by the agents.
depends on the allocations, but not on the agents' payments.

%Since the agents' valuations in our setting are stochastic, the allocation sequence obtained by the agents will be random.

Our objective of interest is the \emph{expected} liquid welfare obtained by the platform over any randomness in the valuation sequence and the agents' bidding strategies (which induces randomness in the allocation sequence). Since the bid placed in one round can depend on allocations obtained in the previous rounds, we define
%it will be helpful to describe
a mapping from the entire sequence of $T$ value profiles to an allocation sequence.  An \emph{allocation sequence rule} is a function $\boldsym{x} \colon [0,\overline{v}]^{nT} \to X^T$, where $x_{k,t}(\boldsym{v}_1, \dotsc, \boldsym{v}_T)$ is the allocation obtained by agent $k$ in round $t$.  Then the expected liquid welfare of allocation sequence rule $\boldsym{x}$ is
%\asmargincomment{made it an eqn and added notation}
\begin{align}\label{eq:LW-bench}
W(\boldsym{x},F) :=
\E_{\boldsym{v}_1\LDOTS\,\boldsym{v}_T \sim F}
    \sbr{W(\boldsym{x}(\boldsym{v}_1\LDOTS\boldsym{v}_T)}.
 \end{align}
%$\mathbb{E}_{\boldsym{v}_1, \dotsc, \boldsym{v}_T \sim F}[W(\boldsym{x}(\boldsym{v}_1, \dotsc, \boldsym{v}_T)]$.

\xhdr{Pacing Algorithm.}
We use a budget-pacing algorithm motivated by stochastic gradient descent that was introduced and analyzed in \citep{BalseiroGur19}
%Our formulation is heavily motivated by \citep{BalseiroGur19} who introduce and analyze this same pacing strategy
in the context of utility maximization in second-price auctions. See~\Cref{alg:bg}.  Each bidder $k$ maintains a \emph{pacing multiplier} $\mu_{k,t} \in [0, \overline{\mu}]$ for each round $t$.  Multiplier $\mu_{k,t}$ is determined by the algorithm before the value $v_{k,t}$ is revealed. The bid is set to $v_{k,t} / (1+\mu_{k,t})$, or the remaining budget $B_{k,t}$ if the latter is smaller. Once the round's outcome is revealed, the multiplier is updated as per \refeq{eq:alg-update}, where $P_{[a,b]}$ denotes the projection onto the interval $[a,b]$. Intuitively, the agent's goal is to keep expenditures near the expenditure rate $\rho_{k}$. Hence, if $\rho_{k}$ is above (resp., below) the current expenditure $z_{k,t}$, the agent decreases (resp., increases) her multiplier, by the amount proportional to the current ``deviation" from the target expenditure rate $\rho_k$.

\begin{algorithm}
\caption{Gradient-based pacing algorithm for agent $k$}\label{alg:bg}
\KwIn{Budget $B_k$, time horizon $T$, {step-size} $\epsilon_k>0$ ,
pacing upper bound $\overline{\mu}$}

Initialize\newline $\mu_{k,1} = 0$ (pacing multiplier),
    $B_{k,1} = B_k$ (remaining budget),
    $\rho_k = B_k/T$ (target spend rate).

\For{round $t=1,\ldots,T$}{

Observe value $v_{k,t}$,
submit bid
    $b_{k,t} = \min\cbr{ v_{k,t}/(1+\mu_{k,t}),\, B_{k,t} }$;

Observe expenditure $z_{k,t}$;

Update $B_{k,t+1}\leftarrow B_{k,t}-z_{k,t}$ and
\begin{align}\label{eq:alg-update}
\mu_{k,t+1} \leftarrow P_{[0,\overline{\mu}]}\rbr{\mu_{k,t}-\epsilon_k(\rho_k-z_{k,t})}.
\end{align}
\label{alg:bg:line-update}

\vspace{-5mm}{STOP if $B_{k,t}\leq 0$.}
}
\end{algorithm}

Parameters $\overline{\mu}$ and $\epsilon_k$ will be specified later. While the upper bound $\overline{\mu}$ could in general depend on the agent~$k$, as in
    $\overline{\mu}=\overline{\mu}_k$,
we suppress this dependence for the sake of clarity. Our bounds depend on
    $\max_k \overline{\mu}_k$.

A convenient property of \Cref{alg:bg} is that it does not run out of budget too early.%
\footnote{This result is proven in \citep{BalseiroGur19} (Eq. (A-4) of Theorem 3.3) for second-price auctions, but an inspection of their proof shows this holds surely for any sequence of valuations and competing prices so long as the price is at most the bid if the agent wins.}
%It was known  for second-price auctions \cite{BalseiroGur19}, but the proof extends immediately to our setting.}

\begin{lemma}[\citep{BalseiroGur19}]
\label[lemma]{lem:bg1}
Fix agent $k$ in a core auction. Let $\tau_k$ be the stopping time of \Cref{alg:bg} for some (possibly adaptive, randomized, adversarially generated) set of valuations and competing bids. Assume all valuations are at most $\overline{v}$, and the parameters satisfy
    $\overline{\mu}\geq \overline{v}/\rho_k-1$
and
    $\epsilon_k\, \overline{v}\leq 1$.
Then $
    T-\tau_k\leq \frac{\overline{\mu}}{\epsilon_k \rho_k}+\frac{\overline{v}}{\rho_k}$
almost surely. %(\BLnew{ $\epsilon$ and $\epsilon_{k} $ are different? also $\rho$? } \bjl{I think $\rho$ and $\epsilon$ should be $\rho_k$ and $\epsilon_k$.  Fixed above.}
%and moreover, the pacing multiplier never exceeds $\overline{\mu}$.
\end{lemma}

\begin{remark}
\Cref{alg:bg} never \emph{overbids}, in the sense that the bids $b_{k,t}$ are always upper-bounded by the respective values $v_{k,t}$.  For utility-maximizing agents, this behavior is provably optimal for second-price auctions \citep{BalseiroGur19}.  More generally, the choice to never overbid may reduce the agent's objective value.  In such cases, we interpret no-overbidding as an additional optimization constraint motivated by the budget pacing context and/or an imposed bound on marginal ROI, as discussed in the Agent Objective paragraph above.
%constraint is directly inherited from the interpretation of $v_{k,t}$ as a constraint on the maximum allowable bid {(see Footnote~\ref{fn:interpretation})}.
% but could be suboptimal for non-truthful auctions such as first-price auctions.
\end{remark}

\section{Aggregate Guarantee: Approximate Liquid Welfare}
\label{sec:liquid.welfare}

{Our main result is a liquid welfare guarantee without convergence.} We show that, when all agents use \Cref{alg:bg}, the expected liquid welfare (as defined in \refeq{eq:LW-bench}) is at least half of the optimal minus a regret term.
%\ascomment{Added the thm statement below, edits in blue compared to \Cref{thm:main.new}}

\begin{theorem}\label{thm:main.simple}
Fix any core auction and any distribution $F$ over agent value profiles. Suppose that each agent $k$ employs {\Cref{alg:bg}} to bid, possibly with a different step-size $\eps_k$. Write
$\boldsym{x}$ for the corresponding allocation sequence rule.
Then for any {other allocation sequence rule $\boldsym{y}$} we have
\begin{equation}\label{eq:wel bound}
 {W(\boldsym{x},F)}
    \geq \frac{W(\boldsym{y},F)}{2} - O\rbr{n\overline{v}\sqrt{T\log(\overline{v}nT)}}.
\end{equation}
\end{theorem}

{In fact, our analysis for \Cref{thm:main.simple} yields a stronger result. First, it does not invoke a full specification of the bidding rule in \Cref{alg:bg}. Instead,} it applies to a wider class of algorithms that do not overbid ($b_{k,t} \leq v_{k,t}$), bid their full value when $\mu_{k,t}=0$, and are not constrained otherwise.

\begin{definition}\label[definition]{def:gen-pacing}
Consider a measurable bidding algorithm that inputs the same parameters as \Cref{alg:bg}, internally updates  multipliers $\mu_{k,t}$ in the same way, and never overbids: $b_{k,t} \leq v_{k,t}$ for all rounds $t$. Call it a \textbf{generalized pacing algorithm} if
    $(\mu_{k,t} = 0) \Rightarrow (b_{k,t} = v_{k,t})$
for all rounds $t$ (\emph{``no unnecessary pacing"}).
\end{definition}

Second, our analysis competes against a stronger benchmark: optimal \emph{ex ante} liquid welfare, which is each agent's willingness to pay for her \emph{expected} allocation sequence. It upper-bounds the expected (ex post) liquid welfare by Jensen's inequality.

\begin{definition}
Fix distribution $F$ over valuation profiles and allocation rule
    $\boldsym{y} \colon [0,\overline{v}]^{n}\to X$.
% and under \Cref{assumption:a1}
The \emph{ex ante liquid value} of each agent $k$ is
%\begin{equation*}
    $\overline{W}_k(\boldsym{y},F)=T \times
    \min\sbr{\rho_k, \mathbb{E}_{\boldsym{v}\sim F}\sbr{y_{k}(\boldsym{v})\, v_{k}}}$
%\end{equation*}
% \ykl{Should this be $W_k(\boldsym{y})=T\cdot \min\left\{\rho_k, \mathbb{E}_{\boldsym{v}\sim F}\left[v_{k,t}\cdot \boldsym{y}_k(\boldsym{v})\right]\right\}$?
% If the allocation cannot depends on the whole valuation profile, this may not be an upper bound on the ex post liquid welfare. }
%
and the ex ante liquid welfare is
$\overline{W}(\boldsym{y},F)
    =\textstyle \sum_{k=1}^n\; \overline{W}_k(\boldsym{y},F)$.
%\begin{equation*}
%    \overline{W}(\boldsym{y},F)
%    =\textstyle \sum_{k=1}^n\; \overline{W}_k(\boldsym{y},F).
%\end{equation*}
\end{definition}

%Notice that our definition of ex ante liquid welfare assumes that
In this definition, the restriction that the same allocation rule $y$ is used in every round is without loss of generality.
%for ex ante liquid welfare:
Indeed, given any allocation sequence rule
%(which might not apply the same allocation rule in every round),
there is a single-round allocation rule with the same ex ante liquid welfare. This property is formally stated and proved as \Cref{lem:single.round} in \Cref{app:proofs}.

{To sum up, the following is the generalization of \Cref{thm:main.simple} that we actually prove:}

%Now we are ready to state our main result: a liquid welfare guarantee without convergence.

%Now that we have defined the notion of ex ante liquid welfare, we are ready to state our main result: the expected liquid welfare when all the agents use generalized pacing algorithms is at least half of the optimal ex ante liquid welfare, less an additive regret term that \jg{is sublinear in} $T$.

%\bjl{Changed this statement so we don't use $WEL_{GPD}$, since the notation was hiding some information about the order of quantification.  Instead, a lot of the simplifying notation is now introduced in the proof itself.}

\begin{theorem}\label{thm:main.new}
Fix any core auction and any distribution $F$ over agent value profiles. Suppose that each agent $k$ employs a generalized pacing algorithm to bid, possibly with a different step-size $\eps_k$. Write
$\boldsym{x} \colon [0,\overline{v}]^{nT} \to X^T$ for the corresponding allocation sequence rule.
Then for any allocation rule $\boldsym{y} \colon [0,\overline{v}]^{n} \to X$ we have
\begin{equation}\label{eq:wel bound}
 %   \E_{\boldsym{v}_1, \dotsc, \boldsym{v}_T \sim F}\sbr{W(\boldsym{x}(\boldsym{v}_1, \dotsc, \boldsym{v}_T))}
 {W(\boldsym{x},F)}
    \geq \frac{\overline{W}(\boldsym{y},F)}{2} - O\rbr{n\overline{v}\sqrt{T\log(\overline{v}nT)}}.
\end{equation}
\end{theorem}

\begin{remark}
In the proof, it is not necessary for each advertiser $k$ to employ a fixed step-size $\eps_k$ throughout the entire time horizon. {Instead, $\eps_k$ can change whenever the corresponding pacing multiplier becomes $0$, and needs to stay fixed till the next time this happens.}
%it suffices if $\eps_k$ is constant within each epoch.
%%% AS: "epoch" not yet defined!
% This is needed for \Cref{lem:stronger}.
%%% AS: we don't mention \Cref{lem:stronger} yet!
\end{remark}

{The rest of this section proves Theorem~\ref{thm:main.new}. We start with a proof sketch to provide intuition, then continue with a formal proof.}

\subsection{Proof Sketch and Intuition}
\label{sec:lw.intuition}

%\bjl{All content in this section is new}

%We prove Theorem~\ref{thm:main.new} formally in Appendix~\ref{app:proofs}; here we will
%Before proving Theorem~\ref{thm:main.new}, let us sketch the proof and provide intuition.
One easy observation is that since $B_k$ is an upper bound on the liquid welfare obtainable by agent $k$, any agents who are (approximately) exhausting their budgets in our dynamics are achieving optimal liquid welfare.  We therefore focus on agents who do not exhaust their budgets.  We'd like to argue that such agents are often bidding very high, frequently choosing pacing multipliers equal to $0$ (i.e., bidding their values).  This would be helpful because our auction is assumed to be a core auction, which implies that either the high-bidding agents are winning (and generating high liquid welfare) or other, budget-exhausting agents are generating high revenue for the seller (which likewise implies high liquid welfare).

Why should agents who are underspending their budget be placing high bids in many rounds?  While it's true that the pacing dynamics increases the next-round bid whenever spend is below the per-round target, this is only a local adjustment and does not depend on total spend.  We make no convergence assumptions about the dynamics and, as we show in \Cref{app:regret.example}, regret bounds do not directly imply a bound on liquid welfare.  So how do we analyze bidding patterns in aggregate across rounds?
%It is here that we make use of a mechanical fact about generalized pacing algorithms.
%When analyzing the pacing dynamics to prove Theorem \ref{thm:main.new}, it will be helpful to consider separately those rounds in which the pacing multiplier is equal to $0$ and rounds in which the pacing multiplier is strictly greater than $0$.
%To this end

It is here where we use the fact that agents use generalized pacing algorithms.  For each $k$, the evolution of the multipliers $\mu_{k,t}$ has the following convenient property: over any contiguous range of time steps $[t_1, t_2]$ such that the multiplier is never $0$ or $\overline{\mu}$, the total amount spent from time $t_1$ to time $t_2$ is determined by $\mu_{k,t_2} - \mu_{k,t_1}$.  More precisely,
\begin{equation}
    \label{eq:payment.dynamics}
    \sum_{t=t_1}^{t_2-1}z_{k,t} = (t_2-t_1-1)\rho_k + \frac{1}{\eps_k}(\mu_{k,t_2} - \mu_{k,t_1})
\end{equation}
which follows immediately from the update rule of $\mu_{k,t}$ on line 5 of Algorithm~\ref{alg:bg}: if the multipliers are never $0$ or $\overline{\mu}$ then the projection operator is never invoked, so we have $\mu_{k,t+1} = \mu_{k,t} - \epsilon_k(\rho_k - z_{k,t})$ and \eqref{eq:payment.dynamics} follows.

Motivated by this observation, we introduce the notion of an \emph{epoch}: essentially,
%which roughly corresponds to
 a maximal contiguous sequence of rounds in which an agent's pacing multiplier is strictly greater than $0$.
%
%\begin{definition}
%\label{defn:epoch}
%Fix an agent $k$, time horizon $T$, and a sequence of pacing multipliers $\mu_{k,1},\ldots,\mu_{k,T}$. A half-open interval $[t_1,t_2)$ is an \textbf{epoch} with respect to these multipliers if it holds that $\mu_{t_1}=0$ and $\mu_{t}>0$ for each $t_1<t<t_2$ and $t_2$ is maximal with this property.
%%An epoch $[t_1,t_2)$ is \textbf{maximal} if $[t_1,t_2+1)$ is not an epoch.  %\bjl{Do we ever discuss non-maximal epochs?  If not, I'd recommend redefining epoch to mean ``maximal epoch."}
%\end{definition}
%
{In Lemma~\ref{lem:stronger} we show that an agent $k$'s total spend over an epoch of length $t$ must be (approximately) $t \rho_k$.  In other words, the average spend over an epoch approximately matches the target per-round spend.  This is a direct implication of the update rule for the pacing multiplier: since multiplier increases and decreases balance out (approximately) over the course of an epoch, the budget deficits and surpluses must balance out as well.}
%\bjl{See note below about maximal vs non-maximal epochs. Revisit the last sentence if we keep both.}
An immediate implication is that an agent whose total spend is much less than $T \rho_k$ must often have her pacing multiplier set equal to $0$.

\iffalse
\begin{lemma}
\label{lem:stronger}
Fix an agent $k$, fix any choice of core auction and any sequence of bids of other agents $(\boldsym{b}_{-k,1}, \dotsc, \boldsym{b}_{-k,T})$.  Fix any realization of values $(v_{k,1}, \dotsc, v_{k,T})$ for agent $k$ and suppose that $\mu_{k,1},\ldots,\mu_{k,T}$ is the sequence of multipliers generated by the generalized pacing algorithm given the auction format and the other bids.  Then for any epoch $[t_1,t_2)$ where $\mu_{k,t_2} \neq \flat$, we have
%is an epoch as given in \Cref{defn:epoch}.
%with parameters satisfying \Cref{lem:bg1} (\BLcomment{Do you use it?.. in particular NO2 or NO is needed? important also for the statement of the welfare result} on a completely arbitrary sequence of values and competing bids and $\mu_{k,T}\neq \flat$. Suppose further that $[t_1,t_2)$ is an epoch as given in \Cref{defn:epoch}. Then the following inequality holds:
\begin{equation}
\label{eq:val}
    \sum_{t=t_1}^{t_2-1} x_{k,t}v_{k,t}\geq x_{k, t_1}v_{k,t_1} - z_{k,t_1}+\rho_k \cdot (t_2-t_1-1).
\end{equation}
\end{lemma}
\fi

To summarize, each agent either spends most of her budget by time $T$ or spends many rounds bidding her value.  It might seem that by combining these two cases we should obtain a constant approximation factor, and not just in expectation but for every realized value sequence $\boldsym{v}$.  But this is too good to be true.  Indeed, this proof sketch misses an important subtlety: whether an agent exhausts her budget or not depends on the realization of the value sequence, which is also correlated with the benchmark allocation $\boldsym{y}$.  For example, what if an agent under-spends her budget precisely on those value sequences where it would have been optimal (according to the benchmark $\boldsym{y}$) for her liquid welfare to equal her budget?  This could result in a situation where, conditional on $\mu_{k,t} = 0$, the value obtained by the benchmark is much higher than expected and cannot be approximated.  To compare against the liquid welfare of $\boldsym{y}$ we must control the extent of this correlation.  To this end we employ a variation of the Azuma-Hoeffding inequality (Lemma~\ref{lem:concentration}) to argue that since value realizations are independent across time, the pacing sequence and the benchmark allocation cannot be too heavily correlated.

%\bjl{Todo: reconsider the ordering of the rest of this section.  Can we go through the full intuition first, then give detailed definitions after?}

%Our welfare result will depend on a careful analysis of the times when each agent has a pacing multiplier of $0$ and the times the pacing multiplier is strictly greater than $0$.

\subsection{Proof Preliminaries}

 Before getting into the details of the proof of \Cref{thm:main.new}, we first introduce some notation and define an epoch more formally.
 %Before defining an epoch more formally we introduce some notation.
 Write $\mu_{k,t}=\flat$ if at time $t$, the agent's algorithm has stopped, i.e., if the agent is out of money or if $t=T+1$. For $t_1\leq t_2\in \mathbb{N}$, we slightly abuse notation and write $[t_1,t_2]\triangleq \{t_1,\ldots,t_2\}$ to be the set of integers between them, inclusive, when the meaning is clear. %from context.
 Similarly, we write $[t_1,t_2)\triangleq \{t_1,\ldots,t_{2}-1\}$ for the half-open set of integers between them and analogously for $(t_1,t_2]$ and $(t_1,t_2)$.

\begin{definition}
\label[definition]{defn:epoch}
Fix an agent $k$, time horizon $T$, and a sequence of pacing multipliers $\mu_{k,1},\ldots,\mu_{k,T}$. A half-open interval $[t_1,t_2)$ is an \textbf{epoch} with respect to these multipliers if it holds that $\mu_{t_1}=0$ and $\mu_{t}>0$ for each $t_1<t<t_2$ and $t_2$ is maximal with this property.
%An epoch $[t_1,t_2)$ is \textbf{maximal} if $[t_1,t_2+1)$ is not an epoch.  %\bjl{Do we ever discuss non-maximal epochs?  If not, I'd recommend redefining epoch to mean ``maximal epoch."}
\end{definition}

\begin{remark}
Because we initialized any generalized pacing algorithm at $\mu_{k,1}=0$ for all $k\in [n]$, the epochs completely partition the set of times the agent is bidding. Moreover, if the pacing multipliers at times $t_1$ and $t_1+1$ are both zero, then $[t_1, t_1+1)$ is an epoch; we refer to this as a trivial epoch.
\end{remark}

%To see the utility of this definition, we prove in
The following lemma shows that an agent's total spend over a maximal epoch can be bounded from below by an amount roughly equal to the target spend rate $\rho_k$ times the epoch length, plus an adjustment for the first round of the epoch. This lemma is crucial in proving \Cref{thm:main.new} and is proved in \Cref{app:proofs}.
%that directly relates pacing algorithm with this notion: \BLcomment{The lemma provides a lower bound on the total value an agent obtains during an epoch in terms of the total time of the epoch when the agent uses a generalized pacing algorithm. This lower bound is crucial in proving the welfare properties for generalized pacing algorithms.  } \bjl{I adjusted the following lemma so that it does not reference price $p_t$, as this is no longer defined.}

%\bjl{Changed the statement of the lemma to clarify the order of quantification.  I also relaxed the condition from $\mu_{k,T} \neq \flat$ to $\mu_{k,t_2} \neq \flat$.  I realize that we actually only need the lemma to hold for $\mu_{k,T} \neq \flat$, but it also holds if just $\mu_{k,t_2} \neq \flat$, right?} \jg{[JG: yes, that's right!]}

%\bjlcomment{Idea -- float up this lemma, make it more visible, so we can describe how it isn't satisfied by other learning methods?}

\begin{lemma}
\label[lemma]{lem:stronger}
Fix an agent $k$, fix any choice of core auction and any sequence of bids of other agents $(\boldsym{b}_{-k,1}, \dotsc, \boldsym{b}_{-k,T})$.  Fix any realization of values $(v_{k,1}, \dotsc, v_{k,T})$ for agent $k$ and suppose that $\mu_{k,1},\ldots,\mu_{k,T}$ is the sequence of multipliers generated by the generalized pacing algorithm given the auction format and the other bids.  Then for any epoch $[t_1,t_2)$ where $\mu_{k,t_2} \neq \flat$, we have
%is an epoch as given in \Cref{defn:epoch}.
%with parameters satisfying \Cref{lem:bg1} (\BLcomment{Do you use it?.. in particular NO2 or NO is needed? important also for the statement of the welfare result} on a completely arbitrary sequence of values and competing bids and $\mu_{k,T}\neq \flat$. Suppose further that $[t_1,t_2)$ is an epoch as given in \Cref{defn:epoch}. Then the following inequality holds:
\begin{equation}
\label{eq:val}
    \sum_{t=t_1}^{t_2-1} x_{k,t}v_{k,t}\geq x_{k, t_1}v_{k,t_1} - z_{k,t_1}+\rho_k \cdot (t_2-t_1-1).
\end{equation}
\end{lemma}

Next, motivated by the intuition in Section~\ref{sec:lw.intuition}, we introduce a concentration inequality that will be helpful for our analysis.  Roughly speaking, we will use this lemma to show that the sequence of values obtained by agent $k$ in the benchmark on rounds in which $\mu_{k,t} = 0$ are not ``far from expectation,'' in the sense that the total expected value obtained over such rounds is not much greater than $\rho_k$ per round.

\begin{lemma}
\label[lemma]{lem:concentration}
Let $Y_1,\ldots,Y_T$ be random variables and $\mathcal{F}_0\subseteq\ldots\subseteq\mathcal{F}_T$ be a filtration such that:
\begin{enumerate}
    \item $0\leq Y_t\leq \overline{v}$ with probability $1$ for some parameter $\overline{v} \geq 0$ for all $t$.
    \item $\mathbb{E}[Y_t]\leq \rho$ for some parameter $\rho\geq 0$ for all $t$.
    \item For all $t$, $Y_t$ is $\mathcal{F}_t$-measurable but is independent of $\mathcal{F}_{t-1}$.
\end{enumerate}
Suppose that $X_1,\ldots,X_n\in [0,1]$ are random variables such that $X_t$ is $\mathcal{F}_{t-1}$-measurable. Then
\begin{equation}
\label{eq:concentration}
    \textstyle \Pr\rbr{\sum_{t=1}^T X_tY_t+(1-X_t)\rho
    \geq \rho\cdot T + \theta}\leq \exp\left(\frac{-2\theta^2}{T\overline{v}^2}\right).
    \end{equation}
\end{lemma}

The proof of Lemma~\ref{lem:concentration} appears in \Cref{app:proofs}.  We are now ready to prove Theorem~\ref{thm:main.new}.

% \begin{proof}[Proof of Theorem~\ref{thm:main.new}]

\subsection{Proof of \Cref{thm:main.new}}

We first introduce some notations.  Fix the generalized pacing dynamics algorithm used by each agent. We will then write $\boldsym{x} = \{ x_{k,t} \}_{k\in [n],t\in [T]}$ for the random variable corresponding to the allocation obtained under these bidding dynamics, given the values $\{ v_{k,t} \}$.  We also write $\mu_{k,t}$ for the pacing multiplier of agent $k$ in round $t$, and $z_{k,t}$ for the realized spend of agent $k$ in round $t$, which are likewise random variables. For notational convenience we will write $\mathsf{WEL}_{\GPD}(\boldsym{v})$ for the liquid welfare (where \GPD stands for ``Generalized Pacing Dynamics") given valuation sequence $\boldsym{v}$.  That is,
%We now define the liquid welfare of these dynamics in the natural way: \bjl{[BJL: Is there any hope of bounding the ex post liquid welfare of \GPD?  Either way, we should comment on this.]}
\begin{equation}
\textstyle
     \mathsf{WEL}_{\GPD}(\boldsym{v})
     \triangleq \sum_{k=1}^n\;
     \min\cbr{B_k,\sum_{t=1}^T\; x_{k,t}v_{k,t}}.
 \end{equation}
We also write $\mathsf{WEL}_{k,\GPD}(\boldsym{v})$ for the liquid welfare obtained by agent $k$, and $\mathsf{WEL}_{\GPD}(F)$ for the total expected liquid welfare $\mathbb{E}_{\boldsym{v} \sim F^T}[\mathsf{WEL}_{\GPD}(\boldsym{v})]$.

%By Lemma~\ref{lem:single.round}, we can assume without loss of generality that $y_{k,t}(\boldsym{v}_1, \dotsc, \boldsym{v}_T)$ depends only on $\boldsym{v}_t$, and this dependency is identical across rounds.  So in a slight abuse of notation we will think of $\boldsym{y}$ as a single-round allocation rule, so that the allocation to agent $k$ in round $t$ under $\boldsym{y}$ is simply $y_k(\boldsym{v}_t)$.

We claim that to prove \cref{thm:main.new}, it is sufficient to show that inequality \eqref{eq:wel bound} holds for any allocation rule $\boldsym{y}$
such that
\begin{align}\label{eq:y constraint}
\mathbb{E}_{\boldsym{v}\sim F}[y_k(\boldsym{v})v_{k}]\leq \rho_k
\quad\text{ and }\quad
\overline{W}_k(\boldsym{y},F) = T\cdot \mathbb{E}_{\boldsym{v}\sim F}[y_k(\boldsym{v})v_{k}]\leq \rho_k \cdot T
\qquad \forall k\in [n].
\end{align}
This is sufficient because {if} one of the above conditions is violated, one can always decrease the allocation for agent $k$, which maintains the feasibility (since we assume that the set of feasible allocations $X$ is downward closed)
without affecting the ex ante liquid welfare $\overline{W}(\boldsym{y},F)$.
We will therefore assume without loss that $\boldsym{y}$ satisfies \eqref{eq:y constraint}.

Preliminaries completed, we now prove Theorem~\ref{thm:main.new} in three steps.  First, we will define a ``good'' event in which the benchmark allocations are not too heavily correlated with the pacing multipliers of the generalized pacing dynamics, and show that this good event happens with high probability.  Second, for all valuation sequences $\boldsym{v}$ that satisfy the good event, we bound the liquid welfare obtained by the pacing dynamics in terms of the benchmark allocation and the payments collected by the auctioneer.  In the third and final step we take expectations over all valuation profiles to bound the expected liquid welfare.

\medskip \noindent
\textbf{Step 1: A Good Event.}
For each agent $k\in [n]$, we define the following quantity, whose significance will become apparent shortly:
\begin{equation*}
    R_k(\boldsym{v}) \triangleq
    \textstyle \sum_{t=1}^T\; \left[\boldsym{1}\{\mu_{k,t} = 0\}y_{k}(\boldsym{v})v_{k,t} + \boldsym{1}\{\mu_{k,t}\neq 0\}\rho_k\right].
\end{equation*}
$R_k(\boldsym{v})$ is the total value obtained by agent $k$ under allocation rule $\boldsym{y}$, except that in any round in which $\mu_{k,t} \neq 0$ this value is replaced by the target spend $\rho_k$.

We would like to apply \Cref{lem:concentration} to bound $R_k(\boldsym{v})$.  Some notation: we will write $Y_t = y_{k}(\boldsym{v}_t)v_{k,t}$ and $X_t = \boldsym{1}\{\mu_{k,t}=0\}$ with $\mathcal{F}_t = \sigma(\boldsym{v}_1,\ldots,\boldsym{v}_t)$.
Then because $\mu_{k,t}$ is $\mathcal{F}_{t-1}$ measurable and the sequence $\{ \boldsym{v}_j \}$ is a sequence of independent random variables, \Cref{lem:concentration} implies that with probability at least $1-1/(\overline{v}nT)^2$, we have
\begin{equation*}
    \textstyle \sum_{t=1}^T\; \left[\boldsym{1}\{\mu_{k,t} = 0\}y_{k}(\boldsym{v})v_{k,t} + \boldsym{1}\{\mu_{k,t}\neq 0\}\rho_k\right] \leq \rho_k \cdot T + \overline{v}\sqrt{T\log(\overline{v}nT)}.
\end{equation*}
Taking a union bound over $k\in [n]$,
with probability at least $1-1/(\overline{v}nT)^2$ over the randomness in the sequence $\boldsym{v}=(\boldsym{v}_1,\ldots,\boldsym{v}_T)$,
we have that
\begin{equation}
\label{eq:Rk}
%\tag{GOOD}
    \quad R_k(\boldsym{v}) \leq \rho_k\cdot T +\overline{v}\sqrt{T\log(\overline{v}nT)},
    \qquad \forall k\in [n].
\end{equation}

We will write $E_\mathsf{GOOD}$ for the event in which \eqref{eq:Rk} holds.  Going back to the intuition provided before the statement of Lemma~\ref{lem:concentration}, $E_\mathsf{GOOD}$ is the event that the value each agent obtains in the benchmark allocation $\boldsym{y}$ is not ``too high'' on rounds in which their pacing multipliers are $0$.

\medskip \noindent
\textbf{Step 2: A Bound on Liquid Welfare For ``Good'' Value Realizations.}
Fix any realized sequence $\boldsym{v}_1,\ldots,\boldsym{v}_T$ such that \Cref{eq:Rk} holds.
We will now proceed to derive a lower bound on the liquid welfare of the agents (under allocation $\boldsym{x}$) by considering the two different possible cases for $\mathsf{WEL}_{k,\GPD}(\boldsym{v})$.
Recall that the liquid welfare $\mathsf{WEL}_{k,\GPD}(\boldsym{v})$ of any agent $k$ is either $B_k = \rho_k\cdot T$ or is $\sum_{t=1}^T x_{k,t} v_{k,t}$.
For any $k$ such that $\mathsf{WEL}_{k,\GPD}(\boldsym{v}) = B_k$, we obtain via \Cref{eq:Rk} that
\begin{equation}
\label{eq:highagents}
    \mathsf{WEL}_{k,\GPD}(\boldsym{v}) = B_k \geq R_k(\boldsym{v}) - \overline{v}\sqrt{T\log(\overline{v}nT)}.
\end{equation}

% For any agent $k$ with valuation below the budget $B_k$,
We now consider all of the remaining agents.  Let $A \subseteq [n]$ be the set of agents $k$ such that $\sum_{t=1}^Tx_{k,t}v_{k,t}<B_k$ on this realized sequence $\boldsym{v}$, and hence $\mathsf{WEL}_{k,\GPD}(\boldsym{v}) = \sum_{t=1}^Tx_{k,t}v_{k,t}$.
That is,
%the value of any agent $k$ in $A$ under the pacing algorithm is less than $B_k$.
%Thus their
each of their contributions to the liquid welfare on this realized sequence is uniquely determined by their true value for winning the items in the auction.
Note that this further implies that no agent in $A$ runs out of budget early because the generalized pacing algorithm does not allow overbidding (i.e., bidding above the value).
Thus for all $k\in A$,
$\mu_{k,T}\neq \flat$
and $\mu_{k,t}\geq 0$ surely for all $t$.
We claim that
\begin{equation}
\label{eq:stronger.new}
    \sum_{k\in A} \mathsf{WEL}_{k,\GPD}(\boldsym{v})=\sum_{k \in A}\sum_{t=1}^{T} x_{k,t}v_{k,t}\geq \sum_{k \in A}R_k(\boldsym{v}) - \sum_{k \in [n]}\sum_{t=1}^{T}z_{k,t} .
\end{equation}
% To see that this holds, observe that because no agent in $A$ runs out of money on this realization,
% for all $k\in A$,
% $\mu_{k,T}\neq \flat$
% % , so that $\mu_{k,t}=0$ or $\mu_{k,t}>0$.\footnote{\Cref{lem:bg1} holds for \emph{any} sequence of competing prices, so in particular for any set of prices generated by the auction.} Therefore,
% and $\mu_{k,t}\geq 0$ surely.
To show that the inequality holds, we partition the interval $[1,T]$ into maximal epochs for each agent $k$ and bound the value obtained by agent $k$ on each maximal epoch separately.
Fix any agent $k \in A$ and suppose that $[t_1,t_2)$ is a maximal epoch inside $[1,T]$.
Since agent $k$ does not run out of budget on the entire interval $[1,T]$, she does not run out of budget on this epoch in particular.
We can therefore apply Lemma \ref{lem:stronger} to obtain
\begin{equation*}
    \sum_{t=t_1}^{t_2-1} x_{k,t}v_{k,t}\geq x_{k, t_1}v_{k,t_1} - z_{k,t_1}+\rho_k \cdot (t_2-t_1-1).
\end{equation*}
Since $[1,T]$ can be partitioned into maximal epochs for each agent $k\in A$, we can sum over all time periods and apply the definition of a maximal epoch to conclude that
\begin{equation*}
\textstyle    \sum_{t=1}^{T} x_{k,t}v_{k,t} \geq \sum_{t=1}^{T}[\boldsym{1}\{\mu_{k,t}=0\}(x_{k,t}v_{k,t} - z_{k,t}) + \boldsym{1}\{\mu_{k,t}\neq 0\}\cdot \rho_k].
\end{equation*}
Summing over all $k \in A$ and changing the order of the summations implies that
\begin{equation}
\label{eq:group.sum.new}
        \sum_{k \in A}\sum_{t=1}^{T} x_{k,t}v_{k,t} \geq \sum_{t=1}^{T}\sum_{k \in A}[\boldsym{1}\{\mu_{k,t}=0\}(x_{k,t}v_{k,t} - z_{k,t})] + \sum_{k \in A}\sum_{t = 1}^{T}\boldsym{1}\{\mu_{k,t}\neq 0\}\cdot \rho_k.
\end{equation}
We will now use the assumption that the auction is a core auction.  For any $t \in [1,T]$, let $S \subseteq A$ be the set of agents in $A$ for which $\mu_{k,t} = 0$.  We have that
\begin{align*}
&\textstyle
\sum_{k \in A}[\boldsym{1}\{\mu_{k,t}=0\}(x_{k,t}v_{k,t} - z_{k,t})]
= \sum_{k \in S}(x_{k,t}v_{k,t} - z_{k,t}) \\
&\qquad \geq \sum_{k \in S}y_k(\boldsym{v}_t)v_{k,t} - \sum_{k =1}^{n}z_{k,t}
= \sum_{k \in A}\boldsym{1}\{\mu_{k,t}=0\}y_k(\boldsym{v}_t)v_{k,t} - \sum_{k=1}^n z_{k,t}.
\end{align*}
The inequality follows from the definition of a core auction and the no unnecessary pacing condition (which implies $b_{k,t}=v_{k,t}$ for all $k\in S$), by considering the deviation in which the agents in $S$ jointly switch to allocation $\{ y_k(\boldsym{v}_t) \}$.
Substituting the above inequality into \Cref{eq:group.sum.new} and rearranging yields \Cref{eq:stronger.new}.

Summing over \Cref{eq:highagents} for each $k\not\in A$ and combining it with \Cref{eq:stronger.new}, we obtain that as long as inequality \Cref{eq:Rk} holds (i.e., event $E_\mathsf{GOOD}$ occurs), we have
\begin{equation}
\label{eq:RkLB}
    \sum_{k\in [n]} \mathsf{WEL}_{k,\GPD}(\boldsym{v})\geq \sum_{k\in [n]} R_k(\boldsym{v}) - \sum_{k\in [n]}\sum_{t\in[T]} z_{k,t} - n\overline{v}\sqrt{T\log(\overline{v}nT)}.
\end{equation}

\medskip \noindent
\textbf{Step 3: A Bound on Expected Liquid Welfare.}
Since the liquid welfare is nonnegative, we can take expectations over $\boldsym{v}_1,\ldots,\boldsym{v}_T$ to conclude from \eqref{eq:RkLB} that
% we obtain that
\begin{align}
\label{eq:group.sum2.new}
    %\mathsf{WEL}_{\GPD}(F) &=
    &\textstyle\E\left[\sum_{k=1}^n \mathsf{WEL}_{k,\GPD}(\boldsym{v})\right]
    \textstyle \geq \E\left[\boldsym{1}\{E_\mathsf{GOOD}\}\cdot\sum_{k=1}^n \mathsf{WEL}_{k,\GPD}(\boldsym{v})\right] \nonumber\\
    &\qquad\qquad\geq \textstyle  \E\left[\boldsym{1}\{E_\mathsf{GOOD}\}\cdot\sum_{k\in [n]} R_k(\boldsym{v})\right]
    - \E\left[ \sum_{k = 1}^n\sum_{t = 1}^{T} z_{k,t} \right]
    -n\overline{v}\sqrt{T\log(\overline{v}nT)},
\end{align}
where the last inequality holds via \Cref{eq:RkLB}.

It remains to analyze the expectations on the right side of the inequality. First, note that
\begin{align*}
\textstyle
    \E\left[\boldsym{1}\{E_\mathsf{GOOD}\}\sum_{k\in [n]} R_k(\boldsym{v})\right]
    &=\textstyle
    \E\left[\sum_{k\in [n]} R_k(\boldsym{v})\right] - \E\left[\left(1-\boldsym{1}\{E_\mathsf{GOOD}\}\right)\sum_{k\in [n]} R_k(\boldsym{v})\right]\\
    &\geq \textstyle
    \E\left[\sum_{k\in [n]} R_k(\boldsym{v})\right] - \frac{n\overline{v}T}{n\overline{v}T^2}
    =\E\left[\sum_{k\in [n]} R_k(\boldsym{v})\right] - 1/T.
\end{align*}
The inequality holds due to the fact that $R_k(\boldsym{v})\leq \overline{v}T$ as well as our bound on the probability that event \Cref{eq:Rk} does not hold.
Let $q_{k,t}$ be the unconditional probability that $\mu_{k,t}=0$. Then
\begin{align*}
    \E\sbr{R_k(\boldsym{v})}
    &=\textstyle \sum_{t=1}^T \E\sbr{\boldsym{1}\{\mu_{k,t}=0\}
    \E[y_{k}(v_{k,t})v_{k,t}\vert \mathcal{F}_{t-1}]+\boldsym{1}\{\mu_{k,t}\neq 0\}\rho_k} \\
    &= \textstyle\sum_{t=1}^T[q_{k,t}\E[y_k(v_k)v_k]+(1-q_{k,t})\rho_k]
    \geq \sum_{t=1}^T \E[y_k(v_k)v_k] = \overline{W}_k(\boldsym{y},F).
\end{align*}
The first inequality uses the conditional independence of $y_{k}(v_{k,t})v_{k,t}$ on $\mu_{k,t}$, as this is already determined by time $t-1$.
The inequality holds since $\mathbb{E}[y_k(v_k)v_k]\leq \rho_k$ according to our assumption on~$\boldsym{y}$.
% Therefore, each summand is a convex combination of $\rho_k$ and $\mathbb{E}[y_k(v_k)v_k]\leq \rho_k$, and thus is at least $\mathbb{E}[y_k(v_k)v_k]$, giving the inequality.
Substituting the inequalities into \Cref{eq:group.sum2.new}, we obtain that
\begin{equation}
\label{eq:wel.gpd}
\mathsf{WEL}_{\GPD}(F)
\geq
\sum_{k=1}^n \overline{W}_k(\boldsym{y},F)- \mathbb{E}\left[ \sum_{k = 1}^n\sum_{t = 1}^{T} z_{k,t} \right]-n\overline{v}\sqrt{T\log(\overline{v}nT)}-1/T.
\end{equation}
Recall that
\begin{equation}
    \mathbb{E}\left[ \sum_{k = 1}^n\sum_{t = 1}^{T} z_{k,t} \right] = \sum_{k = 1}^n\sum_{t=1}^{T}\mathbb{E}\left[z_{k,t}\right]=\sum_{k=1}^n \mathbb{E}[P_k]\leq \mathsf{WEL}_{\GPD}(F),
\end{equation}
where $P_k$ is the total expenditure of agent $k$, and this is upper bounded by the liquid value they obtain. Substituting into \eqref{eq:wel.gpd} and rearranging the terms, then noting that $\mathsf{WEL}_{\GPD}(F)$ is precisely the left-hand side of \eqref{eq:wel bound}, we conclude that inequality \eqref{eq:wel bound} holds.

%\bjlcomment{Got to here}

%\fi

\section{Individual Guarantee: Vanishing Regret for MBB Auctions}
\label{sec:regret}

In this section, we supplement our aggregate welfare guarantees with bounds on the individual performance of \Cref{alg:bg} when used in any auction that satisfies the monotone bang-for-buck (MBB) condition. We focus on a particular single agent $k$, henceforth called simply \emph{the agent}. The agent faces an online bidding problem subject to the budget constraint. The action set consists of pacing multipliers $\mu\geq 0$ like in \Cref{alg:bg}.  That is, the agents are restricted to linear values-to-bids policies without overbidding. We show that under mild conditions on the environment, the agent achieves regret bounds for  utility maximization; in fact, we show that the algorithm achieves regret bounds for any combination of utility and value maximization.

% under rather general conditions.
%\bjledit{It will be technically convenient to first establish regret bounds for the objective of maximizing value, before showing in Section~\ref{sec:regret.utility} how to extend our regret bounds to utility-maximizing agents.}
%\ascomment{AE asked to clarify the agent's problem ... }

\subsection{Environment and Notation}

We make no explicit assumptions on the individual behavior of the other agents. Their bidding strategies $\boldsym{b}_{-k}$ can depend arbitrarily on the realized values and the observed history.  This is desirable since the agents may be reluctant to follow a particular bidding algorithm, and even when they do, the aggregate behavior is not well-understood.

In each round $t$, let $G_t$ be the joint distribution of value $v_{k,t}$ and other agents' bids $\boldsym{b}_{-k,t}$ given the history in the previous rounds. Thus, one can think of pair
    $\rbr{v_{k,t},\, \boldsym{b}_{-k,t}}$
as being drawn from a distribution $G_t$ that could depend on the history up to time $t$. In particular, $G_t$ could depend on the previous distributions $G_1 \LDOTS G_{t-1}$ and the agent's past bids.

We can think of the sequence $G_1 \LDOTS G_T$ as specifying the environment that the agent operates in. Our main guarantee in \Cref{thm:regmain} considers an \emph{adversarial environment}, \ie an arbitrary sequence. We also consider important special cases, particularly the \emph{stochastic environment} where $G_t$ is the same in all rounds.

The dependence of the agent's problem on $G_t$ is captured via the following notation.
Let
    $V_t(\mu_{k,t})$ and $Z_t(\mu_{k,t})$
be, respectively, the agent's expected value and expected payment in round $t$ for a particular multiplier $\mu_{k,t}$, where the expectation is taken over the distribution $G_t$. More formally, for a given auction defined by an allocation rule $\boldsym{x}$ and a payment rule $\boldsym{p}$ (see Section \ref{sec:model})
and a given distribution $G_t$
%\jg{and some past sequence of bids up to time $t$},
we have
\begin{align}
Z_{t} (\mu_{k,t} )
    = \E\sbr{p_k \rbr{b_{k,t}, \boldsym{b}_{-k,t}}}
    \text{ and }
V_{t}(\mu_{k,t})
    = \E\sbr{ v_{k,t}\cdot x_k \rbr{ b_{k,t}, \boldsym{b}_{-k,t}}},
\end{align}
where the bid is
    $b_{k,t} = v_{k,t}/(1+\mu_{k,t})$
and the expectations are over $\rbr{v_{k,t}, \boldsym{b}_{-k,t}} \sim G_t$. Under the same conventions, agent's expected (quasilinear) utility $U_t(\cdot)$ is then
\begin{equation*}
    U_t(\mu_{k,t}) = \E\sbr{ v_{k,t}\cdot x_k \rbr{ b_{k,t}, \boldsym{b}_{-k,t}}-p_k \rbr{b_{k,t}, \boldsym{b}_{-k,t}}}=V_t(\mu_{k,t})-Z_t(\mu_{k,t}).
\end{equation*}
Note that, like $G_t$, the functions $V_t(\cdot)$, $Z_t(\cdot)$ and $U_t(\cdot)$ can depend on the observed history up to round $t$. Monotonicity of the auction allocation and payment rules implies that, for any fixed history, $Z_t$ and $V_t$ are monotonically non-increasing.

In what follows, we drop the dependence on the agent's index $k$ from our notation.
% in the remainder of this section.
%, and refer to her simply as ``the agent."
For each round $t$, we write $v_t = v_{k,t}$, $\mu_t=\mu_{k,t}$ for the agent's value and multiplier, and $b_t = v_t/(1+\mu_t)$ for its bid. Also, we write $\rho = \rho_k$ and $\eps = \eps_k$.

\subsection{Unified Objective and Pacing Regret}

For a unified presentation of both utility- and value-optimization, we employ a unified objective: an arbitrary convex combination of the two. Specifically, the unified objective is defined as
\begin{align}\label{eq:uniObj}
\uniObj(\mu_1 \LDOTS \mu_T)
    := \sum_{t\in[\tau_k]} \gamma\cdot U_t(\mu_t)
        + (1-\gamma)\cdot V_t(\mu_t),
\end{align}
where $\gamma\in[0,1]$ is a fixed parameter and $\tau_k$ is the stopping time (\ie the first time when the budget is not strictly positive). The $t$-th summand on the right-hand side of \eqref{eq:uniObj} is denoted $\uniObj[t](\mu_t)$.

%\yldelete{For example, suppose that there is a first-price auction with no reserve price at round $t$ and assume for simplicity that ties break in favor of agent $k$. Then for agent $k$ with a value $v_{k,t}$ and a pacing multiplier $\mu_{k,t}$, we have
%$$z_{t}(\mu_{k,t}) = \frac{v_{k,t}}{ 1+\mu_{k,t} }1_{ \left \{\frac{v_{k,t}}{ 1+\mu_{k,t} } \geq \max _{i \neq k} {b}_{i,t} \right \} }, \text{ } Z_{t} (\mu_{k,t} ) = \mathbb{E}_{G_{t}} [ z_{t}(\mu_{k,t})], \text{ } V_{t}(\mu_{k,t} ) = \mathbb{E}_{G_{t}}\left [ v_{k,t} 1_{ \left \{\frac{v_{k,t}}{ 1+\mu_{k,t} } \geq \max _{i \neq k} {b}_{i,t} \right \}  } \right ]. $$
%}

%\jgedit{We obtain all of our regret bounds for utility maximization, as well as value maximization},

Our regret bounds are relative to a non-standard benchmark: essentially, a sequence of pacing multipliers
    $\mu^*_1 \LDOTS \mu^*_T\in [0,\overmu]$
with expected spend $Z_t(\mu^*_t)=\rho$ in each round $t$. The formal definition needs to account for the possibility that $Z_t(0)<\rho$:

\begin{definition}
For each round $t$ and any auction history up to round $t$, a \emph{perfect pacing multiplier} is any $\mu^*_t\in [0,\overmu]$ such that $Z_t(\mu^*_t)=\rho$, or $\mu^*_t = 0$ if $Z_t(0) < \rho$. A \emph{perfect pacing sequence} (for a given auction history up to round $T$) is a sequence $\mu^* = (\mu^*_1 \LDOTS \mu^*_T)$, where each $\mu^*_t$ is a perfect pacing multiplier.
\end{definition}

Our assumptions (stated below) ensure that such $\mu^*$ exists, although in general it need not be unique. However, all perfect pacing sequences $\mu^*$ have the same
unified objective $\uniObj(\mu^*)$.
\footnote{This follows by MBB. Fix the history up to some round $t$. Taking expectations on both sides of \refeq{eq.mbb.1}, one obtains
    $Z_t(\nu_t')-Z_t(\nu_t) \geq \tfrac{1}{1+\nu_t} \rbr{V_t(\nu_t') - V_t(\nu_t)}$
for any two pacing multipliers $\nu_t>\nu_t'$. If they are \emph{perfect} pacing multipliers for round $t$, then $Z_t(\nu_t)=Z_t(\nu_t') = \rho$ and consequently $V_t(\nu_t) = V_t(\nu_t')$. It follows that $\uniObj[t](\nu_t)= \uniObj[t](\nu_t')$.
}

Thus, we treat $\uniObj(\mu^*)$ as a benchmark, and define \emph{pacing regret} relative to this benchmark:
\begin{align}\label{eq:def-pacing-regret}
    \PaceReg(T)
        :=  \E\sbr{\uniObj(\mu^*) - \uniObj(\mu_1 \LDOTS \mu_T)},
\end{align}
where $\mu^*$ is (any) perfect pacing sequence for a given auction history up to round $T$.

\begin{remark}
We emphasize that the benchmark $\uniObj(\mu^*)$ is defined with respect to the  realized history of bids.
%We emphasize that, in the definition of pacing regret, the function $V_t(\cdot)$ used in the benchmark calculation is defined with respect to the true realized history of bids.  This is important because other agents' bids, and hence the perfect pacing sequence $(\mu^*_1, \dotsc, \mu^*_T)$ itself, can depend arbitrarily on the history of bids.
As with most regret guarantees in the literature, the benchmark does not include the counterfactual impact of a change of bid by agent $k$ in round $t$ on future bids by other agents. Rather, our pacing regret guarantees apply to the benchmark of perfect pacing in hindsight assuming that other bidders do not change their bidding behavior in response.
\end{remark}

\begin{remark}
A perfect pacing sequence $\mu^*$ is not necessarily best-in-hindsight. That is, fixing the entire history, $\mu^*$ does not necessarily optimize $\uniObj(\cdot)$ among all sequences $\nu\in[0,\overmu]^T$. In fact, $\mu^*$ may even be worse in terms of $\uniObj(\cdot)$ than the best fixed pacing multiplier. This is by design.
Recall from \cref{sec:related} that one cannot obtain sublinear regret in an adversarial environment, even against the best fixed pacing multiplier, because of the spend-or-save dilemma \citep{AdvBwK-focs19,BalseiroGur19}.
While inevitable, this situation may feel  unsatisfying. One interpretation is that the standard benchmark of the best-in-hindsight multiplier is ``too hard" in the adversarial environment, and a more ``fair" alternative benchmark is needed to make progress and properly express the benefits of algorithms such as ours. The reason our benchmark admits vanishing regret is precisely that it gives up on solving the spend-or-save dilemma, and instead optimizes for each round separately.
\end{remark}

\subsection{Smoothness Assumptions}

%Consider $V_t(\cdot)$ and $Z_t(\cdot)$ as functions
%    $\mathbb{R}_{\geq 0}\to \mathbb{R}_{\geq 0}$.
%Both are monotonically non-increasing, by monotonicity of the auction allocation and payment rules.
%\yldelete{ Furthermore, for any auction that satisfies the MBB property, the function $\mu\mapsto V_t(\mu)/Z_t(\mu)$ is nondecreasing as well.}
%It will also be convenient for us to assume that the expenditure function $Z_t$ almost surely satisfies the following smoothness assumptions:
We make the following smoothness assumptions on the expenditure function $Z_t$:

\begin{assumption}%[Regularity]
\label[assumption]{assn:regularity2}
%The environment at each time $t$ surely satisfies:
%\begin{enumerate}
%    \item (Monotonicity) The functions $\mu\mapsto V_t(\mu)$ and $\mu\mapsto Z_t(\mu)$ are nonincreasing in $\mu$.
%    \item (Monotone Bang-for-Buck) The function $\mu\mapsto V_t(\mu)/Z_t(\mu)$ is nondecreasing.
%    \item (Lipschitzness)
There exists $\lambda \geq 0$ such that $Z_t$ is $\lambda$-Lipschitz for each round $t$.
%\end{enumerate}
\end{assumption}

%We now need one final assumption that asserts that the sequence of auctions is not ``too easy" in some sense:

\begin{assumption}
\label[assumption]{assn:boundedMBB}
%There exists $C\geq 0$ such that $V_t(0)/Z_t(0)\leq C$ for all $t$.
There exists $\delta > 0$ such that $Z_t(0)\geq \delta$ for each round $t$. As our bounds will depend inversely on $\delta$, we assume without loss of generality that $\delta\leq \rho$.%
\footnote{We only require the weaker condition: $Z_t(0)\geq \delta \times \tfrac{V_t(0)}{\overline{v}}$ for all $t$.  In particular, this allows $Z_t(0)=0$ if $V_t(0) = 0$.}
\end{assumption}

\noindent Note that these assumptions hold for any history up to round $t$.

%Note that by the monotone bang-for-buck property, $V_t/Z_t$ is minimized at $\mu=0$. This assumption means that it is not possible to ``get something for nothing" in the sense that when an agent bids her value, the expected expenditure is positive. In particular, this assumption holds under the mild condition that there is a positive reserve price.

%We claim that Assumptions~\ref{assn:regularity2} and~\ref{assn:boundedMBB} are very mild
\begin{remark}\label{rem:assns}
We argue that these assumptions are mild in the context of advertising auctions.
%For the former,
%\asdelete{Recall that $Z_t(\mu)$ is the expected payment over all sources of randomness, including randomness in the agent's value for the impression, and that $Z_t(0)$ is the expected payment when the bid equals the agent's value.}
Indeed, the variation in impression types, click rate estimates, etc., and possibly also the randomness in the auction and/or the bidding algorithms
would likely introduce some smoothness into the expected payment $Z_t(\cdot)$, and ensure that the maximum allowable bid would result in a non-trivial payment. In fact, one could eliminate the need for these assumptions by adding a small amount of noise to the auction allocation rule, such as by perturbing bids slightly, and/or a small but positive reserve price, at the cost of a similarly small loss of welfare.
Alternatively, the assumptions hold when all agents use bounded pacing multipliers, the joint value profile distribution is sufficiently smooth (for \cref{assn:regularity2}),%
\footnote{If the pricing rule is nondecreasing in the agent's bid given the other bids (which is the typical case), \cref{assn:regularity2} holds (for some $\lambda>0$) the agent's conditional valuation given the other bids almost surely admits a density that is bounded pointwise by some absolute constant $\sigma>0$. More generally, \cref{assn:regularity2} holds for any bounded pricing rule if the joint valuation distribution has a Lipschitz density, since then the induced density on joint bids is nicely behaved as a function of the agent's multiplier.}
and with probability at least $\eps'>0$, agent $k$ has the highest value while some other agent bids at least $\delta'>0$ (for \cref{assn:boundedMBB}).
%Alternatively, the first assumption would even continue to hold with competing agents \emph{also} using bounded pacing multipliers to bid given their values so long as the joint value profile distribution is sufficiently smooth. The second condition holds under similar fairly general conditions, like the first agent having lower bounded probability of having the highest valuation while some other bidder has a nontrivial value and all pacing multipliers are again bounded.
\end{remark}

We also assume that
    $\overmu \geq \overline{v}/\rho$,
where $\overmu,\overline{v}$ are, resp., upper bounds on the largest feasible pacing multipliers and values.
Given these assumptions, we can now prove (as promised above) that a perfect pacing sequence is well-defined.

%For the latter assumption, note that the monotone bang-for-buck property implies that $V_t(\mu)/Z_t(\mu)$ is minimized at $\mu=0$.  The assumption
%implies that it is not possible to ``get something for nothing" in the sense that if an agent bids her value,
%the expected expenditure (over all randomness in values and other bids) is positive.

%The following assumption is for convenience to help simplify our arguments:
%\begin{assumption}[Existence of Pacing Bids]
%\label{assn:exists2}
%For each $t$, there exists $\mu\in [0,\overmu]$ such that $Z_t(\mu) = \rho$.
%\end{assumption}

\begin{claim}\label[claim]{cl:perfect-exists}
A perfect pacing multiplier $\mu^*_t$ exists, for any given round $t$ and history up to this round.
\end{claim}

\begin{proof}
Note that
%we first note that as long as $\overmu \geq \overline{v}/\rho$, we must have
$Z_t(\overmu) \leq \overline{v}/(1+\overmu) \leq \rho$.  The Lipschitzness and monotonicity of $Z_t$ therefore imply that if there does not exist any $\mu$ such that $Z_t(\mu) = \rho$,  then it must be that $Z_t(0) < \rho$ (and hence $Z_t(\mu) < \rho$ for all $\mu \in [0, \overmu]$).  This latter case occurs if even bidding the true value $v_t$ generates expected spend less than $\rho$.  In this case we take $\mu^*_t = 0$, which corresponds to the most that the agent can bid without violating the mandate to not bid more than the value $v_t$.
\end{proof}

%To state our main result of the section, we first define the regret of the pacing algorithm relative to the perfect pacing sequence, called \emph{pacing regret}, for both utility and value maximization:
%\begin{align}\label{eq:def-utility-regret}
%    \PaceReg^{\jgedit{U}}(T) \triangleq \sum_{t\in[T]} \E\sbr{U_t(\mu^*_t)} - \sum_{t\in[T]} \E\sbr{U_t(\mu_t)},
%\end{align}
%\begin{align}\label{eq:def-pacing-regret}
%    \PaceReg^{\jgedit{V}}(T) \triangleq \sum_{t\in[T]} \E\sbr{V_t(\mu^*_t)} - \sum_{t\in[T]} \E\sbr{V_t(\mu_t)}.
%\end{align}
%For any $\gamma\in [0,1]$, we can thus define the convex combination of these regrets:

%\begin{align}\label{eq:def-combined-regret}
%    \PaceReg^{\gamma}(T) \triangleq \gamma \PaceReg^{\jgedit{U}}(T)+ (1-\gamma)\PaceReg^{\jgedit{V}}(T),
%\end{align}

\subsection{Main Guarantees: Adversarial Environment}

With all definitions in place, we are ready to state our main result of this section: an upper bound on pacing regret, which holds against an arbitrary auction environment and an arbitrary perfect pacing sequence. The guarantee applies uniformly to any mixture of value- and utility-maximization objective, without any dependence on the mixing parameter $\gamma$.

%However, it deteriorates when the environment changes too much, as quantified by the path-length
%    $\sum_{t=1}^{T-1} |\mu^*_{t+1}-\mu^*_t|$.

%\bjl{In the theorem statement below, couldn't we keep $\epsilon_k$ unspecified, and get a regret bound that might be non-trivial even for $P = \Omega(\sqrt{T})$? We should get $REG_1 = O(P/\epsilon + T \epsilon)$, so if $P = \Theta(T^{\alpha})$ we can take $\epsilon = T^{-(1-\alpha)/2}$ and get total regret $O(T^{(1+\alpha)/2}$, right?}

\begin{theorem}\label{thm:regmain}
Consider a repeated auction that is individually rational (IR) and satisfies monotone bang-per-buck (MBB) property. Posit Assumptions \ref{assn:regularity2} and \ref{assn:boundedMBB}. Suppose the path-length
 $P^* = \sum_{t=1}^{T-1} |\mu^*_{t+1}-\mu^*_t|$
is upper-bounded by some number $P$ with probability~$1$. Fix parameter $\gamma\in [0,1]$ in the unified objective. Consider \cref{alg:bg} with step-size 
{$\eps\in (0,\overline{v}]$}.
Its pacing regret is
%$\eps_k =1/\sqrt{T}$ satisfies
%$\mathsf{REG}_1 \triangleq ((P+1)\overmu^2+(\rho+\overline{v})^2)\sqrt{T}$,
\begin{equation}\label{eq:thm-regmain}
%    \sum_{t\in[T]} \E[V_t(\mu^*_t)] - \sum_{t\in[T]} \E[V_t(\mu_t)] =
\PaceReg(T) < O\rbr{\frac{\overline{v}}{{\delta}}\sqrt{2\lambda T\cdot\mathsf{REG}_1}+\frac{\overline{v}\overmu}{\eps\rho}},
\quad\text{where }
\mathsf{REG}_1
    \triangleq \frac{P+1}{\eps}\cdot \overmu^2 +
         \eps T\cdot (\rho+\overline{v})^2.
\end{equation}
\end{theorem}

\begin{corollary}\label[corollary]{cor:regmain}
In the setting of \cref{thm:regmain}, assume that parameters
    $(\overline{v},\overmu,\lambda,\delta)$
are absolute constants. The regret bounds simplify as follows:
\begin{itemize}
\item[(a)]
$\PaceReg(T) < O\rbr{\sqrt{ T \rbr{ \frac{P+1}{\eps} + \eps T }}}$.

\item[(b)]
Taking step-size $\eps=1/\sqrt{T}$ (not knowing $P$), we obtain
    $\PaceReg(T) < \sqrt{P+1}\cdot O\rbr{T^{3/4}}$.

\item[(c)] If $P$ is known, taking $\eps=\sqrt{\frac{P+1}{T}}$ yields an improved bound,
    $\PaceReg(T) < (P+1)^{1/4}\cdot O\rbr{T^{3/4}}$.
\end{itemize}
\end{corollary}

%\subsection{Discussion, Corollaries and Extensions}
%\label{sec:regret-discuss}

%\cref{thm:regmain} treats the auction environment as a black box that produces a $\rbr{v_{k,t},\, \boldsym{b}_{-k,t}}$ pair for each round $t$ according to some distribution $G_t$. The theorem makes no explicit assumptions on the other agents' algorithms or their aggregate behavior. This approach is desirable since the agents may be reluctant to follow a particular bidding algorithm, and even when they do, the aggregate behavior is not well-understood.

We emphasize that the environment, as expressed by distributions $G_1 \LDOTS G_T$, can change over time. The dependence on this change is summarized with a uniform upper bound $P$ on path-length
 $P^* = \sum_{t=1}^{T-1} |\mu^*_{t+1}-\mu^*_t|$.
We obtain a non-trivial guarantee for an arbitrary non-stationary environment with $P=o(T)$ if $P$ is known, and $P = o(\sqrt{T})$ otherwise.

Path-length $P^*$ is a rather \emph{weak} notion of change. Indeed, distribution $G_t$ can change a lot from one round to another while $\mu^*_t$ stays the same. Moreover, if distribution $G_t$ changes arbitrarily from one round to another and stays the same afterwards, the path-length only increases once, by at most $\overmu$.
\footnote{For concreteness: if the perfect pacing sequence has at most $N$ changes from one round to another (\eg because do does the sequence of distributions $G_t$), then \Cref{cor:regmain}(bc) holds with $P = \overmu\,N$.}
Therefore, our guarantees go far beyond slight perturbations of the same stochastic environment.

% AS: commented out, because I/we realized that this result easily follows
% from the stochastic results. Instead, moved to a footnote.
%
%in fact, our environment may be very far from any one stochastic environment. To make this point concrete, consider a ``piecewise-stochastic" environment, when the sequence of distributions $G_1 \LDOTS G_T$ has a few ``switches" (changes from one round to another) and stays the same in between.

%\begin{corollary}\label[corollary]{cor:regret-switches}
%If sequence $G_1 \LDOTS G_T$ has at most $N$ switches,   %\asedit{\Cref{cor:regmain}(c)} holds with $P = \overmu\,N$.
%\end{corollary}

%\bjlcomment{If we must say something here, it's really important that we replace ``large deviations'' with something more concrete. What exactly are we saying is the technical separation between what we've done and what prior work has done?}
%\bjlcomment{But what I really recommend is removing the remark from this section, and instead making this commentary on novelty in the introduction / related work.  E.g., in the new xhdr on novelty and significance.}

\begin{remark}
This is the first non-trivial regret bound for budget-constrained online bidding in the adversarial environment.
%Moreover, the idea to measure regret against a perfect pacing sequence (to bypass the spend-or-save dilemma) has been fruitfully applied to general BwK problems in subsequent work \citep{LagCBwK-jmlr24,BwK-braverman2025}.}
\end{remark}

Our guarantees extend beyond the original model: the marginal distribution of the agent's value and the auction itself can change over time, as long as the assumptions hold.

\begin{extension}
\Cref{thm:regmain} and \Cref{cor:regmain} hold as written if
%the marginal distribution of the agent's value $v_t$
the distribution from which the value profile $\boldsym{v}_t$ is sampled and the auction itself can arbitrarily change from one round to another, possibly depending on the previous rounds.
\end{extension}

\subsection{Regret Bounds for Stochastic Environment}
\label{sec:regret-stochastic}

Now let us consider the stochastic environment, \ie assume that the distribution $G_t = G$ does not change over time almost surely. Then neither do the functions $V_t, U_t,Z_t,\uniObj[t]$, and there exists some $\nu^*\in[0,\overmu]$ (determined by $G$) that is a perfect pacing multiplier for all rounds $t$.

We immediately obtain a regret bound against $\nu^*$ (we state it formally for completeness). Let $\uniObjFixed(\mu) := \uniObj(\mu \LDOTS \mu)$ be the unified objective when a  fixed pacing multiplier $\mu$ is used in all rounds.

\begin{corollary}\label[corollary]{cor:regret-stochastic-gen}
In the setting of \Cref{thm:regmain}, suppose distributions $G_t$ do not change over time almost surely and assume that parameters
    $(\overline{v},\overmu,\lambda,\delta)$
are absolute constants. Letting $\eps=1/\sqrt{T}$ be the step-size,
\begin{align}
\uniObjFixed(\nu^*) - \E\sbr{\uniObj(\mu_1 \LDOTS \mu_T)}
\leq O(T^{3/4}).
\end{align}
\end{corollary}

It is desirable to have a guarantee against the ``best" pacing multiplier --- one that maximizes $\uniObjFixed(\cdot)$ --- and we prove that $\nu^*$ comes close for value-maximization. It is crucial for this result that we disallow overbidding (as discussed in \Cref{sec:model}), \ie only consider $\mu\geq 0$.

\begin{corollary}\label[corollary]{cor:regret-stochastic}
Consider the setting of \Cref{cor:regret-stochastic-gen} under value-maximization ($\uniObj[t]=V_t$). Then
\begin{align}\label{eq:cor:regret-stochastic}
\sup_{\mu\in[0,\overmu]} \uniObjFixed(\mu) - \E\sbr{\uniObj(\mu_1 \LDOTS \mu_T)}
\leq O(T^{3/4}).
\end{align}
\end{corollary}

\begin{remark}
This is the first non-trivial regret bound for value-maximization in online bidding under budget, in any auction format. While the $T^{3/4}$ regret rate may be suboptimal, we emphasize that the goal of this paper is not (necessarily) to obtain optimal regret rates for a specific environment, but rather to provide a combination of non-trivial aggregate and individual guarantees.
\end{remark}
%\begin{corollary}\label[corollary]{cor:regret-stochastic}
%In the setting of \Cref{cor:regret-stochastic-gen}, assume either that (a) the objective is value-maximization: $\uniObj[t]=V_t$, or that (b) the auction rule is second-price. Then
%\begin{align}
%\sup_{\mu\in[0,\overmu]} \uniObjFixed(\mu) - \E\sbr{\uniObj(\mu_1 \LDOTS \mu_T)} \leq O(T^{3/4}).
%\end{align}
%\end{corollary}

%For the stochastic environment \jgedit{in any MBB auction environment}, it turns out that the perfect pacing multiplier is close to the value-optimizing fixed multiplier subject to no-overbidding, and therefore \Cref{thm:regmain} yields regret bounds against this benchmark. In the case where we further assume that the auction is second-price, it instead will hold that the perfect pacing multiplier is close to the multiplier that optimizes any convex combination of value and utility maximization, and thus \Cref{thm:regmain} again yields regret bounds against the optimal fixed multiplier for this objective as well. This latter guarantee thus generalizes existing no-regret results for utility maximization in truthful auctions.

Moreover, $\nu^*$ (nearly) maximizes $\uniObjFixed(\cdot)$ for utility-maximization in repeated second-price auctions \citep{BalseiroGur19}, and \Cref{alg:bg} is known to achieve $\sqrt{T}$ regret rate in this setting \citep{BalseiroGur19,Balseiro-BestOfMany-Opre}.

\begin{corollary}[\citep{BalseiroGur19,Balseiro-BestOfMany-Opre}]
\label[corollary]{cor:regret-stochastic-BOMW}
Consider the setting of \Cref{cor:regret-stochastic-gen} for utility-maximization ($\uniObj[t]=U_t$) in repeated second-price auctions. Then
\begin{align}
\sup_{\mu\in[0,\overmu]} \uniObjFixed(\mu) - \E\sbr{\uniObj(\mu_1 \LDOTS \mu_T)}
\leq O(\sqrt{T}).
\end{align}
\end{corollary}

\begin{remark}
\citet{BalseiroGur19} prove this regret bound under additional convexity assumptions. \citet{Balseiro-BestOfMany-Opre} prove it without convexity assumptions, for a class of algorithms which (implicitly) includes \Cref{alg:bg} as a special case with Euclidean regularizer. The benchmark $\uniObjFixed(\nu^*)$ is even stronger in this setting: it matches the best bidding policy (\ie mapping from value to bid), which may not be linear \citep{BalseiroGur19}.
\end{remark}

%\begin{remark}
Staying on repeated second-price auctions, we saw that $\nu^*$ maximizes $\uniObjFixed(\cdot)$ for both utility- and value-maximization. Consequently, it does so for any mixture in \eqref{eq:uniObj}. So, we again obtain \refeq{eq:cor:regret-stochastic}.
%\end{remark}

\begin{corollary}\label[corollary]{cor:regret-stochastic-second-price}
In the setting of \Cref{cor:regret-stochastic-gen} for repeated second-price auctions, \refeq{eq:cor:regret-stochastic} holds.
\end{corollary}

\medskip

\subsection{Discussion: Desiderata for Individual Guarantees}

Let's take a step back and discuss what types of results are desirable for individual guarantees. The stochastic environment is commonly understood as a ``minimal" desiderata, as a relatively simple special case for which vanishing regret is feasible. A common motivation is that a ``small" bidder enters a ``large" market in which the bidding dynamics has already converged to a (possibly randomized) stationary state. Of course, the downside is that realistic environments are neither stationary nor guaranteed to converge.

The other extreme is to analyze an adversarial environment,
%treat the auction environment as a ``black box"
like we do in \cref{thm:regmain}, possibly with improved guarantees for some ``nice" special cases. In particular, path-length $P^*$ can quantify convergence to a ``well-behaved" environment in which the perfect pacing multiplier does not change over time, so that our guarantee is strong when the convergence is sufficiently fast. However, such ``black-box" guarantees for budget-constrained bidders tend to be weak or vacuous in the worst case.

A \emph{third} approach would be to explicitly take advantage of the fact that all agents are controlled by bidding algorithms with particular properties, \eg by instances of the same bidding algorithm. In particular, if all \emph{other} agents use our algorithm with step size $\eps'$ (and parameters
    $(\overline{v},\overmu,\lambda,\delta)$
are absolute constants), one can show that the perfect pacing multiplier changes by at most $O(\eps')$ in each round, so that the path-length $P^*$ can be upper-bounded as $O(\eps' T)$. Plugging this back into \cref{cor:regmain}(b), we obtain vanishing regret with step-size $\eps' = o(1/\sqrt{T})$.

\begin{corollary}\label[corollary]{cor:regret-stochastic-2}
%Consider the setting of \cref{thm:regmain} with absolute-constant parameters
%    $(\overline{v},\overmu,\lambda,\delta)$.
Consider a repeated auction that is individually rational (IR) and satisfies monotone bang-per-buck (MBB) property. Assume that parameters
 $\overline{v},\overmu$
are absolute constants.
Suppose all agents use \cref{alg:bg}, so that agent $k$ has step-size $\eps_k =1/\sqrt{T}$ and all other agents have step-size $\eps' = f(T)/\sqrt{T}$, where $f(T) \to 0$.
Posit that Assumptions \ref{assn:regularity2} and \ref{assn:boundedMBB} are satisfied (see \cref{rem:assns}) for some absolute-constant $\lambda$ and $\delta$.
Then for agent $k$, the path-length $P^*$ is at most $O(\eps' T)$ with probability $1$, so the right-hand side of \eqref{eq:thm-regmain} is at most
    $O\rbr{T\cdot \sqrt{f(T)}}$.
\end{corollary}

\noindent Note that this analysis yields
%Unfortunately, we obtain
vanishing regret only for a particular agent $k$, whereas for all other agents the guarantee that comes out of our analysis is vacuous. And if all agents use the same parameter $\eps = 1/\sqrt{T}$, then we can only guarantee path-length $P^* = O(\sqrt{T})$, which does not suffice to guarantee vanishing regret.
We leave open the question of whether
%It is unclear whether
it is
%feasible
possible to obtain vanishing regret simultaneously for all agents (with any bidding algorithms).

\subsection{Stochastic Gradient Descent (SGD) interpretation}
\label{sec:regret-SGD}

To prove \Cref{thm:regmain} it will be helpful to interpret \Cref{alg:bg} as using \emph{stochastic gradient descent} (SGD), a standard algorithm in online convex optimization (see \Cref{app:OCO}). Here, SGD uses pacing multipliers $\mu$ as actions, and optimizes an appropriately-defined artificial objective $H_t(\cdot)$ in each round $t$, where
\begin{equation}\label{eq:def-H}
\textstyle     H_t(\mu) = \rho\cdot \mu - \int_0^{\mu} Z_t(x)\mathrm{d}x.
\end{equation}
The per-round spend $z_t$ provides a stochastic signal that in expectation equals the gradient of $H_t$:
\begin{align}
H_t'(\mu_{t})
    = \rho-Z_t(\mu_{t}) = \rho - \mathbb{E}_{G_{t}} [z_t(\mu_{t})].
\end{align}

\noindent The SGD machinery applies because the function $H_t$ is convex and $(\overline{v}+\rho)$-Lipschitz (see \Cref{app:H_t}).

This interpretation provides a concrete intuition for what \Cref{alg:bg} actually does: for better or worse, it optimizes the aggregate artificial objective
    $\sum_{t\in[T]} H_t(\mu_t)$,
whose optimum is precisely $\sum_t H_t(\mu^*_t)$.
The analysis simplifies accordingly, as we can directly invoke the known guarantees for SGD and conclude that the artificial objective comes close to this target optimum. However, outside the well-studied special case of second-price auctions, this objective does not appear to have \emph{a priori} meaning. Much of the proof of \Cref{thm:regmain} thus involves relating this artificial objective to the objectives of interest, \emph{namely utility and value maximization}, in rather general settings.

\begin{remark}\label{rem:regret-benchmark}
The fact that \Cref{alg:bg} optimizes an artificial objective whose optimum is attained by the perfect pacing sequence suggests the latter as an appropriate benchmark for analyzing this algorithm.
\end{remark}

A similar (but technically different) interpretation appeared in \citep{BalseiroGur19} in the context of second-price auctions. Their interpretation relies on Lagrangian duality and does not appear to extend beyond second-price auctions.

\begin{remark}
Our analysis only uses two properties of the dynamics. First, that the dynamics obtains low regret
%(i.e. sublinear term $\mathsf{REG}_1$)
with respect to an adversarially generated sequence of convex functions $H_t$ (we inherit this from SGD).
Second, since the regret analysis does not account for early stopping from exhausting the budget, we require that the dynamics do not terminate too early (as given by \Cref{lem:bg1}) so as to bound any possible loss on such periods. The proof of \Cref{thm:regmain} generalizes almost immediately to any algorithm satisfying these properties.
\end{remark}

\subsection{Proof of \Cref{thm:regmain}}
\label{sec:proof.regmain}

Our analysis proceeds in three steps. First, we invoke the SGD machinery to show that the sequence of chosen multipliers $(\mu_t)$ approximates the target optimum $\sum_t H_t(\mu^*_t)$. Second, we relate the $H_t$ functions to the per-round payments to show that the expected payment is not too far from $\rho$ in each round. Finally, we use the MBB property of the auctions
%to relate the payments to the loss in value each round, thereby
to bound the total regret in either objective.

The following technical lemma will enable us to meaningfully relate the $Z_t$ and $H_t$ functions whose proof is deferred to the Appendix:
\begin{lemma}\label[lemma]{lem:lipint}
Let $f:\mathbb{R}\to \mathbb{R}$ be an increasing, $\lambda$-Lipschitz function such that $f(0)=0$. Let $R = \int_0^x f(y)\mathrm{d}y$ for some $x\in \mathbb{R}$. Then
%\begin{equation*}
    $\vert f(x)\vert\leq \sqrt{2\lambda R}$.
%\end{equation*}
\end{lemma}

To begin, observe that \Cref{alg:bg}
is equivalent to running (projected) stochastic gradient descent on the sequence of $H_t$ functions from \eqref{eq:def-H},
%above by the fundamental theorem of calculus,
where the maximum norm of the gradient is at most $\overline{v}$ using the assumption that the payment is at most the bid,
which in turn is at most the maximum possible value. Note that these functions are indeed convex and Lipschitz (see \Cref{app:H_t} for a proof)
%by \Cref{lem:H}
and indeed, $H_t'(\mu) = \rho-Z_t(\mu) = \rho - \mathbb{E}_{G_{t}} [z_t(\mu)]$ by construction of $H_t$.  We now wish to bound the sequence regret, $\sum_t \mathbb{E}[H_t(\mu_t) - H_t(\mu^*)]$. While the dynamics may end before time $T$, we may upper bound the regret by continuing to time $T$, as each term is nonnegative.
%using the fact that $\mu^*_t$ is an optimizer of $H_t$.
Standard analysis of SGD (see \Cref{thm:sgd} and \Cref{app:H_t} for details) then implies
%\bjl{It seems like we could generalize this to $(P/\eps + \eps T)$, and optimize $\eps$ given $P$.  Right?  If $P = T^\alpha$, we can choose $\eps = T^{-(1-\alpha)/2}$ to get regret $T^{(1 + \alpha)/2}$.}
%by  \Cref{prop:sgd}, we obtain
\begin{equation}
\label{eq:reginit}
\textstyle
    \mathsf{REG}_1\triangleq
    \sum_{t=1}^T\; \E\sbr{H_t(\mu_t)-H_t(\mu^*)}
    \leq O\rbr{\frac{P+1}{\eps}\cdot \overmu^2 +
         \eps T\cdot (\rho+\overline{v})^2}.
\end{equation}
% For convenience, we denote the right side of \Cref{eq:reginit} by $\mathsf{REG}_1$.

%\bjl{Changed the order of equations above; double-check that this is what is intended.} \jg{[This looked fine to me...]}

In the remainder of the proof, we will translate \eqref{eq:reginit} into a bound on pacing regret. It suffices to show the regret bound in \eqref{eq:thm-regmain} for pacing regret with respect to values and utility individually; the regret guarantee then translates to the unified objective for any choice of $\gamma\in [0,1]$. It will be technically convenient to perform the analysis for values first and then derive the same bound for utility as a simple variation.

From \Cref{lem:bg1}, we may assume that the dynamics continue until time $T$ without stopping at the cost of $O(\overline{v}\overmu\sqrt{T}/\rho)$ loss in value which we will account for at the end. Define the function $W_t(\mu)$ by
\begin{equation*}
    W_t(\mu) = \begin{cases}
    V_t(\mu) & \text{if $Z_t(\mu)<\rho$}\\
    V_t(\mu)\times \frac{\rho}{Z_t(\mu)} & \text{if $Z_t(\mu)\geq \rho$}.
    \end{cases}
\end{equation*}
Observe that $W_t(\mu)\leq V_t(\mu)$ by construction; it can be viewed as (approximately) the per-round value obtained by bidding with pacing multiplier $\mu$ --- if the environment at time $t$ persisted for all rounds --- until the budget is expected to terminate due to overspending. We further observe that $V_t(\mu_t)-V(\mu^*_t)\geq W_t(\mu_t)-W_t(\mu^*_t)$ for all $t$.
To see this, note that $V(\mu^*_t) = W_t(\mu^*_t)$ since $Z_t(\mu^*_t) \leq \rho$,
and moreover,
% these expressions are identical if $Z_t(\mu_t)<\rho$, and otherwise,
% the left side is nonnegative while the
the first term on the right side equals the first term of the left side, which is non-negative, with some scaling factor at most $1$.
Therefore,
\begin{equation}
\label{eq:VtLB}
\textstyle    \sum_{t=1}^T \left(V_t(\mu^*_t)-V(\mu_t)\right)\leq \sum_{t=1}^T \left(W_t(\mu^*_t)-W_t(\mu_t)\right),
\end{equation}
so it will suffice to show that the right side is not too large.

For any round $t$, we now consider two cases depending on the value of $\mu_t$ and derive a lower bound on the difference in $W_t$ values:
\begin{description}
    \item[\textbf{Case 1: $\mu_t<\mu^*_t$.}] In this case we must have $\mu^*_t > 0$ and hence $Z_t(\mu_t)\geq Z_t(\mu^*_t) = \rho$, and therefore
    \begin{equation*}
        W_t(\mu_t)=V_t(\mu_t)\times \frac{\rho}{Z_t(\mu_t)}\geq W_t(\mu^*_t)\times \frac{\rho}{Z_t(\mu_t)}.
        %= W_t(\mu^*_t)\times \frac{Z_t(\mu^*_t)}{Z_t(\mu_t)}.
    \end{equation*}
    By rearranging and adding $W_t(\mu^*_t)$ to each side of the last inequality, it follows that
    \begin{equation}
    \label{eq:regcase1}
        W_t(\mu^*_t)-W_t(\mu_t)\leq \frac{Z_t(\mu_t)-\rho}{Z_t(\mu_t)}\cdot W_t(\mu^*_t)\leq (Z_t(\mu_t)-\rho)\cdot \frac{\overline{v}}{\rho},
        %W_t(\mu^*_t)-W_t(\mu_t)\leq \left(\frac{Z_t(\mu_t)-Z_t(\mu^*_t)}{Z_t(\mu_t)}\right)\cdot W_t(\mu^*_t)\leq (Z_t(\mu_t)-Z_t(\mu^*_t))\cdot \frac{\overline{v}}{\rho},
    \end{equation}
    where we use the fact that $Z_t(\mu_t)\geq \rho$ as well as $W_t(\mu^*_t)\leq \overline{v}$.
    \item[\textbf{Case 2: $\mu_t\geq \mu^*_t$}.] In this case, we have $Z_t(\mu_t)\leq Z_t(\mu^*_t) \leq \rho$ by assumption, and hence $W_t(\mu^*_t) = V_t(\mu^*_t)$ and $W_t(\mu_t) = V_t(\mu_t)$.
    {
    We now claim that the MBB property of the auction implies
        \begin{equation}
        \label{eq:mbb}
%        \frac{W_t(\mu^*_t)}{\rho}=\frac{V_t(\mu^*_t)}{Z_t(\mu^*_t)}\geq \frac{V_t(\mu_t)}{Z_t(\mu_t)} = \frac{W_t(\mu_t)}{Z_t(\mu_t)}\iff \frac{W_t(\mu_t)}{W_t(\mu^*_t)}\geq \frac{Z_t(\mu_t)}{\rho}.
        \frac{W_t(\mu^*_t)}{W_t(\mu_t)}=\frac{V_t(\mu^*_t)}{V_t(\mu_t)}\leq \frac{Z_t(\mu^*_t)}{Z_t(\mu_t)}. %\leq \frac{\rho}{Z_t(\mu_t)}.
        %= \frac{W_t(\mu_t)}{Z_t(\mu_t)}\iff \frac{W_t(\mu_t)}{W_t(\mu^*_t)}\geq \frac{Z_t(\mu_t)}{\rho}.
    \end{equation}

    %Then \Cref{prop:MBB} implies \BLnew{I still don't see why Prop 2.1 implies this immediately.. note that the value of the agent $v_{k}$ is random now (does not appear in Prop 2.1) and can be also correlated to the dist. of bids of other players. }

\iffalse
    \begin{equation*}
%        \frac{W_t(\mu^*_t)}{\rho}=\frac{V_t(\mu^*_t)}{Z_t(\mu^*_t)}\geq \frac{V_t(\mu_t)}{Z_t(\mu_t)} = \frac{W_t(\mu_t)}{Z_t(\mu_t)}\iff \frac{W_t(\mu_t)}{W_t(\mu^*_t)}\geq \frac{Z_t(\mu_t)}{\rho}.
        \frac{W_t(\mu^*_t)}{W_t(\mu_t)}=\frac{V_t(\mu^*_t)}{V_t(\mu_t)}\leq \frac{Z_t(\mu^*_t)}{Z_t(\mu_t)}. %\leq \frac{\rho}{Z_t(\mu_t)}.
        %= \frac{W_t(\mu_t)}{Z_t(\mu_t)}\iff \frac{W_t(\mu_t)}{W_t(\mu^*_t)}\geq \frac{Z_t(\mu_t)}{\rho}.
    \end{equation*}
\fi

%    $\frac{Z_t(\mu_t)}{V_t(\mu_t)}\leq \frac{Z_t(\mu^*_t)}{V_t(\mu^*_t)}$.
    To prove \eqref{eq:mbb}, define $\gamma = Z_t(\mu_t)/V_t(\mu_t)$ and note that since the auction is individually rational, it must be that an agent bidding $v_t/(1+\mu_t)$ pays at most $v_t/(1+\mu_t)$ per unit allocated, and hence $\gamma \leq \frac{1}{1+\mu_t}$.  If we write $Z_t(\mu_t\ |\ v_t)$ for the expected value of $Z_t(\mu_t)$ conditional on the realization of $v_t$ (recalling that $v_t$ can be correlated with the competing bids), and similarly for $V_t(\mu_t\ |\ v_t)$, then by MBB and linearity of expectation
\[Z_t(\mu^*_t\ |\ v_t) - Z_t(\mu_t\ |\ v_t) \geq \frac{1}{1+\mu_t}\left( V_t(\mu^*_t\ |\ v_t) - V_t(\mu_t\ |\ v_t)\right)\]
where we used the fact that the bid given $v_t$ and $\mu_t$ is $v_t/(1+\mu_t)$, and $V_t(\cdot\ |\ v_t)$ is precisely $v_t$ times the agent's expected allocation.  Taking expectations over $v_t$ then implies $Z_t(\mu^*_t) - Z_t(\mu_t) \geq \frac{1}{1+\mu_t}\left( V_t(\mu^*_t) - V_t(\mu_t)\right)$.  Rearranging and recalling that $Z_t(\mu_t) = \gamma V_t(\mu_t)$, we conclude
\[ Z_t(\mu^*_t) \geq \gamma V_t(\mu_t) + \frac{1}{1+\mu_t}\left(V_t(\mu^*_t) - V_t(\mu_t)\right) \geq \gamma V_t(\mu^*_t)\]
which implies \eqref{eq:mbb}.}
    A rearrangement of \eqref{eq:mbb} then implies
    \begin{equation}
        \label{eq:regcase2}
%        W_t(\mu^*_t)-W_t(\mu_t)\leq W_t(\mu^*_t)\times \frac{\rho - Z_t(\mu_t)}{\rho}\leq \frac{\overline{v}}{\rho}\times (\rho-Z_t(\mu_t)).
        W_t(\mu^*_t)-W_t(\mu_t)\leq \frac{V_t(\mu^*_t)}{Z_t(\mu^*_t)}\times (Z_t(\mu^*_t) - Z_t(\mu_t))\leq \frac{\overline{v}}{{\min\{\rho,\delta\}}}(Z_t(\mu^*_t) - Z_t(\mu_t)).%\leq %\max\left\{C,\frac{\overline{v}}{\rho}\right\}(Z(\mu^*_t) - Z_t(\mu_t)).%\leq \frac{\overline{v}}{\rho}\times (\rho-Z_t(\mu_t)).
    \end{equation}
    The second inequality comes from considering the cases that $\mu^*_t=0$ (and applying \Cref{assn:boundedMBB}) and $\mu^*_t>0$.
\end{description}
In either case, \Cref{eq:regcase1} and \Cref{eq:regcase2} implies that {(recalling the assumption $\delta\leq \rho$)}
\begin{equation}
    \label{eq:W_tLB}
    W_t(\mu^*_t)-W_t(\mu_t)\leq \frac{\overline{v}}{{{\delta}}} \vert Z_t(\mu^*_t)-Z_t(\mu_t)\vert.
    %
    %W_t(\mu^*_t)-W_t(\mu_t)\leq \max\left\{C,\frac{\overline{v}}{\rho}\right\} \vert Z_t(\mu^*_t)-Z_t(\mu_t)\vert.
\end{equation}
We now relate this bound to the $H_t$ functions. Observe that
\begin{equation*}
    H_t(\mu_t)-H_t(\mu^*_t) = \int_{\mu^*_t}^{\mu_t} \left[\rho-Z_t(x)\right]\mathrm{d}x=\int_{0}^{\mu_t-\mu^*_t} \left[\rho-Z_t(\mu^*_t+x)\right]\mathrm{d}x,
\end{equation*}
and moreover that the integrand
%$f(x) = Z_t(\mu^*_t)-Z_t(\mu^*_t+x)=\rho-Z_t(\mu^*_t+x)$
$f(x) = \rho-Z_t(\mu^*_t+x)$
is increasing and $\lambda$-Lipschitz by the assumption on $Z_t$.
In general, $f(0)\geq 0$ (with equality if $Z_t(\mu^*_t)=\rho$). If we consider $g(x)=f(x)-f(0)$ and apply \Cref{lem:lipint}, we obtain
\begin{equation*}
    \vert Z_t(\mu_t) - Z_t(\mu^*_t)\vert =\vert f(\mu_t-\mu^*_t)-f(0)\vert\leq \sqrt{2\lambda (H_t(\mu_t)-H_t(\mu^*_t)-(\mu_t-\mu^*_t)(\rho-Z_t(\mu^*_t)))}.
\end{equation*}
Note that if $\mu_t<\mu^*_t$, then $\rho=Z_t(\mu^*_t)$ so the last subtracted term is zero, while if $\mu_t\geq \mu^*_t$, then the last subtracted term is nonnegative. In either case, we deduce that
\begin{equation}
\label{eq:ZtUB}
    \vert Z_t(\mu_t) - Z_t(\mu^*_t)\vert\leq \sqrt{2\lambda (H_t(\mu_t)-H_t(\mu^*_t))}
\end{equation}

By combining these inequalities, we obtain:
\begin{align*}
    \textstyle
    \E\left[\sum_{t=1}^T V_t(\mu^*_t)-V(\mu_t)\right]
    &\leq \textstyle
    \E\left[\sum_{t=1}^T W_t(\mu^*_t)-W_t(\mu_t)\right] &&\text{(by \Cref{eq:VtLB})}\\
    &\leq \frac{\overline{v}}{{\delta}}
    \E\left[\textstyle\sum_{t=1}^T \left\vert Z_t(\mu^*_t)-Z_t(\mu_t)\right\vert\right] &&\text{(by \Cref{eq:W_tLB})}\\
    &\leq \sqrt{2\lambda}\frac{\overline{v}}{{\delta}}\E\left[\textstyle\sum_{t=1}^T \sqrt{ H_t(\mu_t)-H_t(\mu^*_t)}\right] &&\text{(by \Cref{eq:ZtUB})}\\
    &\leq \sqrt{2\lambda T}\frac{\overline{v}}{{\delta}}\E\left[\sqrt{\textstyle\sum_{t=1}^T  H_t(\mu_t)-H_t(\mu^*_t)}\right] &&\text{(by Cauchy-Schwarz)}\\
    &\leq \sqrt{2\lambda T}\frac{\overline{v}}{{\delta}}\sqrt{\E\left[\textstyle\sum_{t=1}^T  H_t(\mu_t)-H_t(\mu^*_t)\right]} &&\text{(by Jensen's inequality)}\\
    &\leq \frac{\overline{v}}{{\delta}}\sqrt{2\lambda T\cdot\mathsf{REG}_1} &&\text{(by \Cref{eq:reginit})}.
\end{align*}
Accounting for the $O(\overmu\overline{v}/\eps\rho)$ loss from possibly terminating early using \Cref{lem:bg1} establishes the claimed regret bound of \eqref{eq:thm-regmain} for regret with respect to value maximization.

We now show how to adapt the preceding analysis to bound the pacing regret for utility-maximization as a nearly direct consequence. By inspection of the proof above,
%in the case of $\PaceReg^V(T)$,
 it suffices to show that the following direct analogue of \Cref{eq:W_tLB} holds for utilities:
    \begin{equation}\label{eq:U_tLB}
        U_t(\mu^*_t)-U_t(\mu_t)\leq \frac{\overline{v}}{{{\delta}}} \vert Z_t(\mu^*_t)-Z_t(\mu_t)\vert.
    \end{equation}
    Given \Cref{eq:U_tLB}, the remainder of the argument is exactly the same. To see that \Cref{eq:U_tLB} holds, we again consider cases based on the value of $\mu_t$:

    \begin{description}
    \item[\textbf{Case 1: $\mu_t<\mu^*_t$.}] In this case, we immediately have
    \begin{align*}
        U_t(\mu^*_t)-U_t(\mu_t)&=\underbrace{(V_t(\mu_t^*)-V_t(\mu_t))}_{\leq 0}+(Z_t(\mu_t)-Z_t(\mu_t^*))\\
        &\leq \underbrace{Z_t(\mu_t)-Z_t(\mu_t^*)}_{\geq 0}
        =\vert Z_t(\mu_t)-Z_t(\mu_t^*)\vert.
    \end{align*}
    The first inequality follows from the fact obtained valuations are non-increasing in $\mu$, while the last equality holds because expenditures are also non-increasing in $\mu$.

    \item[\textbf{Case 2: $\mu_t\geq\mu^*_t$.}] In this case, we can instead bound
    \begin{align*}
        U_t(\mu^*_t)-U_t(\mu_t)&= (V_t(\mu_t^*)-V_t(\mu_t))+\underbrace{(Z_t(\mu_t)-Z_t(\mu_t^*))}_{\leq 0}
        \leq V_t(\mu_t^*)-V_t(\mu_t).
    \end{align*}
    This holds again because $Z_t(\mu_t)\leq Z_t(\mu_t^*)$ in this case by monotonicity of expenditure.
    But in this case, we recall that by construction in the proof of \Cref{thm:regmain} that $V_t(\mu_t^*)-V_t(\mu_t)\leq W_t(\mu_t^*)-W_t(\mu_t)$. So, the previous display combined with the bound for $W_t$ in \Cref{eq:regcase2} immediately implies
    \begin{equation*}
        U_t(\mu^*_t)-U_t(\mu_t)\leq \frac{\overline{v}}{{{\delta}}} \vert Z_t(\mu^*_t)-Z_t(\mu_t)\vert.
    \end{equation*}
    \end{description}

    Putting these two cases establishes \Cref{eq:U_tLB} (recall the assumption that $\overline{v}\geq 1$) and thus completes the proof for pacing regret for utilities, and therefore the proof of \Cref{thm:regmain}.

\section{Aggregate and Individual Guarantees for Other Algorithms}
\label{sec:other}
%\section{Liquid Welfare for Optimistic Generalized Pacing}
%\label{app:extensions_poa}

It is natural to ask whether other online bidding algorithms achieve similar liquid-welfare guarantees, or in fact a similar \emph{combination} of aggregate and individual guarantees. A natural no-regret condition alone cannot suffice for liquid welfare: we provide a simple example in  \Cref{app:regret.example}. In this section, we extend our analyses and guarantees to several online bidding algorithms that replace the \emph{stochastic gradient descent (SGD)} update step in \Cref{alg:bg} step with an update step from another well-known algorithm for online convex optimization (OCO). Specifically, we invoke \emph{optimistic gradient descent (OGD)}
\citep{Rakhlin-colt13,Rakhlin-nips13,mokhtari-aistats20,Jiang-siam25}, \emph{optimistic mirror descent} (OMD) \citep{Chiang-colt12,Rakhlin-colt13,Rakhlin-nips13}, and \emph{optimistic follow-the-regularized-leader} (OFTRL) \cite{Rakhlin-colt13}.

The new algorithms are defined as follows. Throughout, recall that $\epsilon_k> 0$ is the step-size, and let $\myGrad = \rho_k-z_{k,t}$ denote the estimated gradient in round $t$. In \PacingOMD and \PacingOFTRL, in each round $t$, an algorithm is given an estimate $M_{k,t}$ of the \emph{next} gradient. Thus:
\begin{description}
\item[\PacingOGD] Replace \refeq{eq:alg-update} in \Cref{alg:bg} with an OGD step: for some fixed $\eps'_k \in [-\eps_k/2,0]$,
\begin{align*}%\label{eq:PacingOGD}
 \mu_{k,t+1} \leftarrow P_{[0,\overline{\mu}]}
    \rbr{\mu_{k,t}-\eps_k\cdot\myGrad - \eps'_k\cdot\myGrad[k,t-1]}.
\end{align*}

\item[\PacingOMD] Replace \refeq{eq:alg-update} in \Cref{alg:bg} with an OMD step: initializing $\widehat{\mu}_{k,1}=0$,
\begin{align*}%\label{eq:PacingOMD}
 \widehat{\mu}_{k,t+1} &\leftarrow
    P_{[0,\overline{\mu}]}
        \rbr{\widehat{\mu}_{k,t}-\eps_k\cdot\myGrad},\\
 \mu_{k,t+1} &\leftarrow
    P_{[0,\overline{\mu}]}
        \rbr{\widehat{\mu}_{k,t+1}-\eps_k \cdot M_{k,t}}.
\end{align*}

\item[\PacingOFTRL] Replace \refeq{eq:alg-update} in \Cref{alg:bg} with an OFTRL step with Euclidean regularizer:
\begin{align*}%\label{eq:PacingOFTRL}
 \mu_{k,t+1} \leftarrow
    P_{[0,\overline{\mu}]}
        \rbr{-\eps_k\rbr{M_{k,t}+{\textstyle \sum_{s\in[t]}}\; \myGrad[k,s]}}.
\end{align*}
%where $\eps_k'\in [-\eps_k, \eps_k]$ is another parameter.
%\ascomment{Jason: which range gives the claimed guarantees?}\jgcomment{It seems to me that we can take any smaller pacer. It doesn't really matter in the regret bounds so long is has the same step-size order. This is Lemma 2 of RS13.}
\end{description}

\noindent In line with this notation, one can refer to \Cref{alg:bg} as \PacingSGD.

\begin{remark}
In \PacingOGD, the case $\eps'=0$ corresponds to the original algorithm (\Cref{alg:bg}). The case  $\eps'_k = -\eps_k/2$ corresponds to OGD \citep{Daskalakis-iclr18,Daskalakis-neurips18-limitPoints}, another OCO algorithm that obtains faster convergence in certain learning-in-games settings. Further work \citep{mokhtari-aistats20,Jiang-siam25} extends OGD to an arbitrary $\eps'_k\in [-\eps/2,0]$.
\end{remark}

\begin{remark}
While the original versions of OMD and OFTRL with Euclidean regularizer are defined for $\mK\subseteq \R^d$, $d\in\N$ and use Bregman divergences, we invoke a simpler  equivalent formulation for $d=1$.
\end{remark}

%OMD was introduced \citet{Chiang-colt12} and broadly extended in \citet{Rakhlin-colt13}. OFTRL was introduced in \citet{Rakhlin-colt13}.

\begin{remark}\label[remark]{rem:predictable}
OMD and OFTRL achieve the same regret rates as (we invoke for) SGD for the default choice of the next-gradient predictors, $M_{k,t} = \myGrad$. They achieve \emph{better} regret rates when the gradient sequences are predictable, in some formal sense, and this is the key advantage of these algorithms \citep{Chiang-colt12,Rakhlin-colt13}.

Our analysis is consistent with this intuition. Our liquid-welfare guarantees for \PacingOMD and \PacingOFTRL apply to arbitrary $M_{k,t}$'s, whereas our regret bounds for these algorithms are for $M_{k,t} = \myGrad$. Moreover, our regret analysis extends to any other choice of $M_{k,t}$'s that yields same or better regret bound for the underlying OCO algorithms. Better OCO regret bounds would translate into improved regret bounds for online bidding (but we do not spell this out explicitly).
\end{remark}

\xhdr{Regret bounds.}
All results in \Cref{sec:regret} carry over to \PacingOGD and \PacingOMD (with $M_{k,t} = \myGrad$). This is because OGD and OMD with $M_{k,t} = \myGrad$ satisfy the same OCO regret bound as SGD (\Cref{thm:sgd}), and this regret bound plugs into \Cref{sec:regret} same way as it does for \Cref{alg:bg}. Moreover, we obtain an analog of \Cref{lem:bg1}, showing that the algorithms do not run out of budget too early (see \Cref{lem:other-runout}). These are the only facts about an algorithm invoked by our analysis in \Cref{sec:regret}.

For \PacingOFTRL with $M_{k,t} = \myGrad$, we only obtain regret bounds for the stochastic environment (\Cref{cor:regret-stochastic-gen,cor:regret-stochastic} in \Cref{sec:regret-stochastic}). We obtain an analog of \Cref{lem:bg1} (see \Cref{lem:other-runout}). Then we invoke the regret bound for OFTRL (\Cref{thm:sgd}, in the special case when all $f_t$'s are the same), and observe that it plugs into the analysis in \Cref{sec:regret} for the stochastic environment and implies \Cref{cor:regret-stochastic-gen}.

Since the results from \Cref{sec:regret} carry over as stated, we do not restate them here. The only change is a slightly stronger condition on the problem parameters:
    $\overline{\mu}\geq2\overline{v}/\rho_k+1$,
which stems from \Cref{lem:other-runout}. The relevant background on OCO algorithms is summarized in \Cref{app:OCO}.

\xhdr{Liquid welfare.} We obtain essentially the same liquid-welfare guarantee as in \Cref{sec:liquid.welfare}, under a mild additional assumption that
$\max_k (\eps_k/\rho_k) = o(1)$.

\begin{theorem}\label{thm:other}
Fix any core auction and any distribution $F$ over agent value profiles. Suppose that each agent $k$ uses \PacingOGD, \PacingOMD, or \PacingOFTRL. (Different agents can use different algorithms and/or different step-sizes $\eps_k$.) The next-gradient predictors $M_{k,t}$ can be arbitrary, as long as $|M_{k,t}|\leq \bar{v}$.  Write $\boldsym{x}$ for the corresponding allocation sequence rule. Then for any allocation rule $\boldsym{y} \colon [0,\overline{v}]^{n} \to X$ and $c = 4\,\overline{v}\cdot \max_k \sqrt{\eps_k/\rho_k}$
we have
\begin{align}\label{eq:thm:other}
W(\boldsym{x},F)
\geq (1-c)\cdot \frac{\overline{W}(\boldsym{y},F)}{2} - O\rbr{n\overline{v}\sqrt{T\log(\overline{v}nT)}}.
\end{align}
\end{theorem}

\subsection{The General Algorithm Property: Proof Sketch of \Cref{thm:other}}

We obtain \Cref{thm:other} by extending the analysis in \Cref{sec:liquid.welfare}. In this analysis, the dependence on the algorithm's details is essentially captured by the following property: the amount spent over a sequence of rounds with non-zero pacing is lower-bounded by a quantity that depends only on the starting and ending bids, but otherwise not the learning path in between (Lemma~\ref{lem:stronger}, which is itself a corollary of \refeq{eq:payment.dynamics}). In what follows, we define a more general property, $c$-event-feasibility, that captures all three algorithms in \Cref{thm:other}, and prove \refeq{eq:thm:other} when all agents use algorithms that are $c$-event-feasible.

%The general algorithm property alluded to above is captured by the following definition.

\begin{definition}[$c$-event-feasible]\label[definition]{def:other-general}
Consider a variant of \Cref{alg:bg} in which the multiplier update in \refeq{eq:alg-update} is replaced with some other update rule which computes $\mu_{k,t}\in [0,\overline{\mu}]$. The algorithm is called \emph{$c$-event-feasible}, $c\in(0,1)$ if for each round $t$ before it runs out of budget, there exists an event $\mE_{k,t}$ (determined by the auction history up to this round) such that the following two properties hold:
\begin{align}
&\sum_{t\in[T]} x_{k,t}v_{k,t}\geq \sum_{t\in[T]} (x_{k,t}v_{k,t}-z_{k,t})\cdot \mathbf{1}(\mE_{k,t})+(1-c)\sum_{t\in[T]}{\rho_k} \cdot \mathbf{1}(\mE_{t,k}^c).
    \\
&\text{Event $\mE_{k,t}$ implies that }
        b_{k,t}=\frac{v_{k,t}}{1+\mu_{k,t}}\geq (1-c)\cdot v_{k,t}.
\end{align}
\end{definition}

\Cref{alg:bg} is event-feasible by \Cref{lem:stronger}, with
    $\mE_{k,t}=\cbr{\mu_{k,t}=0}$ and $c=0$.
The definition instantiates to the other algorithms as follows:

\begin{lemma}\label[lemma]{lem:other-prop}
\PacingOGD, \PacingOMD, and \PacingOFTRL are $c$-event feasible with
    $c = 4\,\overline{v}\sqrt{\eps_k/\rho_k}$.
The events $\mE_{k,t}$ are
    $\cbr{\mu_{k,t}\leq c}$
for \PacingOGD and \PacingOFTRL, and
    $\cbr{\widehat{\mu}_{k,t}\leq c}$
for \PacingOMD.
\end{lemma}

When each agents uses an algorithm from \Cref{def:other-general}, liquid welfare is as follows.

\begin{theorem}\label{thm:other-general}
Fix any core auction and any distribution $F$ over agent value profiles. Suppose each agent uses a $c$-event-feasible algorithm, for some $c\in(0,1)$ (possibly a different algorithm for different agents). Let $\boldsym{x}$ be the corresponding allocation sequence rule. Then for any allocation rule $\boldsym{y} \colon [0,\overline{v}]^{n} \to X$,
\begin{align}\label{eq:thm:other-general}
W(\boldsym{x},F)
\geq (1-c)\cdot \frac{\overline{W}(\boldsym{y},F)}{2} - O\rbr{n\overline{v}\sqrt{T\log(\overline{v}nT)}}.
\end{align}
\end{theorem}

The proof follows that of \Cref{thm:main.new} with relatively minor modifications. We lower bound the liquid welfare using a similar charging argument, but tailored to these new events. The agents satisfying $\mathcal{E}_{k,t}$ at any time $t$ are approximately bidding their value, and so the core auction properties again ensure that any valuation loss can be charged to prices. \Cref{lem:other-prop} and \Cref{thm:other-general} are proved in \Cref{app:other}.

\section{Numerical Evaluation}
\label{sec:expts}

%Prior work has established analytically that Algorithm~\ref{alg:bg} guarantees vanishing regret, versus the best fixed policy in hindsight, in stationary stochastic environments~\cite{balseiro2021budget}.  However, as discussed previously, no such regret bounds are known for the (non-stationary) setting of simultaneous learning by competing bidders.

%In this section we provide evidence for vanishing regret via numerical simulation.  We conduct a numerical study based on campaign and bidding data collected from the Bing Advertising platform.

In this Section we complement our theoretical findings with a numerical simulation study of \cref{alg:bg}.  We consider the \emph{multi-player} environment with simultaneous learning by competing bidders. Our simulations are semi-synthetic, based on campaign and bidding data collected from the Bing Advertising platform.

We focus on regret relative to the standard benchmark: the best fixed pacing multiplier in hindsight. Recall that, relative to this benchmark, vanishing regret is achievable for a stochastic environment but is provably impossible for an adversarial environment.  The achievability of vanishing regret in a multi-player environment remains unknown and is currently an open question.

This gap motivates our simulation study, which focuses on this multi-player environment with simultaneous learning.  Our simulations suggest that simultaneous execution of \cref{alg:bg} yields vanishing regret, with regret rate $O(T^{\alpha})$ for $\alpha \leq 3/5$.
We also compare the performance of \cref{alg:bg} with some other approaches to online bid optimization in the literature, analyzing both empirical regret rate and liquid welfare outcomes on our semi-synthetic dataset.

\subsection{Data and Simulation Description}

Our data consists of campaign and auction data collected over a 7-day period in April 2022.  The dataset contains daily budget targets and other campaign parameters (such as maximum bid, when specified) for campaigns from a North American advertising segment.  It also contains, for a random subsample of {$N\approx 2.4M$}
%$N=2,398,226$ impressions
auction instances over this period, the list of participating campaigns, click probability predictions for each participant, realized bids from each participant, and auction outcomes (including auction winner and payment).  All monetary measurements are expressed in normalized units.

Given this dataset, we simulate a joint execution of \Cref{alg:bg} as follows. For any campaign that participates in fewer than $\theta = 1000$ auction instances, we maintain the bid that they originally placed in the dataset.  For all other campaigns, we generate new bids by simulating the execution of an online bidding algorithm for each participant of each auction instance.  The target spend rate for campaign $k$ is calculated by taking $T_k$ to be the average number of per-day auction instances for each campaign, and per-impression values are taken to be proportional to platform estimated click rates.\footnote{In all visualizations and reported welfare metrics, exact values and counts are renormalized to prevent data leakage.}

Given the bids for all auction participants, we simulate the auction outcome using estimated click rates and a simplified auction rule.
%Each bidding agent's bid is calculated as in Algorithm~\ref{alg:bg}, using the estimated click rate as the value.
We take there to be only a single winner in each auction instance, corresponding to the highest total bid.  Payments are calculated according to a specified payment rule; we consider both first-price and second-price payment rules in our experiments.

\subsection{Bidding Algorithms}

%\bjlcomment{Everything in this subsection is new}

Here we describe the online bidding algorithms evaluated in our simulations.%
\footnote{In addition to \Cref{alg:bg}, we have chosen as comparators two ``advanced" versions of SGD (incorporating, resp., the ``optimistic correction" $G_{k,t} - G_{k,t-1}$ and weighted gradient-averaging) and a multiplicative approach from an important early work on simultaneous learning in auctions~\cite{Borgs-www07}.
%on the subject.\footnote{Recall that MU attains provable convergence guarantees for simultaneous learning in first-price auctions~\cite{Borgs-www07}.}
A more detailed empirical investigation of other variations is left for future studies.}  All algorithms adjust a bid multiplier that is constrained to lie in  $[0,\overline{\mu}]$. {All algorithms invoke some gradient-based technique from the literature, treating $G_{k,t} = \rho_k - z_{k,t}$ as a gradient, for each agent $k$ and round $t$. (For simplicity, we use the technique's name to label the algorithm.)} All meta-parameters are tuned via grid search.  We provide a warm start by initializing each campaign's multiplier to be proportional to $1/\rho_k$.

\begin{itemize}
    \item \textbf{Stochastic Gradient Descent (SGD):} This is \Cref{alg:bg}, the online bidding method described in Section~\ref{sec:model} and analyzed theoretically in this paper.
        %\asdelete{The gradient for agent $k$ in step $t$ is $G_{k,t} = (\rho_k - z_{k,t})$, and}
        The multiplier is updated according to the rule $\mu_{k,t+1} = \mu_t - \epsilon_k G_{k,t}$. The learning rate $\epsilon_k$ for each agent $k$ is chosen to be proportional to $1/\sqrt{T_k}$ with a tuneable coefficient as meta-parameter.  %We provide a warm start by initializing each campaign's multiplier to be proportional to $1/\rho_k$.
    \item \textbf{Optimistic Gradient Descent (OGD):} an instantiation of \PacingOGD where the multiplier update is $\mu_{k,t+1} = \mu_t - \epsilon_k (2G_{k,t} - G_{k,t-1})$.  The learning rate is again chosen to be proportional to $1/\sqrt{T_k}$, with coefficient treated as a tunable meta-parameter.
    \item \textbf{Adaptive Moment Estimation (Adam):} a self-tuning gradient descent method in which the update in each round is a weighted average of prior gradients, scaled by a weighted $\ell_2$-norm of prior gradients.  Learning rate is treated as a tunable meta-parameter.
    \item \textbf{Multiplicative Updating (MU):} a bid update method proposed in~\cite{Borgs-www07}, in which $\mu_{k,t}$ is scaled up or down by a fixed multiplicative factor $(1+\epsilon)$ depending on whether the previous-round's spend was above or below the target spend $\rho_k$.  The factor $\epsilon$ is a tunable meta-parameter.
\end{itemize}

%\bjlcomment{Got to here}

%\begin{remark}
%\bjledit{In addition to \Cref{alg:bg}, we have chosen as comparators two ``advanced" versions of SGD (incorporating, resp., the ``optimistic correction" $G_{k,t} - G_{k,t-1}$ and weighted gradient-averaging) and a multiplicative approach from an important early work on simultaneous learning in auctions~\cite{Borgs-www07}.
%on the subject.\footnote{Recall that MU attains provable convergence guarantees for simultaneous learning in first-price auctions~\cite{Borgs-www07}.}
%A more detailed empirical investigation of other variations is left for future studies.}
%\end{remark}

\subsection{Regret Analysis}

\iffalse
\begin{figure}[t]
    \centering
    \begin{tabular}{cc}
        \includegraphics[width=0.3\textwidth]{plot_traces_spa.png}
        &
        \includegraphics[width=0.3\textwidth]{plot_traces_fpa.png} \\
        (a) & (b) \\
    \end{tabular}
    \caption{Illustration of estimated regret in repeated auction simulation.  For each of 100 randomly-sampled campaigns, we trace out the evolution of regret as the time horizon (the number of auction instances) is amplified, illustrated in log-log scale.  Each line corresponds to a single campaign.  Results are shown for (a) second-price auctions and (b)  first-price auctions. Instances with negative regret are excluded.}
    \label{fig:regret}
\end{figure}
\fi

\begin{figure}[t]
    \centering
    \begin{tabular}{ccc}
        \includegraphics[width=0.3\textwidth]{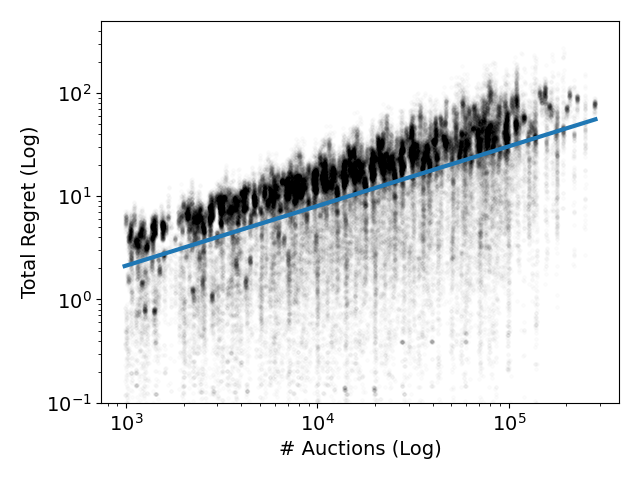}
        &
        \includegraphics[width=0.3\textwidth]{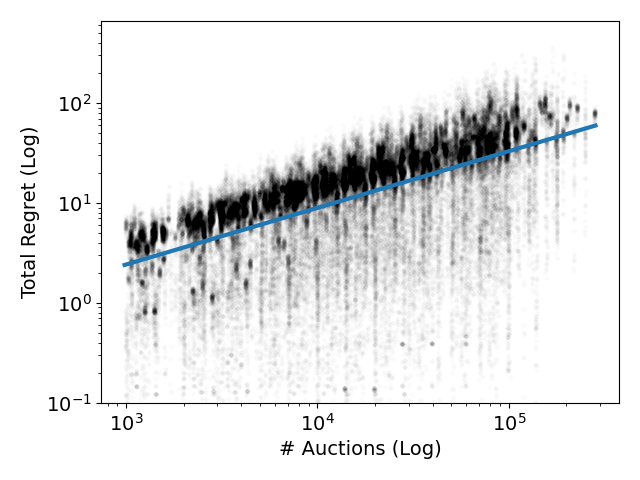}
        &
        \includegraphics[width=0.3\textwidth]{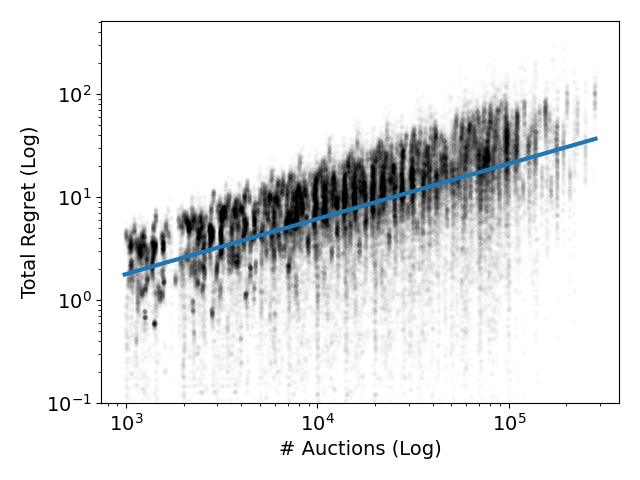} \\
        (a) & (b) & (c) \\
    \end{tabular}
    \caption{Regressing on observed regret in log-log scale to estimate regret rate for simultaneous multi-agent use of Algorithm~\ref{alg:bg}. Each dot plots the total regret of a simulated campaign against the number of auction instances in which it participates. Standard errors in the linear regression are included but not visible.  Plots are for (a) the utility objective in second-price auctions, (b) the value objective in second-price auctions, and (c) the value objective in first-price auctions.}
    \label{fig:regret.regression}
\end{figure}

%\bjlcomment{TODO: add a note about utility vs value objectives}

Following each execution of the simulation, we calculate the total regret for each bidding agent by determining the single best linear policy multiplier in hindsight.  This is done via binary search, up to an error of $10^{-9}$.  We run $n=100$ simulations for each scenario, each using a random subsample of $95\%$ of auction instances.  This provides an estimate of the regret of each agent $k$.

For second-price auctions, we consider both utility maximization and value maximization as objectives for the purpose of calculating regret, recalling that our benchmark of the best fixed pacing multiplier in hindsight is optimal for both objectives. For first-price auctions we consider the value maximization objective, since the best fixed pacing multiplier in hindsight is optimal for value maximization but not for utility maximization.

To estimate the rate of regret growth for each agent over time, we simulate longer time periods by iterating over the collection of auction instances in our dataset multiple times.  We simulated $K$ iterations over the dataset, where $K \in \{1, 2, 5, 10, 20, 50, 100\}$.  To reduce the potential impact of cyclic data patterns, each of the $K$ iterations includes only a subsample of the impression auctions drawn independently at random, where each impression is included in each iteration with probability $\beta = 1/2$.
%drawing an independent subsample of the impression auctions on each iteration and scaling budgets appropriately.
We then calculate
regret for each agent, as described above, as the number of auction instances grows.
%\ascomment{Brendan: what do you mean by "inflation factors"?}

\subsection{Results}

%Figure~\ref{fig:regret} illustrates the outcome of our simulation for Algorithm~\ref{alg:alg1} on a random sample of 100 of the simulated campaigns.  The evolution of regret is plotted against the increase in the number of auction instances, in log-log scale.  Both first-price and second-price simulation results are shown.  We note that regret increases approximately linearly in log-log scale for many advertisers, informally suggesting a polynomial regret rate.

Figure~\ref{fig:regret.regression} illustrates the total regret obtained in our simulation of Algorithm~\ref{alg:bg} over all simulated campaigns.  Each point on the plot corresponds to a combination of campaign (i.e., bidding agent representing an advertiser) and choice of $K$, and plots normalized auction count against total regret metric.  Further visualizations of per-campaign regret evolution are provided in \Cref{app:expts}.  Hypothesizing a regret rate of the form $O( T^{\alpha})$, we estimate $\alpha$ using log-log regression. Figure~\ref{fig:regret.regression} illustrates the outcome of log-log regression for all campaigns and simulation treatments, separately for first-price and second-price auction formats. We obtain an estimated slope of $\alpha = 0.573$ for second-price auctions (standard error 0.0026) and $\alpha = 0.540$ for first-price auctions (standard error 0.0028).

\begin{figure}[t]
    \centering
    \begin{tabular}{ccc}
        \includegraphics[width=0.3\textwidth]{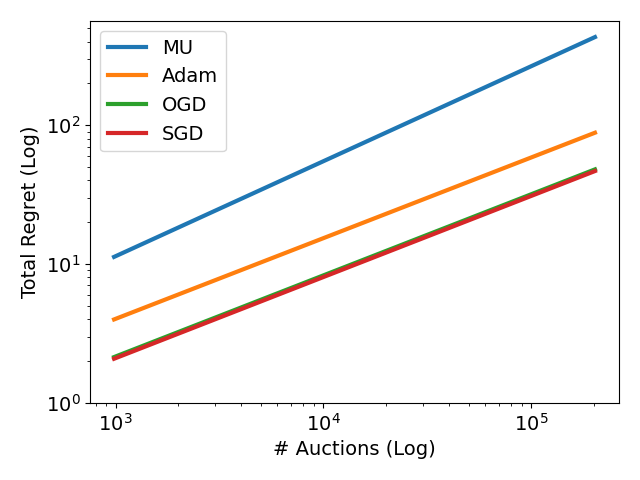}
        &
        \includegraphics[width=0.3\textwidth]{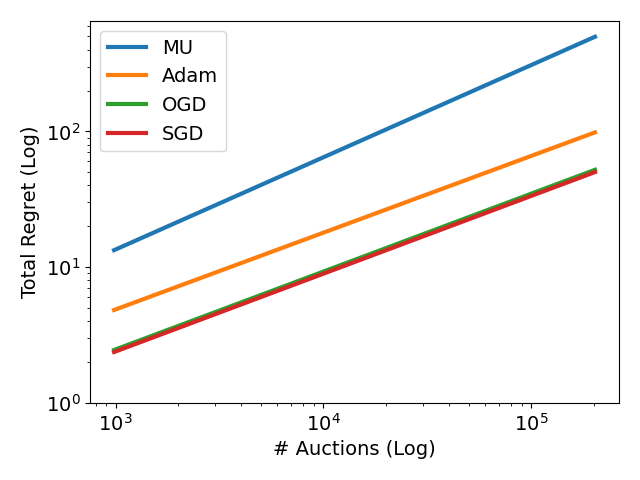}
        &
        \includegraphics[width=0.3\textwidth]{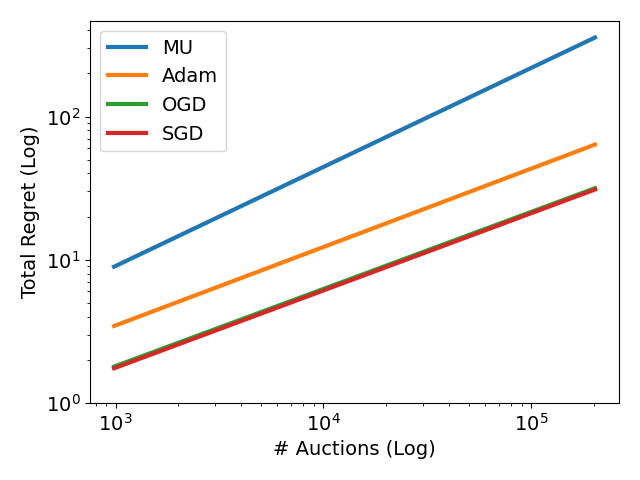} \\
        (a) & (b) & (c) \\
    \end{tabular}
    \caption{Visual comparison of regressed regret, in log-log scale, as a function of the number of auctions, for different learning algorithms.  Plots are for (a) the utility objective in second-price auctions, (b) the value objective in second-price auctions, and (c) the value objective in first-price auctions.  Note that OGD is partially obscured due to proximity to SGD.}
    \label{fig:regret.comparison}
\end{figure}

We repeat this estimation exercise for all algorithms in our comparison set.  Figure~\ref{fig:regret.comparison} plots the resulting regression for first-price and second-price auctions, with parameter estimates listed in Table~\ref{table:regret.comparison}.  Table~\ref{table:regret.comparison} also lists the total liquid welfare generated over the simulated campaigns for each algorithm and auction format.  For liquid welfare comparisons, we ran at $n=1000$ simulations to test for statistical significance in comparisons between algorithms. %\bjlcomment{TODO: stderr for welfare in table}
The top entries in each category, at a statistical significance of at least $p=0.05$, are displayed in bold.

We find that SGD and OGD have comparable performance with no statistically significant differences in terms of regret or liquid welfare. Relative to them, Adam has an improved regret rate for second-price auctions and a slightly worse regret rate for first-price auctions, though with higher absolute regret and lower liquid welfare for the auction counts covered in our dataset.  This finding aligns with the intuition that advanced self-tuning methods like Adam are most effective when the number of iterations is very large, whereas the generally slower convergence time that comes with self-tuning may come at a loss for campaigns that activate less frequently.  Multiplicative updating noticeably under-performs the other methods in both aggregate liquid welfare and regret.  We emphasize that, in all cases, hyperparameters for each algorithm were tuned individually and separately for both first-price and second-price auctions.

\begin{table}[h!]
\centering
\begin{tabular}{l|ccc|cc}
\toprule
 & \multicolumn{3}{c|}{SPA} & \multicolumn{2}{c}{FPA} \\
\cmidrule(lr){2-4} \cmidrule(lr){5-6} Algorithm & \shortstack{Regret Rate\\(Utility)} & \shortstack{Regret Rate\\(Value)} & \shortstack{Liquid Welfare\\\, } & \shortstack{Regret Rate\\(Value)} & \shortstack{Liquid Welfare\\\, } \\
\midrule
SGD & \shortstack{\textbf{0.584}\\(0.003)} & \shortstack{0.573\\(0.003)} &
\shortstack{\textbf{22.84}\\(2.25)} & \shortstack{\textbf{0.540}\\(0.003)} & \shortstack{\textbf{23.42}\\(2.24)} \\
\midrule
OGD & \shortstack{\textbf{0.584}\\(0.003)} & \shortstack{0.574\\(0.003)} &
\shortstack{\textbf{22.84}\\(2.25)} & \shortstack{\textbf{0.538}\\(0.003)} & \shortstack{\textbf{23.41}\\(2.24)} \\
\midrule
Adam & \shortstack{\textbf{0.581}\\(0.003)} &
\shortstack{\textbf{0.565}\\(0.003)} &
\shortstack{22.65\\(2.22)} & \shortstack{0.547\\(0.004)} & \shortstack{23.15\\(2.19)} \\
\midrule
MU & \shortstack{0.684\\(0.003)} &
\shortstack{0.679\\(0.003)} & \shortstack{22.60\\(2.22)} & \shortstack{0.692\\(0.003)} & \shortstack{23.14\\(2.22)} \\
\bottomrule
\end{tabular}
\caption{Comparison of algorithms under first- and second-price auction formats (resp., FPA and SPA). Estimated (regressed) regret rates are provided (i.e., the estimated $\alpha$ in regret rate of the form $O(T^\alpha)$) as well as the average liquid welfare (in normalized units) over all simulated campaigns.  Standard errors provided in parentheses.
Top algorithms in each category (up to statistical significance at $p=0.05$) are represented in bold.}
\label{table:regret.comparison}
\end{table}

%Hypothesizing a regret rate of the form $O( T^{\alpha})$, we estimate $\alpha$ using log-log regression. Figure~\ref{fig:regret.regression} illustrates the outcome of log-log regression for all campaigns and simulation treatments, separately for first-price and second-price auction formats. We obtain an estimated slope of $\alpha = 0.58$ for second-price auctions (standard error 0.0026) and $\alpha = 0.55$ for first-price auctions (standard error 0.0028).

%Taken together, these results provide evidence that  simultaneous execution of \cref{alg:bg} yields vanishing regret, with regret rate $< T^{3/5}$.

%We view this as strong evidence that, on empirically observed auction data, simultaneous execution of Algorithm~\ref{alg:bg} yields vanishing regret.  Moreover these simulations are highly suggestive that the regret rate is less than $O(T^{2/3})$.

%\subsection{Path Length}
%
%\bjl{Include a similar analysis studying empirically observed path length.  Evidence seems to suggest $P = o(T^{3/4})$.}

\newpage
\section{Conclusions and Open Questions}
\label{sec:conclusions}

%The scope of this paper is budget-constrained online bidding in repeated auctions.

The purpose of this paper is to simultaneously guarantee high aggregate welfare and low individual regret for budget-constrained online bidding in a repeated auction, without relying on convergence to an equilibrium. We establish these guarantees for a wide class of auction rules, arbitrarily correlated private values, bandit feedback, and several natural budget-pacing algorithms. Our individual guarantees hold for both utility- and value-maximization.

On a more technical level, the main result is a guarantee on the expected total liquid welfare achieved over multiple auction rounds. This approximation guarantee matches the best possible for a pure Nash equilibrium in a static truthful auction, but is the first to hold without requiring convergence to equilibrium. Our individual guarantees are of independent interest: we obtain the first non-trivial regret bounds for budget-constrained online bidding for (i) the adversarial environment and (ii) value-maximization in the stochastic environment. The result for the adversarial environment holds against non-standard benchmark (perfect pacing sequence) which side-steps impossibility results from prior work and has been fruitfully applied to general bandits-with-knapsacks problems in subsequent work \citep{LagCBwK-jmlr24,BwK-braverman2025}. The modularity of our techniques enables extending both aggregate and individual guarantees from \Cref{alg:bg} to several other algorithms.\footnote{This addresses an open question from the preliminary conference version of this paper \citep{Gaitonde-itcs23}.}

%We have shown that a natural budget-pacing algorithm achieves good aggregate and individual performance guarantees for a wide class of auctions and arbitrarily correlated private values, without relying on convergence to an equilibrium.

%The main question left open by our work concern simultaneously achieving similar individual and aggregate guarantees. Can such results be extended to other (classes of) algorithms, perhaps bringing to bear the ``generalized pacing" notion from \Cref{def:gen-pacing}? More concretely, are there other ``generalized pacing" algorithms that satisfy compelling individual guarantees, \eg vanishing regret in the stochastic environment? Can one improve the individual guarantees while keeping the same liquid welfare guarantee?

Let us emphasize several open questions:

\xhdr{(1)}.
Improving the approximation factor against our liquid-welfare benchmark appears unlikely (given the negative result for static truthful auctions). Could one identify a different liquid-welfare benchmark --- possibly weaker but hopefully more fair --- which would allow for a better approximation factor or, ideally, a regret bound with no approximation factor at all?

\xhdr{(2).}
The key goal for individual guarantees is vanishing regret for all agents at once, particularly under \emph{self-play} (when all agents use the same algorithm).  Previously this goal seemed hopeless for budget-constrained bidders, given the strong impossibility results known for adversarial environments. We provide a plausible framework in which this goal might be achievable: regret bounds relative to the perfect pacing sequence. Recall that our specific guarantee falls just short: \Cref{cor:regret-stochastic-2} requires path-length $P^* = o(\sqrt{T})$, but we can only guarantee $P^* = O(\sqrt{T})$ when all agents run our algorithm. Bridging this gap may be within reach, possibly by leveraging the ``optimism" in \PacingOMD (see \Cref{sec:other}). Vanishing regret under self-play may also be achievable against the \emph{standard} benchmark: best pacing multiplier in hindsight. Indeed, our semi-synthetic numerical simulations exhibit this for \Cref{alg:bg} and several other bidding algorithms.

%Our semi-synthetic numerical simulations suggest that our algorithm achieves vanishing regret against the standard benchmark (best pacing multiplier in hindsight) in simultaneous learning environments.  We leave open the question of whether vanishing regret under simultaneous learning can be proven analytically, but conjecture that this is the case.

%\footnote{Regret $\tilde{O}(T^{2/3})$ is likely the best possible individual guarantee, even in the stochastic environment, under Lipschitz assumptions, and without budget constraints, given the impossibility results on related bandit problems \citep{Bobby-nips04,KleinbergL03}.}

\xhdr{(3).} While optimizing the guarantees on individual regret is not our goal per se, it would be desirable to improve upon the $T^{3/4}$ regret rates in our theoretical results and/or compete against non-linear bidding policies for value-maximization. Improved results for the stochastic environment are especially interesting when the same algorithm enjoys non-trivial guarantees for the adversarial environment.

%regret bounds for repeated non-truthful auctions, particularly for repeated first-price auctions, even if they are not accompanied by aggregate guarantees. Improved regret bounds for the stochastic environment are especially interesting when combined with some guarantees for the adversarial environment. As linear policies (mappings from values to bids) are not necessarily  value-optimizing for non-truthful auctions, an important sub-question is which class(es) of policies one would want to optimize over.

%One question left open by our work concerns simultaneously achieving similar individual and aggregate guarantees with (more) general classes of bidding algorithms. A concrete goal along these lines is to improve the regret bound for the stochastic environment to $\tilde{O}(T^{2/3})$, while simultaneously achieving a constant-factor liquid welfare guarantee.

%\begin{acks}
%This work was initiated while the first two named authors were interns at Microsoft Research. JG is supported in part by NSF Award CCF-1408673 and AFOSR Award FA9550-19-1-0183.
%\end{acks}

\begin{small}
\bibliographystyle{plainnat}
\bibliography{learning,bib-abbrv,bib-slivkins,bib-bandits,bib-AGT}

@article{Moulin-78,
    author = {Herve Moulin and Jean-Paul Vial},
    title = {Strategically zero-sum games: the class of games whose completely mixed equilibria cannot be improved upon},
    journal = {Intl. J. of Game Theory},
    volume = {7},
    number = {3},
    pages = {201–221},
    year = {1978}
}

@article{Aumann-74,
    author = {Robert J. Aumann},
    title = {Subjectivity and correlation in randomized strategies},
    journal = {J. of Mathematical Economics},
    volume = {1},
    pages = {67–96},
    year = {1974}
}

@inproceedings{TotalAnarchy-stoc08,
  author    = {Avrim Blum and
               MohammadTaghi Hajiaghayi and
               Katrina Ligett and
               Aaron Roth},
  title     = {Regret minimization and the price of total anarchy},
  booktitle = proc # "40th" # STOC,
  pages     = {373--382},
  year      = {2008}
}

@inproceedings{Roughgarden-PoA-stoc09,
  author    = {Tim Roughgarden},
  title     = {Intrinsic robustness of the price of anarchy},
  booktitle = proc # "41st" # STOC,
  pages     = {513--522},
  year      = {2009}
}

@inproceedings{Syrgkanis-nips15,
  author    = {Vasilis Syrgkanis and
               Alekh Agarwal and
               Haipeng Luo and
               Robert E. Schapire},
  title     = {Fast Convergence of Regularized Learning in Games},
  booktitle =  proc # "28th" # NIPS,
  pages     = {2989--2997},
  year      = {2015}
}

@inproceedings{Rakhlin-nips13,
  author    = {Alexander Rakhlin and
               Karthik Sridharan},
  title     = {Optimization, Learning, and Games with Predictable Sequences},
  booktitle = proc # "27th" # NIPS,
  pages     = {3066--3074},
  year      = {2013}
}

@inproceedings{Daskalakis-iclr18,
  author    = {Constantinos Daskalakis and
               Andrew Ilyas and
               Vasilis Syrgkanis and
               Haoyang Zeng},
  title     = {Training GANs with Optimism},
  booktitle = proc # "6th" # ICLR,
  year      = {2018}
}

@inproceedings{Daskalakis-neurips18-limitPoints,
  author       = {Constantinos Daskalakis and
                  Ioannis Panageas},
  title        = {The Limit Points of (Optimistic) Gradient Descent in Min-Max Optimization},
  booktitle    = proc # "31st" # NeurIPS,
  year = {2018}
}

@inproceedings{Daskalakis-neurips20-lastIterate,
  author    = {Noah Golowich and
               Sarath Pattathil and
               Constantinos Daskalakis},
  title     = {Tight last-iterate convergence rates for no-regret learning in multi-player games},
  booktitle = proc # "33rd" # NeurIPS,
  year      = {2020}
}

@inproceedings{Daskalakis-itcs19-lastIterate,
  author    = {Constantinos Daskalakis and
               Ioannis Panageas},
  title     = {Last-Iterate Convergence: Zero-Sum Games and Constrained Min-Max Optimization},
  booktitle = proc # "10th" # ITCS,
  volume    = {124},
  pages     = {27:1--27:18},
  year      = {2019}
}

@inproceedings{Haipeng-iclr21-lastIterate,
  author    = {Chen{-}Yu Wei and
               Chung{-}Wei Lee and
               Mengxiao Zhang and
               Haipeng Luo},
  title     = {Linear Last-iterate Convergence in Constrained Saddle-point Optimization},
  booktitle = proc # "9th" # ICLR,
  year      = {2021}
}

@inproceedings{Piliouras-colt19,
  author    = {Yun Kuen Cheung and
               Georgios Piliouras},
  title     = {Vortices Instead of Equilibria in MinMax Optimization: Chaos and Butterfly
               Effects of Online Learning in Zero-Sum Games},
  booktitle = COLT,
  pages     = {807--834},
  year      = {2019}
}

@inproceedings{Piliouras-ec18,
  author    = {James P. Bailey and
               Georgios Piliouras},
  title     = {Multiplicative Weights Update in Zero-Sum Games},
  booktitle =  ACMEC,
  pages     = {321--338},
  year      = {2018}
}

@inproceedings{Piliouras-soda18,
  author    = {Panayotis Mertikopoulos and
               Christos H. Papadimitriou and
               Georgios Piliouras},
  title     = {Cycles in Adversarial Regularized Learning},
  booktitle = proc # "29th" # SODA,
  pages     = {2703--2717},
  year      = {2018}
}

@InProceedings{mokhtari-aistats20,
  title = 	 {A Unified Analysis of Extra-gradient and Optimistic Gradient Methods for Saddle Point Problems: Proximal Point Approach},
  author =       {Mokhtari, Aryan and Ozdaglar, Asuman and Pattathil, Sarath},
  booktitle =  proc # "23rd" # AISTATS,
  pages = 	 {1497--1507},
  year = 	 {2020}
}

@article{Jiang-siam25,
author = {Jiang, Ruichen and Mokhtari, Aryan},
title = {Generalized Optimistic Methods for Convex-Concave Saddle Point Problems},
journal = {SIAM Journal on Optimization},
volume = {35},
number = {3},
pages = {2066-2097},
year = {2025}
}

@misc{BwK-braverman2025,
      title={A New Benchmark for Online Learning with Budget-Balancing Constraints},
      author={Mark Braverman and Jingyi Liu and Jieming Mao and Jon Schneider and Eric Xue},
      year={2025},
      eprint={2503.14796},
      archivePrefix={arXiv},
      primaryClass={cs.LG},
      url={https://arxiv.org/abs/2503.14796},
}

@inproceedings{Castiglioni-icml22,
  author    = {Matteo Castiglioni and
               Andrea Celli and
               Christian Kroer},
  title     = {Online Learning with Knapsacks: the Best of Both Worlds},
  booktitle = proc # "39th" # ICML,
  year      = {2022}
}

@article{AgrawalDevanur-ec14-OpRe,
  author    = {Shipra Agrawal and Nikhil R. Devanur},
  title     = {Bandits with Global Convex Constraints and Objective},
  journal   = {Operations Research},
  volume    = {67},
  number    = {5},
  pages     = {1486--1502},
  year      = {2019},
  note      = {Preliminary version in \emph{ACM EC 2014}.}
}

@inproceedings{Fikioris-LPoA23,
  author    = {Giannis Fikioris and {\'{E}}va Tardos},
  title     = {Liquid Welfare guarantees for No-Regret Learning in Sequential Budgeted Auctions},
  booktitle = proc # "24rd" # ACMEC,
  year      = {2023}
}

@inproceedings{Chen-wine21,
  author    = {Xi Chen and
               Christian Kroer and
               Rachitesh Kumar},
  title     = {Throttling Equilibria in Auction Markets},
  booktitle = proc # "17th" # WINE,
  pages     = {551},
  year      = {2021}
}

@inproceedings{StandardAuctions-ec22,
  author    = {Santiago R. Balseiro and
               Christian Kroer and
               Rachitesh Kumar},
  title     = {Contextual Standard Auctions with Budgets: Revenue Equivalence and Efficiency Guarantees},
  booktitle = proc # "23rd" # ACMEC,
  pages     = {476},
  year      = {2022},
}

@inproceedings{Chen-ec21,
  author    = {Xi Chen and
               Christian Kroer and
               Rachitesh Kumar},
  title     = {The Complexity of Pacing for Second-Price Auctions},
  booktitle = proc # "21st" # ACMEC,
  pages     = {318},
  year      = {2021}
}

@inproceedings{Balseiro-ec21,
  title={The Landscape of Auto-bidding Auctions: Value versus Utility Maximization},
  author={Balseiro, Santiago R and Deng, Yuan and Mao, Jieming and Mirrokni, Vahab S and Zuo, Song},
  booktitle= proc # "22nd" # ACMEC,
  pages={132--133},
  year={2021}
}

@inproceedings{Babaioff-itcs21,
  author    = {Moshe Babaioff and
               Richard Cole and
               Jason D. Hartline and
               Nicole Immorlica and
               Brendan Lucier},
  title     = {Non-Quasi-Linear Agents in Quasi-Linear Mechanisms (Extended Abstract)},
  booktitle = proc # "12th" # ITCS,
  pages     = {84:1--84:1},
  year      = {2021}
}

@inproceedings{Gagan-wine19,
  author    = {Gagan Aggarwal and
               Ashwinkumar Badanidiyuru and
               Aranyak Mehta},
  title     = {Autobidding with Constraints},
  booktitle = proc # "15th" # WINE,
  pages     = {17--30},
  year      = {2019}
}

@article{BalseiroGur19,
  author    = {Santiago R. Balseiro and
               Yonatan Gur},
  title     = {Learning in Repeated Auctions with Budgets: Regret Minimization and
               Equilibrium},
  journal   = {Manag. Sci.},
  volume    = {65},
  number    = {9},
  pages     = {3952--3968},
  year      = {2019},
  note      = {Preliminary version in \emph{ACM EC 2017}.}
}

@article{Balseiro-BestOfMany-Opre,
  author    = {Santiago R. Balseiro and
               Haihao Lu and
               Vahab S. Mirrokni},
  title     = {The Best of Many Worlds: Dual Mirror Descent for Online Allocation
               Problems},
  journal   = {Operations Research},
   volume      = {71},
  number       = {1},
  pages        = {101--119},
  year      = {2023},
  note      = {Preliminary version in \emph{ICML 2020}.}
}

@inproceedings{Conitzer-wine18,
  author    = {Vincent Conitzer and
               Christian Kroer and
               Eric Sodomka and
               Nicol{\'{a}}s E. Stier Moses},
  title     = {Multiplicative Pacing Equilibria in Auction Markets},
  booktitle = proc # "14th" # WINE,
  pages     = {443},
  year      = {2018}
}

@inproceedings{Conitzer-ec19,
  author    = {Vincent Conitzer and
               Christian Kroer and
               Debmalya Panigrahi and
               Okke Schrijvers and
               Eric Sodomka and
               Nicol{\'{a}}s E. Stier Moses and
               Chris Wilkens},
  title     = {Pacing Equilibrium in First-Price Auction Markets},
  booktitle =  ACMEC,
  pages     = {587},
  year      = {2019}
}

@inproceedings{Azar-wine17,
  author    = {Yossi Azar and
               Michal Feldman and
               Nick Gravin and
               Alan Roytman},
  title     = {Liquid Price of Anarchy},
  booktitle = proc # "10th" # SAGT,
  pages     = {3--15},
  year      = {2017}
}

@inproceedings{Shahar-icalp14,
  author    = {Shahar Dobzinski and
               Renato Paes Leme},
  title     = {Efficiency Guarantees in Auctions with Budgets},
  booktitle = proc # "41st" # ICALP,
  pages     = {392--404},
  year      = {2014},
}

@inproceedings{Borgs-www07,
  author    = {Christian Borgs and
               Jennifer T. Chayes and
               Nicole Immorlica and
               Kamal Jain and
               Omid Etesami and
               Mohammad Mahdian},
  title     = {Dynamics of bid optimization in online advertisement auctions},
  booktitle =  proc # "16th" # WWW,
  pages     = {531--540},
  year      = {2007}
}

@String{ proc         =  ""}

@String{ FOCS         = " IEEE Symp. on Foundations of Computer Science (FOCS)"}

@String{ STOC         = " ACM Symp. on Theory of Computing (STOC)"}

@String{ SODA         = " ACM-SIAM Symp. on Discrete Algorithms (SODA)"}

@String{ ICALP         = " Intl. Colloquium on Automata, Languages and Programming (ICALP)"}

@String{ WINE         = " Workshop on Internet \& Network Economics (WINE)"}

@String{ ITCS         = " Innovations in Theoretical Computer Science Conf. (ITCS)"}

@String{ SAGT         =  " Symp. on Algorithmic Game Theory (SAGT)"}

@String{ EC         =  " ACM Conf. on Electronic Commerce (ACM-EC)"}

@String{ ACMEC         =  " ACM Conf. on Economics and Computation (ACM-EC)"}

@String{ NIPS         = " Advances in Neural Information Processing Systems (NIPS)"}

@String{ NeurIPS         = " Advances in Neural Information Processing Systems (NeurIPS)"}

@String{ COLT         = " Conf. on Learning Theory (COLT)"}

@String{ ICML         = " Intl. Conf. on Machine Learning (ICML)"}

@String{ AISTATS = " Intl. Conf. on Artificial Intelligence and Statistics (AISTATS)"}

@String{ ICLR   = " International Conference on Learning Representations (ICLR)"}

@String{ WWW  = " Intl. World Wide Web Conf. (WWW)"}

@String{ JACM   = "J. of the ACM"}

@String{ JMLR         = "J. of Machine Learning Research (JMLR)"}

@String{ CHI  = " Conf. on Human Factors in Computing Systems (CHI)"}

@inproceedings{Chiang-colt12,
  author       = {Chao{-}Kai Chiang and
                  Tianbao Yang and
                  Chia{-}Jung Lee and
                  Mehrdad Mahdavi and
                  Chi{-}Jen Lu and
                  Rong Jin and
                  Shenghuo Zhu},
  title        = {Online Optimization with Gradual Variations},
  booktitle    = proc # "25th" # COLT,
  year         = {2012},
}

@inproceedings{Rakhlin-colt13,
  author    = {Alexander Rakhlin and Karthik Sridharan},
  title     = {Online Learning with Predictable Sequences},
  booktitle = proc # "26th" # COLT,
  volume    = {30},
  pages     = {993--1019},
  year      = {2013}
}

@inproceedings{jadbabaie2015online,
  title={Online optimization: Competing with dynamic comparators},
  author={Jadbabaie, Ali and Rakhlin, Alexander and Shahrampour, Shahin and Sridharan, Karthik},
  booktitle={Artificial Intelligence and Statistics},
  pages={398--406},
  year={2015},
  organization={PMLR}
}

@article{Hazan-OCO-book,
   author         = {Elad Hazan},
   title         = "{Introduction to Online Convex Optimization}",
  journal={Foundations and Trends in Optimization},
  volume={2},
  number={3-4},
  pages={157-325},
  year={2015},
  note = "Published with \emph{Now Publishers} (Boston, MA, USA).
    Also available at {https://arxiv.org/abs/1909.05207}."
}

@article{bandits-exp3,
  author    = {Peter Auer and
               Nicol{\`o} Cesa-Bianchi and
               Yoav Freund and
               Robert E. Schapire},
  title     = {The Nonstochastic Multiarmed Bandit Problem.},
  journal   = {SIAM J. Comput.},
  volume    = {32},
  number    = {1},
  year      = {2002},
  pages     = {48-77},
  note            = {Preliminary version in {\em 36th IEEE FOCS}, 1995}
}

@article{HartMasCollel-econometrica00,
  author = {Sergiu Hart and Andreu Mas-Colell},
  title = {A simple adaptive procedure leading to correlated equilibrium},
  journal = {Econometrica},
  volume = {68},
  pages = {1127–1150},
  year = {2000},
}

@article{slivkins-MABbook,
    author = "Aleksandrs Slivkins",
    title= "Introduction to Multi-Armed Bandits",
    journal = "Foundations and Trends$\circledR$ in Machine Learning",
    volume = {12},
    number = {1-2},
    pages = {1-286},
    month = nov,
    year = "2019",
    note = "Published with \emph{Now Publishers} (Boston, MA, USA). Also available at
        {\tt https://arxiv.org/abs/1904.07272}."
}

@inproceedings{cBwK-colt14,
  author =       "Ashwinkumar Badanidiyuru and John Langford and Aleksandrs Slivkins",
  title =        "Resourceful Contextual Bandits",
  booktitle         = proc # "27th" # COLT,
  year = "2014"
}

@article{BwK-focs13,
  author =       "Ashwinkumar Badanidiyuru and Robert Kleinberg and Aleksandrs Slivkins",
  title =        "Bandits with Knapsacks",
  journal   = JACM,
  volume    = {65},
  number    = {3},
  pages     = {13:1--13:55},
  year      = {2018},
  note = "Preliminary version in {\em FOCS 2013}."
}

@article{AdvBwK-focs19,
  author    = {Nicole Immorlica and Karthik Abinav Sankararaman and Robert Schapire and Aleksandrs Slivkins},
  title     = {Adversarial Bandits with Knapsacks},
  journal = JACM,
  year      = {2022},
  month = aug,
  note  = {Preliminary version in \emph{60th IEEE FOCS}, 2019.}
}

@article{LagCBwK-jmlr24,
  author  = {Aleksandrs Slivkins and Xingyu Zhou and Karthik Abinav Sankararaman and Dylan J. Foster},
  title   = {Contextual Bandits with Packing and Covering Constraints: A Modular Lagrangian Approach via Regression},
  journal = JMLR,
  year    = {2024},
  volume  = {25},
  number  = {394},
  pages   = {1--37}
}

@inproceedings{Autobidding-colt24,
  author    = {Brendan Lucier and
               Sarath Pattathil and
               Aleksandrs Slivkins and
               Mengxiao Zhang},
  title     = {Autobidders with Budget and {ROI} Constraints: Efficiency, Regret, and Pacing Dynamics},
  booktitle = proc # "37th" # COLT,
  year      = {2024}
}

@inproceedings{Gaitonde-itcs23,
  author    = {Jason Gaitonde and
               Yingkai Li and
               Bar Light and
               Brendan Lucier and
               Aleksandrs Slivkins},
  title     = {Budget Pacing in Repeated Auctions: Regret and Efficiency without Convergence},
  year      = {2023},
  booktitle = proc # "14th" # ITCS
}

@article{balseiro2021budget,
  title={Budget-management strategies in repeated auctions},
  author={Balseiro, Santiago and Kim, Anthony and Mahdian, Mohammad and Mirrokni, Vahab},
  journal={Operations Research},
  volume={69},
  number={3},
  pages={859--876},
  year={2021},
  publisher={INFORMS}
}

@article{DBLP:journals/jair/RoughgardenST17,
  author    = {Tim Roughgarden and
               Vasilis Syrgkanis and
               {\'{E}}va Tardos},
  title     = {The Price of Anarchy in Auctions},
  journal   = {J. Artif. Intell. Res.},
  volume    = {59},
  pages     = {59--101},
  year      = {2017},
  url       = {https://doi.org/10.1613/jair.5272},
  doi       = {10.1613/jair.5272},
  bibsource = {dblp computer science bibliography, https://dblp.org}
}

@misc{enberg_2019, 
 title={US digital ad spending 2019}, url={https://www.emarketer.com/content/us-digital-ad-spending-2019}, 
 journal={Insider Intelligence}, 
 publisher={Insider Intelligence}, 
 author={Enberg, Jasmine}, 
 year={2019}, 
 month={Mar}
 }

@article{goel2015polyhedral,
  title={Polyhedral clinching auctions and the adwords polytope},
  author={Goel, Gagan and Mirrokni, Vahab and Leme, Renato Paes},
  journal={Journal of the ACM (JACM)},
  volume={62},
  number={3},
  pages={1--27},
  year={2015},
  publisher={ACM New York, NY, USA}
}

@inproceedings{deng2021towards,
  title={Towards efficient auctions in an auto-bidding world},
  author={Deng, Yuan and Mao, Jieming and Mirrokni, Vahab and Zuo, Song},
  booktitle={Proceedings of the Web Conference 2021},
  pages={3965--3973},
  year={2021}
}

@article{pai2014optimal,
  title={Optimal auctions with financially constrained buyers},
  author={Pai, Mallesh M and Vohra, Rakesh},
  journal={Journal of Economic Theory},
  volume={150},
  pages={383--425},
  year={2014},
  publisher={Elsevier}
}

@article{ausubel2002ascending,
  title={Ascending auctions with package bidding},
  author={Ausubel, Lawrence M and Milgrom, Paul R},
  journal={Advances in Theoretical Economics},
  volume={1},
  number={1},
  year={2002},
  publisher={De Gruyter}
}

@inproceedings{goel2015core,
  title={Core-competitive auctions},
  author={Goel, Gagan and Khani, Mohammad Reza and Leme, Renato Paes},
  booktitle={Proceedings of the Sixteenth ACM Conference on Economics and Computation},
  pages={149--166},
  year={2015}
}

@inproceedings{hartline2018fast,
  title={Fast core pricing for rich advertising auctions},
  author={Hartline, Jason and Immorlica, Nicole and Khani, Mohammad Reza and Lucier, Brendan and Niazadeh, Rad},
  booktitle={Proceedings of the 2018 ACM Conference on Economics and Computation},
  pages={111--112},
  year={2018}
}

@article{roughgarden2017price,
  title={The price of anarchy in auctions},
  author={Roughgarden, Tim and Syrgkanis, Vasilis and Tardos, Eva},
  journal={Journal of Artificial Intelligence Research},
  volume={59},
  pages={59--101},
  year={2017}
}

@article{aggarwal2024auto,
  title={Auto-bidding and auctions in online advertising: A survey},
  author={Aggarwal, Gagan and Badanidiyuru, Ashwinkumar and Balseiro, Santiago R and Bhawalkar, Kshipra and Deng, Yuan and Feng, Zhe and Goel, Gagan and Liaw, Christopher and Lu, Haihao and Mahdian, Mohammad and others},
  journal={ACM SIGecom Exchanges},
  volume={22},
  number={1},
  pages={159--183},
  year={2024},
  publisher={ACM New York, NY, USA}
}

@article{kingma2014adam,
  title={Adam: A method for stochastic optimization},
  author={Kingma, Diederik P},
  journal={arXiv preprint arXiv:1412.6980},
  year={2014}
}

@inproceedings{rakhlin2013online,
  title={Online learning with predictable sequences},
  author={Rakhlin, Alexander and Sridharan, Karthik},
  booktitle={Conference on Learning Theory},
  pages={993--1019},
  year={2013},
  organization={PMLR}
}

@inproceedings{balseiro2023robust,
  title={Robust budget pacing with a single sample},
  author={Balseiro, Santiago R and Kumar, Rachitesh and Mirrokni, Vahab and Sivan, Balasubramanian and Wang, Di},
  booktitle={International Conference on Machine Learning},
  pages={1636--1659},
  year={2023},
  organization={PMLR}
}

@inproceedings{wang2023learning,
  title={Learning to bid in repeated first-price auctions with budgets},
  author={Wang, Qian and Yang, Zongjun and Deng, Xiaotie and Kong, Yuqing},
  booktitle={International Conference on Machine Learning},
  pages={36494--36513},
  year={2023},
  organization={PMLR}
}

@book{perloff2009microeconomics,
  title={Microeconomics, 7th ed.},
  author={Perloff, Jeffrey M},
  year={2015},
  publisher={Pearson Education}
}
\end{small}

\newpage
\addtocontents{toc}{\protect\setcounter{tocdepth}{2}}
\renewcommand{\contentsname}{APPENDICES}
\tableofcontents
\newpage

\appendix

% uncomment if you are using cleveref
\crefalias{section}{appendix}
\crefalias{subsection}{appendix}

\section{Examples of Auctions}\label{app:ExamplesofAuctions}
%In this section we provide the details of the ad-auctions we mentioned in Section \ref{sec:model}.
\subsection{Single-slot and Multiple-slot Ad Auctions}
\label{app:single-multiple-slot}
\fakeItem
{\em Single-Slot Ad Auctions:} A round corresponds to an ad impression, with a single ad slot available.
%and there is a single ad slot available for each impression.  
An impression has a \emph{type} $\theta_t \in \Theta$; this type might describe, for example, a keyword being searched for, user demographics, intent prediction, etc.  In each round $t$ the impression type $\theta_t$ is drawn independently from a distribution over types.  Each agent $k$ has a fixed value function $v_k \colon \Theta \to \reals_{\geq 0}$ that maps each impression type to a value for being displayed.  The value profile in round $t$ is then $\mathbf{v}_t = (v_1(\theta_t), \dotsc, v_n(\theta_t))$.  A single ad can be displayed each round.  We then interpret $x_{k,t} \in [0,1]$ as the probability that advertiser $k$ is allocated the ad slot, and an allocation profile $\mathbf{x}_t$ is feasible if and only if $\sum_k x_{k,t} \leq 1$.

%\item
\fakeItem
{\em Multiple-Slot Pay-per-click Ad Auctions:} We can generalize the previous example to allow multiple ad slots, using a polymatroid formulation due to \citep{goel2015polyhedral}.  Impression types and agent values are as before, but we now think of there as being $m \geq 1$ slots available in each round, with click rates $1 \geq \alpha_1 \geq \dotsc \geq \alpha_{m} \geq 0$. If ad $k$ is placed in slot $i$, the value to agent $k$ is $v_k(\theta_t) \times \alpha_i$.  That is, we think of $v_k(\theta_t)$ as representing both the advertiser-specific click rate (which can depend on the impression type) as well as the advertiser's value for a click. In this case we would take $v_{k,t} = v_k(\theta_t)$ and $x_{k,t} = \alpha_i$.  The set of feasible allocation profiles $X \subseteq [0,1]^n$ is then a polymatroid: $\mathbf{x}_t \in X$ if and only if, for each $\ell \leq m$ and each $S \subseteq [n]$ with $|S| = \ell$, $\sum_{k \in S}x_{k,t} \leq \sum_{i = 1}^k \alpha_i$.

\subsection{Core and MBB Auctions} \label{app:CoreAuctions}

Let us list three notable examples of core auctions. They also satisfy the monotone bang-per-buck (MBB) property, as defined in \refeq{eq.mbb.1}.
\begin{itemize}
    \item {\em First-Price Auction} chooses a welfare-maximizing allocation $\mathbf{x}(\mathbf{b}) \in \arg\max_{\mathbf{x}}\{\sum_k b_k x_k\}$.  Each agent $k$ then pays her bid for the allocation obtained: $p_k(\mathbf{b}) = b_k x_k(\mathbf{b})$. 
        
        To see this auction is MBB, first note that $x_k(\mathbf{b})$ is nondecreasing in $b_k$. Then observe that if $\mathbf{b}_{-k}$ is fixed, and $b_k< b_k'$, then setting $\mathbf{b}'=(b_k',\mathbf{b}_{-k})$, we clearly have
    \begin{equation*}
        p_k(\mathbf{b}')-p_k(\mathbf{b})
            =b_k'\, x_k(\mathbf{b}')- b_k\, x_k(\mathbf{b})\geq b_k\rbr{x_k(\mathbf{b}')-x_k(\mathbf{b})},
    \end{equation*}
    as needed by the definition of MBB.
    
    \item {\em Second-Price Auction for Single-Slot Ad Auctions.}  
    %Consider the single-slot environment described earlier; 
            %$X = \{ \mathbf{x} \colon \sum_k x_k \leq 1\}$.
        The second-price auction chooses a welfare-maximizing allocation $\mathbf{x}(\mathbf{b}) \in \arg\max_{\mathbf{x}}\{\sum_k b_k x_k\}$, 
        where feasible allocations $\mathbf{x}$ satisfy $\sum_k x_k \leq 1$.
        Then each agent $k$ pays $x_k$ times the second-highest bid.

        To see this auction is MBB, fix $\mathbf{b}_{-k}$ and let $b^*$ denote the largest single bid in $\mathbf{b}_{-k}$. Note that if $b_k \geq b^*$ then the second-highest bid is $b^*$ so $p_k(\mathbf{b}) = b^* x_k(\mathbf{b})$, and if $b_k < b^*$ then $x_k(\mathbf{b}) = p_k(\mathbf{b}) = 0$ so $p_k(\mathbf{b}) = b^* x_k(\mathbf{b})$.  Thus, for any $b_k < b'_k$, we have
        \[
        p_k(\mathbf{b}') - p_k(\mathbf{b}) = b^*( x_k(\mathbf{b}') - x_k(\mathbf{b}) \geq b_k( x_k(\mathbf{b}') - x_k(\mathbf{b})
        \]
        as needed by the definition of MBB, where the inequality follows because if $b_k < b^*$ then the left- and right-hand sides are both zero.
    \item {\em Generalized Second-Price (GSP) Auction for Multi-Slot Ad Auctions.} 
     %Consider the multi-slot environment described earlier.
        %in Section~\ref{sec:model.allocation}.
        In the GSP auction, slots are allocated greedily by the respective bid, and each agent pays a price per unit equal to the next-highest bid.  Formally, given the bids $\mathbf{b}$, we let $\pi$ be a permutation of the agents so that $b_{\pi(1)} \geq b_{\pi(2)} \geq \dotsc \geq b_{\pi(n)}$.  That is, $\pi(1)$ is the highest-bidding agent, then $\pi(2)$, etc.  Agent $\pi(k)$ is then allocated to slot $k$ for each $k \leq m$.  That is, $x_{\pi(k)} = \alpha_k$ for each $k \leq m$ and $x_{\pi(k)} = 0$ for each $k > m$.  The payments are set so that $p_{\pi(k)} = x_{\pi(k)} b_{\pi(k+1)}$ for all $k < n$, and $p_{\pi(n)} = 0$.
        \footnote{Note that $p_{\pi(k)} = 0$ whenever $x_{\pi(k)} = 0$.}
        
        We note that GSP is a core auction; see \Cref{app:gsp.core} for a proof. 
        
        To see that the MBB property holds, let $b_k<b_k'$ and suppose that with bid profiles $\mathbf{b}$, agent $k$ is assigned the $j$th slot, while under $\mathbf{b}'$, the agent is assigned the $\ell\leq j$th slot (where we allow the value of a slot to be $0$ if the agent is not in the top $m$ bids). If $j=\ell$, then it is easy to see that the price and allocation of agent $k$ is unchanged, trivially confirming that the MBB condition holds. If instead $j>\ell$, then it is easy to see that $b_{\pi'(\ell+1)}\geq b_k$, with $\pi'$ the permutation under $\mathbf{b}'$, and hence
        \begin{equation*}
           p_k(\mathbf{b}')-p_k(\mathbf{b}')\geq \alpha_{\ell}b_k - \alpha_{\ell}b_{\pi(j+1)}\geq b_k(\alpha_{\ell}-\alpha_{j})=b_k(x_k(\mathbf{b}')-x_k(\mathbf{b})).
        \end{equation*}
        This implies that the MBB property holds for single-slot auctions.
\end{itemize}

\subsection{Proof: Generalized Second Price is a Core Auction}
\label{app:gsp.core}

In this section we show that the GSP auction for sponsored search allocation with separable click rates is a core auction.  Recall the definition of a GSP auction.  There are $m \geq 1$ slots with click rates $1 \geq \alpha_1 \geq \dotsc \geq \alpha_m \geq 0$.  There are $n$ bidders, each bidder placing a bid $b_i \geq 0$.  We will reindex agents in order of bid, so that $b_1 \geq b_2 \geq \dotsc \geq b_n$.  Without loss of generality we will assume $m = n$ (by adding extra bidders with bid $0$ or extra slots with click rate $0$), and we will define $b_{n+1} = \alpha_{m+1} = 0$ for convenience.

In the GSP auction, slots are allocated greedily by bid, and each agent pays a price per unit equal to the next-highest bid.  That is, agent $i$ receives slot $i$ for a declared value of $b_i \alpha_i$, and pays $b_{i+1} \alpha_i$.

We claim that the GSP auction is a core auction.  First, since $b_{i+1} \leq b_i$ for all $i$, we have that $b_{i+1} \alpha_i \leq b_i \alpha_i$, and hence each bidder pays at most her declared welfare for the allocation received.

It remains to show the second property of a core auction.  Choose any subset of bidders $S \subseteq [n]$.  The allocation $\mathbf{y}$ to agents in $S$ that maximizes declared welfare is the one that allocates greedily in index order.  More formally, for each $i \in S$, let $\sigma(i)$ be 1 plus the number of elements of $S$ with index less than $i$.  For example, if $S = \{2, 6, 7\}$, then $\sigma(2) = 1$, $\sigma(6) = 2$, and $\sigma(7) = 3$.  Then the declared-welfare-maximizing allocation $\mathbf{y}$ to agents in $S$ is such that $y_i = \alpha_{\sigma(i)}$ for each $i$, for a total declared welfare of $\sum_{i \in S} b_i \alpha_{\sigma(i)}$.  The core auction property on subset of bidders $S$ therefore reduces to showing that
\begin{equation}
\label{eq:gsp.core}
    \sum_{i \not\in S} b_{i+1} \alpha_i + \sum_{i \in S} b_i \alpha_i \geq \sum_{i \in S} b_i \alpha_{\sigma(i)}.
\end{equation}
To establish inequality \eqref{eq:gsp.core}, we first note that
\begin{align*}
    \sum_{i \in S} b_i \alpha_{\sigma(i)} - \sum_{i \in S} b_i \alpha_i
    & = \sum_{i \in S} b_i (\alpha_{\sigma(i)} - \alpha_i) \\
    & = \sum_{i \in S} \sum_{j = \sigma(i)}^{i-1} b_i (\alpha_j - \alpha_{j+1}) \\
    & \leq \sum_{i \in S} \sum_{j = \sigma(i)}^{i-1} b_{j+1}(\alpha_j - \alpha_{j+1}).
\end{align*}
This final double summation contains an instance of the term $b_{j+1}(\alpha_j - \alpha_{j+1})$ for each $i \in S$ such that $\sigma(i) \leq j < i$.  But for each $j$ and each $i$ such that $\sigma(i) \leq j < i$ (i.e., such that the ``new'' allocation to agent $i$ under $\mathbf{y}$ is slot $j$ or better, and the ``old'' allocation is worse than slot $j$), there must be some $k \leq j$ such that $k \not\in S$.  The number of such $i$ is therefore at most the number of agents at index $j$ or less that are not in $S$.  More precisely, observe that $
   j\geq \vert \{i\in S: i\leq j\}\vert+\vert \{i\in S: \sigma(i)\leq j<i\}\vert$. This holds as $\sigma(i)\leq i$ for all $i$ and is injective, so the sets on the right hand side are evidently disjoint and $\sigma$ maps each such element to a unique index in $[j]$.
But $j = \vert \{i\in S: i\leq j\}\vert+\vert \{k\not\in S: k\leq j\}\vert$, so cancelling terms shows that $\vert \{k\not\in S: k\leq j\}\vert \geq \vert \{i\in S: \sigma(i)\leq j<i\}\vert$.
Thus, by rearranging the order of summation, we have

\begin{align*}
%    \sum_{i \in S} b_i \alpha_{\sigma(i)} - \sum_{i \in S} b_i \alpha_i
%    & \leq \sum_{i \in S} \sum_{j = \alpha_{\sigma(i)}}^{\alpha_i+1} b_{j+1}(\alpha_j - \alpha_{j+1}) \\
\sum_{i \in S} \sum_{j = \sigma(i)}^{i-1} b_{j+1}(\alpha_j - \alpha_{j+1})
& \leq \sum_{j = 1}^{m} b_{j+1}(\alpha_j - \alpha_{j+1}) \times |\{ k \leq j \colon k \not\in S \}| \\
    & = \sum_{k \not\in S} \sum_{j \geq k} b_{j+1}(\alpha_j - \alpha_{j+1}) \\
    & \leq \sum_{k \not\in S} \sum_{j \geq k} b_{k+1}(\alpha_j - \alpha_{j+1}) \\
    & \leq \sum_{k \not\in S} b_{k+1} \alpha_k
\end{align*}
where the final inequality is a telescoping sum.  We therefore conclude that
\[ \sum_{i \in S} b_i \alpha_{\sigma(i)} - \sum_{i \in S} b_i \alpha_i \leq \sum_{k \not\in S} b_{k+1} \alpha_k\]
and rearranging yields the desired inequality \eqref{eq:gsp.core}.

%\jg{[Possible justification for the logic between the chain of inequalities.] We now claim that for each fixed $j$, the number of $i\in S$ such that $\sigma(i)\leq j<i$ is at most the number of agents $k\not \in S$ such that $k\leq j$. To see this, observe that $
%    j\geq \vert \{i\in S: i\leq j\}\vert+\vert \{i\in S: \sigma(i)\leq j<i\}\vert$. This holds as $\sigma(i)\leq i$ for all $i$ and is injective, so the sets on the right hand side are evidently disjoint and $\sigma$ maps each such element to a unique index in $[j]$.
%But $j = \vert \{i\in S: i\leq j\}\vert+\vert \{k\not\in S: k\leq j\}\vert$, so cancelling gives the claim.}

%For convenience of exposition we will not distinguish between bids and values, and we will simply assume that agents bid their values. (we will not be concerned with incentives in this analysis).
%Reindex agents in order of bid, so that $b_1 \geq b_2 \geq \dotsc \geq b_n$.  Then each agent $i \leq m$ receives slot $i$ for a click rate of $\alpha_i$ and pays

%Formally, given the bids $\mathbf{b}$, we let $\pi$ be a permutation of the agents so that $b_{\pi(1)} \geq b_{\pi(2)} \geq \dotsc b_{\pi(n)}$.  That is, $\pi(1)$ is the highest-bidding agent, then $\pi(2)$, etc.  Agent $\pi(k)$ is then allocated to slot $k$ for each $k \leq m$.  That is, $x_{\pi(k)} = \alpha_k$ for each $k \leq m$ and $x_{\pi(k)} = 0$ for each $k > m$.  The payments are set so that $p_{\pi(k)} = x_{\pi(k)} b_{\pi(k+1)}$ for all $k < n$, and $p_{\pi(n)} = 0$.

\section{Additional Discussions}

\subsection{Interpreting Liquid Welfare as Compensating Variation}
\label{app:cv}

Compensating variation (CV) is a measure of welfare change.  It refers to an amount of money an agent would need to be given in order to reach their original utility after some change in a market, typically a change in prices. For a textbook discussion, see, for example, Chapter 5 of~\cite{perloff2009microeconomics}.

For example, suppose that an agent's utility can be written as $U(w,x)$, where $w$ represents the agent's wealth (in money), $x$ is some allocation of non-monetary goods, and $U$ is  non-decreasing in $w$.  Suppose there is a default outcome $(w_0,x_0)$ for the agent, denoting their wealth and allocation prior to some proposed change.  Following some change in the marketplace, the agent's outcome shifts to $(w_1,x_1)$.  The CV of this change, for this agent, is the minimum monetary transfer $p$ such that
\[ U(w_1+p,x_1) = U(w_0,x_0), \]
if such a $p$ exists.  More generally, to allow for discontinuities in the utility function, we define the CV as the supremum of transfers for which the agent would have lower utility:
\[ \text{CV} := \sup_p \cbr{ U(w_1+p,x_1) < U(w_0,x_0) }. \]

We claim that liquid welfare can be interpreted as (the negation of the) compensating variation in our setting of a budget-constrained agent with budget $B$ and valuation function $v$ over a space $X$ of goods that includes a null outcome $\emptyset$ with $v(\emptyset) = 0$.

To see this, suppose the agent's utility function is $U(w,x) = w + v(x)$ if $w \geq 0$, and $U(w,x) = -\infty$ if $w < 0$. We will take the default outcome to be $(w_0,x_0) = (B,\emptyset)$: no allocation and a wealth of $B$.  The alternative outcome is $(w_1,x_1) = (B,x)$, an allocation $x \in X$ and no change in wealth.  If $v(x) \leq B$, then 
\[ U(B-v(x),x) = B - v(x) + v(x) = B = U(B,0) \]
so the compensating variation is $-v(x)$. If $v(x) > B$ we have
\[ U(B+p,x) = -\infty < U(B,0) \]
for all $p < -B$ and
\[ U(B+p,x) = B + p + v(x) > B = U(B,0) \]
for all $p \geq -B > -v(x)$, and hence $\text{CV} = -B$.  We conclude that the compensating variation is precisely $-\min\{B, v(x)\}$, the negation of the liquid welfare of allocation $x$.

\subsection{Vanishing Regret Does Not Imply Approximately Optimal Liquid Welfare}
\label{app:regret.example}

In this work we construct bidding algorithms that simultaneously achieve low individual regret and an approximation to the optimal aggregate liquid welfare.  This combination is reminiscent of similar results from the analysis of smooth games, which include many auction games.  No-regret learning algorithms converge (in distribution) to coarse correlated equilibria (CCE), and for smooth games it is known that the allocations obtained at CCE approximately optimize the aggregate welfare~\cite{TotalAnarchy-stoc08,Roughgarden-PoA-stoc09,roughgarden2017price}.  Thus, for auction games that satisfy smoothness conditions, achieving low regret directly implies (on a per-instance basis) an aggregate welfare approximation.

These results do not directly apply in our case: our class of games includes non-smooth games, and our welfare metric is different.  But one might still wonder whether a similar direct implication applies. In this section we present an example showing that this direct implication does not hold in our setting with budgets and liquid welfare, even for single-item second-price auctions.  We provide an auction environment and construct a coarse correlated equilibrium for the agents.  The fact that agents employ bidding strategies that forms a coarse correlated equilibrium means, in particular, that all agents achieve zero regret.  Nevertheless, in our example, the allocation that results from this bidding equilibrium has an unbounded approximation factor with respect to expected liquid welfare.
%a dynamic bidding strategy that achieves vanishing (in fact, zero) regret for all agents might still have an unbounded approximation factor with respect to the liquid welfare of the allocation sequence.

\begin{proposition}\label[proposition]{prop:no-regret low liquid welfare}
There exists an instance of a second-price auction for a single good and two agents,
described by a distribution over valuations,
and a coarse correlated equilibrium of the auction (i.e., a pair of bidding strategies under which each agent has regret zero) such that
%when both agents follow these bidding strategies,
the resulting expected liquid welfare is arbitrarily small compared to the optimal liquid welfare.
\end{proposition}
\begin{proof}
%We'll begin by showing something simpler: there exists a distribution over valuations and a bidding strategy for each agent such that, under these strategies, both agents have zero regret.  We will then show how to extend this construction to one in which each agent is using a truly no-regret bidding strategy.
%
The per-round auction in our example is a second-price auction for a single good.  There are two agents.  The distribution $F$ over value profiles is such that $(v_1,v_2) = (2,1)$ with probability $1$.
%$-\delta$, and $(v_1,v_2) = (1,2)$ with the remaining probability $\delta$.
The target per-round spend rates for the agents are $(\rho_1, \rho_2) = (1/(1+{\overline{\mu}}), 1)$, where $\overline{\mu}$ is some arbitrarily large constant independent of $T$.

% Consider the following bidding strategies for the agents.  These will be pacing strategies, in which each agent $k$ chooses some $\mu_{k,t} \leq \overline{\mu}$ for each $k \in \{1,2\}$ and each $t \leq T$, then bids $v_{k,t} / (1 + \mu_{k,t})$.  But our update rule will be different than the one used in Algorithm~\ref{alg:bg}.  Agent $1$ sets $\mu_{1,t}$ --- and hence bids $v_2$ --- in every round.  Agent $2$ chooses $\mu_{2,t} = \overline{\mu}$ in every round, and hence places the minimum bid.  Under these pacing strategies, agent $1$ wins each round and pays $1/(1+\overline{\mu})$.  Agent $1$ therefore spends her entire budget over the $T$ rounds, and agent $2$ receives nothing and pays nothing.

Consider the following bidding strategies for the agents. Agent 1 bids value $2$ for all periods and agent $2$ bids $0$ for all periods.
Under this strategy profile,
agent $1$ receives all the items over the $T$ rounds, and both agents pay nothing.

Under these strategies, agent $1$ has zero regret as she obtains the maximum possible value.  Agent $2$ likewise has zero regret, since no choice of bid less than $v_2$ can cause her to win in any round.  Note that the agents would still have zero regret if their objective were changed to maximizing value minus (any scalar multiple $\lambda \in [0,1]$ times) payments.

The liquid welfare of this equilibrium is $T/(\overline{\mu}+1)$, the total budget of agent $1$.  However, allocating all goods to agent $2$ achieves a liquid welfare of $T$.  Since $\overline{\mu}$ is an arbitrarily large constant, this approximation factor is unbounded.

\iffalse
\medskip

One possible critique of the above construction is that the strategy each agent chooses only achieves no-regret for a particular adversary.
%In the following proposition,
We now show that the ideas of the construction can be generalize to the situation
where both agents are choosing strategies that achieves no-regret against any adversaries,
while the approximation ratio of the equilibrium liquid welfare is arbitrarily large.

Let $\sigma$ be any no-regret bidding strategy against any adversary.
Let $\sigma_1$ be the bidding strategy that always bid value $v$ if the historical bids of the opponent is always $0$,
and follows the bidding strategy $\sigma$ otherwise.
Let $\sigma_2$ be the bidding strategy that always bid $0$ if the historical bids of the opponent is always $2$,
and follows the bidding strategy $\sigma$ otherwise.
Note that given any adversary,
strategy $\sigma_1$ achieves no-regret for agent $1$ given any adversary,
and strategy $\sigma_2$ achieves no-regret for agent $2$.
This is because both agents achieve no-regret
when the adversary proceeds as anticipated,
and the remaining budget of the agent until detecting any deviation does not decrease.

Finally, when both agents follow the strategy profile $(\sigma_1,\sigma_2)$,
no agent will deviate to strategy $\sigma$
and the equilibrium liquid welfare is
$T/(\overline{\mu}+1)$,
which can be unbounded smaller than the optimal liquid welfare.
\fi
\end{proof}

\subsection{Linear Pacing and Advertiser-Feasible Bidding Strategies}
\label{app:external}

In our analysis we restricted our attention to bidding strategies that are \emph{linear} in the following sense.  In each round $t$ an agent $k$ first chooses a pacing factor $\mu_{k,t} \geq 0$, \emph{then} $v_{k,t}$ is revealed and the agent bids $b_{k,t} = v_{k,t} / (\mu_{k,t} + 1)$.  The important restriction is that $\mu_{k,t}$ is chosen independently of $v_{k,t}$.  This choice of $\mu_{k,t}$ can therefore be interpreted as a linear mapping from $v_{k,t}$ to $b_{k,t}$.

When the underlying auction is truthful, this restriction to linear policies is known to be without loss.  However, for non-truthful auctions (such as a first-price auction) a value-maximizing or utility-maximizing agent might be able to strictly improve their outcome with a non-linear mapping from value to bid.  For example, in a first-price auction where the highest competing bid is known to be exactly $1$, an agent with value $2$ and an agent with value $3$ would both optimize their quasi-linear utility by bidding (slightly more than) $1$ even though this is not implementable with a linear bidding strategy.

One motivation for our restriction to linear policies is that they capture bidding strategies that are implementable by an \emph{external} agent (i.e., an advertiser or third-party bid optimizer) who does not have visibility into the precise value estimates of the platform, such as click-rate estimates.

We now make this intuition more precise by adding click rates to our auction model.  We will associate each round $t$ with a potential ad impression to be shown to a user.  The auction in round $t$ determines which ad will be shown; the user may or may not click on the ad.  Each agent $k$ has a value $v_k$ that is obtained only if their advertisement is clicked.  The advertising platform has access to an estimated click rate $c_{k,t} \in [0,1]$ that describes the likelihood that the user will click on agent $k$'s advertisement if shown.  Thus the expected value to agent $k$ of winning the auction in round $t$ is $v_k c_{k,t}$.  We will write $v_{k,t} = v_k c_{k,t}$.

In each round, the agent can place a bid $\beta_{k,t}$ that is interpreted as a willingness to pay per click.  Once all agents have placed bids, the mechanism multiplies these by the corresponding click rate estimates to determine the effective bids for winning the auction.  These are denoted $b_{k,t} = \beta_{k,t} c_{k,t}$.  The auction mechanism in round $t$ is then resolved using the effective bids $b_{k,t}$.

We note that if the click rates $c_{k,t}$ are visible to the agents then this formulation is equivalent to our model from Section~\ref{sec:model} as it simply expresses the values $v_{k,t}$ and the bids $b_{k,t}$ in a different way.  However, for an advertiser that is external to the platform, the bid $\beta_{k,t}$ must be placed without observing the realization of the estimated click rate $c_{k,t}$.  Equivalently, the estimated click rate $c_{k,t}$ is realized after the bid $\beta_{k,t}$ is fixed.  We can therefore write $\mu_{k,t} = \tfrac{v_k}{\beta_{k,t}} - 1$ which is independent of $c_{k,t}$.  We then have that for every possible realization of $v_{k,t}$,
\[ b_{k,t} = \beta_{k,t} c_{k,t} = \frac{1}{\mu_{k,t} + 1} v_k c_{k,t} = \frac{1}{\mu_{k,t} + 1} v_{k,t}  \]
and hence the advertiser's effective bidding strategy will be linear.

\newpage
\section{Omitted Proofs from \Cref{sec:liquid.welfare}: Aggregate Guarantees}
\label{app:proofs}

\subsection{Motivating ex ante liquid welfare}

Our definition of ex ante liquid welfare assumes that
the same allocation rule $y$ is used in every round.  We now show that this is without loss of generality.
%for ex ante liquid welfare:
The following lemma shows that given any allocation sequence rule
%(which might not apply the same allocation rule in every round),
there is a single-round allocation rule with the same ex ante liquid welfare.

\begin{lemma}
\label[lemma]{lem:single.round}
Let $\tilde{\mathbf{y}}:[0,\overline{v}]^{nT}\to X^{T}$ be an allocation sequence rule that takes in the entire sequence $\mathbf{v}_1,\ldots,\mathbf{v}_T$ and allocates $\tilde{y}_{k,t}(\mathbf{v}_1,\ldots,\mathbf{v}_{T})$ units to agent $k$ at time $t$. Then there exists a (single-round) allocation rule ${y}:[0,\overline{v}]^n\to X$ such that
\begin{align*}
\tilde{W}(\tilde{\mathbf{y}},F)
&\triangleq \textstyle
\sum_{k=1}^n\; \min\cbr{B_k,\;\E_{\mathbf{v}_1,\ldots,\mathbf{v}_T\sim F}\left[\sum_{t=1}^T\; \tilde{y}_{k,t}(\mathbf{v}_1,\ldots,\mathbf{v}_T) v_{k,t}\right]}\\
    &= \textstyle \sum_{k=1}^n\; T\cdot \min\left\{\rho_k,\mathbb{E}_{\mathbf{v}\sim F}\left[ y_{k}(\mathbf{v}) v_{k}\right]\right\} = \overline{W}(\mathbf{y},F).
\end{align*}
\end{lemma}

\begin{proof}
%\cbl{$W$ can't take $z$ and $y$ both an argument because they have different dim. I changed to $\overline{W}$ I don't think we use it somewhere else but worth checking }.
For each $t$, by slightly abusing notation,
we define an allocation rule
$\mathbf{\hat{y}}_{t}:[0,\overline{v}]^n\to [0,1]^n$ by
\begin{equation*}
    \hat{y}_{k,t}(\mathbf{v}_t)\triangleq \mathbb{E}_{\mathbf{v}_{-t}\in F^{T-1}}\left[y_{k,t}(\mathbf{v}_1,\ldots,\mathbf{v}_{T})\vert \mathbf{v}_t\right],
\end{equation*}
Note that this is a feasible allocation rule as
%\delbl{is a convex combination of feasible allocations}
{the set of feasible allocations is convex and closed.}  %\cbl{So the set $X$ is convex? we never assumed it I  added this assumption in the model now. Also the set must be closed. A  'convex combination' typically refers to a finite mixture.}.
We have
\begin{align*}
    \mathbb{E}_{\mathbf{v}_1,\ldots,\mathbf{v}_T\sim F}\left[ y_{k,t}(\mathbf{v}_1,\ldots,\mathbf{v}_T)\cdot v_{k,t}\right]&=\mathbb{E}_{\mathbf{v}_t}\left[\mathbb{E}_{\mathbf{v}_{-t}}\left[ y_{k,t}(\mathbf{v}_1,\ldots,\mathbf{v}_T)\cdot v_{k,t}\vert \mathbf{v}_t\right]\right]\\
    &=\mathbb{E}_{\mathbf{v}_t}\left[\hat{y}_{k,t}(\mathbf{v}_t)\cdot v_{k,t}\right]
    =\mathbb{E}_{\mathbf{v}\sim F}\left[\hat{y}_{k,t}(\mathbf{v})\cdot v_{k}\right].
\end{align*}
Now we define the allocation rule $\tilde{\mathbf{y}}$ by setting $\tilde{\mathbf{y}}_k=\frac{1}{T}\sum_{t=1}^T \hat{y}_{k,t}$
for each $k\in[n]$, which is again feasible because the set of feasible allocations is convex. By the linearity of the expectations operator, we have
\begin{align*}
    &\mathbb{E}_{\mathbf{v}_1,\ldots,\mathbf{v}_T\sim F}\left[\sum_{t=1}^T y_{k,t}(\mathbf{v}_1,\ldots,\mathbf{v}_T)\cdot v_{k,t}\right]
    =\sum_{t=1}^T \mathbb{E}_{\mathbf{v}_1,\ldots,\mathbf{v}_T\sim F}\left[y_{k,t}(\mathbf{v}_1,\ldots,\mathbf{v}_T)\cdot v_{k,t}\right]\\
    &=\sum_{t=1}^T \mathbb{E}_{\mathbf{v}\sim F}\left[\hat{y}_{k,t}(\mathbf{v})\cdot v_{k}\right]
    =\mathbb{E}_{\mathbf{v}\sim F}\left[\sum_{t=1}^T \hat{y}_{k,t}(\mathbf{v})\cdot v_{k}\right]
    =T\cdot \mathbb{E}_{\mathbf{v}\sim F}\left[\tilde{y}_k(\mathbf{v})\cdot v_k\right]. \qedhere
\end{align*}
\end{proof}

\subsection{Proof of Lemma~\ref{lem:stronger}}
The assumption that
$\mu_{k,t_2}\neq \flat$
means that the agent participates in all periods of $[t_1, t_2)$.  Moreover, if $t_2=t_1+1$, then \Cref{eq:val} is trivial as the {third} term on the right hand of side of inequality (\ref{eq:val}) is zero and $z_{k,t_1}\geq 0$. Therefore, we may assume $t_2\geq t_1+2$.

By the definition of an epoch, there is no negative projection in the dynamics on the epoch until possibly time $t_2$.  I.e., $\mu_{k,t} > 0$ for all $t_1 \leq t < t_2$.  The pacing recurrence condition implies that
\begin{align*}
0   &<\mu_{k,t_2-1}
    = P_{[0,\overline{\mu}]}(\mu_{k,t_{2}-2}+\epsilon(z_{k,t_{2}-2}-\rho_{k}))
    \leq
     \mu_{k,t_{2}-2}+\epsilon(z_{k,t_{2}-2}-\rho_{k}) \\
    & \leq \epsilon \sum_{t=t_1}^{t_2-2} (z_{k,t}-\rho_k)
     = \epsilon \left(\sum_{t=t_1}^{t_2-2}
            z_{k,t}\right)-\epsilon(t_{2}-t_1-1)\rho_k.
\end{align*}
%\cbl{Any reason why $\mu$ does not have a $k$ ? i.e., $\mu_{k,t_{2}-1}$? I added the $k$ and provided some details. Also fixed a typo in the last line $t$ --> $t_{2}$}
{The first and second inequalities follow because the multipliers are positive during an epoch. The third inequality follows from applying the second inequality repeatedly. }  Note that the inequality holds even if there is a positive projection during the epoch (i.e., if $\mu_{k,t} = \overline{\mu}$ for some $t \in [t_1, t_2)$).

Thus, the expenditure of agent $k$ on this epoch is at least $\sum_{t=t_1}^{t_2-2} z_{k,t}\geq (t_2-t_1-1)\rho_k$. Let us now consider the \emph{value} obtained by the agent on this epoch. Because agents never overbid in the pacing algorithm and because payments are always lower than the value in a core auction, the value obtained by the agent on the epoch is at least the expenditure, which we just lower bounded. To get a slight sharpening of this, note that on the first period of the epoch, the agent actually receives $x_{k,t_1}v_{k,t_1}$ value and pays $z_{k,t_1}$, which we know is at most $x_{k,t_1}v_{k,t_1}$ from the first property of core auctions and the no overbidding condition. Therefore, we can trade $z_{k,t_1}$ expenditure for $x_{k,t_1}v_{k,t_1}$ value. It follows from the above bound and this observation that
\begin{equation*}
\sum_{t=t_1}^{t_2-1} x_{k,t}v_{k,t}
\geq x_{k, t_1}v_{k,t_1}
    -z_{k,t_1}+\sum_{t=t_1}^{t_2-2} z_{k,t}
\geq x_{k,t_1}v_{k,t_1}-z_{k,t_1}+ (t_2-t_1-1)\rho_k.
\qedhere
\end{equation*}

\subsection{Proof of Lemma~\ref{lem:concentration}}

Let $\Delta_t = X_tY_t + (1-X_t)\rho- \left(\mathbb{E}[X_tY_t + (1-X_t)\rho\vert \mathcal{F}_{t-1}]\right)$. Clearly, the sequence $\Delta_t$ forms a $\mathcal{F}_t$-martingale difference sequence by construction. Moreover, we observe that
\begin{equation*}
    \mathbb{E}[X_tY_t + (1-X_t)\rho\vert \mathcal{F}_{t-1}]  = X_t\mathbb{E}[Y_t\vert \mathcal{F}_{t-1}] + (1-X_t)\rho = X_t\mathbb{E}[Y_t] + (1-X_t)\rho,
\end{equation*}
where we use the facts that  $X_t$ is $\mathcal{F}_{t-1}$-measurable and t, $Y_t$ is independent of $\mathcal{F}_{t-1}$. It follows that $\Delta_t\in [-\mathbb{E}[Y_t],\overline{v}-\mathbb{E}[Y_t]]$. As an immediate consequence of the Azuma-Hoeffding inequality, we obtain
\begin{equation}
\label{eq:ah}
    \Pr\left(\sum_{t=1}^T \Delta_t \geq \theta\right)\leq \exp\left(\frac{-2\theta^2}{T\overline{v}^2}\right).
\end{equation}
But observe that
\begin{align*}
    \left\{\sum_{t=1}^T \Delta_t \geq  \theta\right\} &= \left\{\sum_{t=1}^T [X_tY_t+(1-X_t)\rho] \geq \theta + \sum_{t=1}^T [X_t\mathbb{E}[Y_t] + (1-X_t)\rho]\right\}\\
    &\supseteq \left\{\sum_{t=1}^T [X_tY_t+(1-X_t)\rho] \geq \theta + T\cdot \rho \right\},
\end{align*}
using the assumption $\mathbb{E}[Y_t]\leq \rho$. This inclusion together with \Cref{eq:ah} yields \Cref{eq:concentration}.
\newpage

\section{Online Convex Optimization: Regret Bounds from Prior Work}
\label{app:OCO}
%\section{Dynamic Regret Guarantees for SGD}
%\label{sec:sgd}

Recall that we establishing individual guarantees through a reduction to \emph{online convex optimization} (OCO). This section summarizes the relevant background and guarantees from prior work on OCO.

The guarantees for \cref{alg:bg} invoke a reduction to \emph{stochastic gradient descent (SGD)}, a very standard OCO algorithm. The additional guarantees in \Cref{sec:other} invoke several other gradient-based OCO algorithms: \emph{optimistic gradient descent (OGD)}
\citep{Rakhlin-colt13,Rakhlin-nips13,mokhtari-aistats20,Jiang-siam25}, \emph{optimistic mirror descent} (OMD) \citep{Rakhlin-colt13,Rakhlin-nips13}, and \emph{optimistic follow-the-regularized-leader} (OFTRL) \cite{Rakhlin-colt13} with Euclidean regularizer.

\begin{algorithm}
\caption{Online Convex Optimization with gradient estimates}\label{alg:OCO}
\textbf{Environment Parameters:} convex set $\mK\subset \R$, time horizon $T$.

Algorithm chooses a point $x_1\in \mK$.

\For{$t=1,\ldots,T$}{

Nature chooses a convex function $f_t:\mK\to\R$, possibly depending on history $\mH_t$

\tcp{history $\mH_t$ comprises $\rbr{x_s, f_s, \myGrad[s]}$ for each round $s<t$}

Nature chooses a gradient estimate $\myGrad[t]\in \R$, possibly at random

\tcc{$\myGrad[t]$ is an unbiased estimate of $\nabla f_t(x_t)$
given $\mH_t$: $\E\sbr{\myGrad[t]\mid\mH_t}=\nabla f_t(x_t)$}

Algorithm observes $\myGrad[t]$ and chooses $x_{t+1}\in\mK$.
}
\end{algorithm}

The general framework of OCO with gradient estimates is summarized in
\Cref{alg:OCO}. While we focus on $\mK\subset \R$, we note in passing that all algorithms and guarantees extend to $\mK\subset \R^d$, $d\in\N$.

Within this framework, the particular algorithms are defines as follows. Throughout, $\eps>0$ denotes the step-size, and $P_{\mK}(x)$ denotes the projection of $x\in\R$ into $\mK$. For OMD and OFTRL, in each round $t$, an algorithm is given an estimate $M_t$ of the \emph{next} gradient $\nabla f_{t+1}(x_{t+1})$. Thus:

\begin{description}
\item[SGD]
$ x_{t+1} \leftarrow P_{\mK}
    \rbr{x_t-\eps\cdot\myGrad[t]}$.

\item[OGD]
$ x_{t+1} \leftarrow P_{\mK}
    \rbr{x_t-\eps\cdot\myGrad[t] - \eps'\cdot\myGrad[t-1]}$,
where $\eps' \in [-\eps/2,0]$ is another parameter.

\item[OMD]
$ \hat{x}_{t+1} \leftarrow
    P_{\mK}
        \rbr{\hat{x}_{t}-\eps\cdot\myGrad[t]}$ and
 $x_{t+1} \leftarrow
    P_{\mK}
        \rbr{\hat{x}_{t+1}-\eps \cdot M_t)}$,
where $\hat{x}_1=0$.

\item[OFTRL]
$ x_{t+1} \leftarrow
    P_\mK
        \rbr{-\eps\rbr{M_t+{\textstyle \sum_{s\in[t]}}\; \myGrad[s]}}$.
\end{description}

\begin{remark}
The original versions of OMD and OFTRL with Euclidean regularizer are defined for $\mK\subseteq \R^d$, $d\in\N$ and involve explicit regularizers and Bregman divergences. We only provide a simpler  equivalent formulation for $d=1$. The default choice for $M_t$ is $M_t=\myGrad[t]$, see \Cref{rem:predictable}.
\end{remark}

We use the following guarantees from prior work:

\begin{theorem}\label{thm:sgd}
Consider the framework in \Cref{alg:OCO}. Fix some numbers $D\geq 1$, $G>0$ and $P\geq 1$. Assume the feasible set $\mK$ has diameter at most $D$, and that  $|\myGrad[t]|\leq G$ almost surely for all rounds $t$. Fix an arbitrary comparator sequence $u_1,\ldots,u_T \in \mK$ satisfying $\sum_{t=1}^{T-1} \|u_{t+1}-u_t\|_2 +1\leq P$ almost surely. For OMD and OFTRL, posit that $\vert M_t\vert \leq G$ for all rounds $t$. Then SGD, OGD and OMD satisfy the following regret bound, for any given step size $\eps>0$:
\begin{align}\label{eq:thm:sgd}
    \E\sbr{\sum_{t=1}^T f_t(x_t) - \sum_{t=1}^T f_t(u_t)} \leq
    O\rbr{ \frac{D^2 P}{\eps}+\eps\cdot G^2 T}.
\end{align}
Moreover, OFTRL satisfies \refeq{eq:thm:sgd} for the stationary comparator (when all $u_t$'s are the same).
\end{theorem}

\begin{remark}
The result for SGD and OGD is essentially contained in the proof of
\citep[Theorem 10.1.1]{Hazan-OCO-book}. Since the latter only considers the special case of (deterministic) gradient descent, we spell out the proof below for the sake of completeness. The result for OMD is contained in \citep[Lemma 1]{jadbabaie2015online}. The result for OFTRL follows from \citep[Theorem 19 and Lemma 20]{Syrgkanis-nips15}, extending an earlier analysis of \citep{rakhlin2013online}.
\end{remark}

\subsection{Proof of \Cref{thm:sgd} for SGD and OGD}

%The proof \jgedit{follows by adapting existing proofs of stochastic gradient descent and dynamic regret while accounting for the optimistic step}.
%\jgedit{We show how to adapt the dynamic regret proof of} Theorem 10.1.1 of \cite{Hazan-OCO-book} \jgedit{for the optimistic setting}.

Focus on OGD (SGD being a special case with $\eps'=0$).
%Observe that under the assumption $\eps'\in [-\eps/2,0]$, we know that $\eps+\eps'\geq \eps/2$ and $-\eps'\leq \eps+\eps'$. Hence, by redefining $\eps \leftarrow \eps + \eps'$ and $\eps' \leftarrow - \eps' $,
It will be convenient to study the equivalent reparameterization of the dynamics:
\begin{equation}
\label{eq:reparameterize-final}
    x_{t+1} = P_{\mK}
    \rbr{x_t-\eps\cdot\myGrad[t] - \eps'\cdot N_t},
\end{equation}
where we define $N_t: = \widetilde{\nabla}_t-\widetilde{\nabla}_{t-1}$ and now have the equivalent assumption that $0 \leq \eps'\leq \eps$. (We've redefined $\eps \leftarrow \eps + \eps'$ and $\eps' \leftarrow - \eps' $.) Since we changed the original value of $\eps$ by a factor of at most $2$, it suffices to prove the regret bound for \eqref{eq:reparameterize-final}.

We have deterministically by convexity that
\begin{align*}
    \sum_{t=1}^T f_t(x_t) - \sum_{t=1}^T f_t(u_t) &\leq \sum_{t=1}^T \nabla f_t(x_t)^T(x_t - u_t)
     = \E\sbr{\sum_{t\in[T]} \tilde{\nabla_t}^T(x_t - u_t)},
\end{align*}
where the expectation is over all random choices in the dynamics using the assumption on the gradient estimates. %At this point, the dynamics follow online \jgedit{(optimistic)} gradient descent with cost functions $\tilde{f}_t(x) = \tilde{\nabla_t}^Tx$.
%For convenience, write $M_t = \widetilde{\nabla}_t-\widetilde{\nabla}_{t-1}$.
Observe that
\begin{align*}
    \|x_{t+1}-u_t\|^2&\leq \|x_t-\epsilon \widetilde{\nabla}_t-\epsilon' N_t-u_t\|^2\\
    &=\|x_t-u_t\|^2+\|\eps\widetilde{\nabla}_t+\eps'N_t\|^2-2\epsilon \widetilde{\nabla}_t^T(x_t-u_t)-2\epsilon' N_t^T(x_t-u_t)\\
    &\leq \|x_t-u_t\|^2+9\epsilon^2 G^2-2\epsilon \widetilde{\nabla}_t^T(x_t-u_t)-2\epsilon' N_t^T(x_t-u_t)
\end{align*}
We first show that the last term is small. We surely have
\begin{align*}
    \sum_{t=1}^T N_t^T(x_t-u_t)&=\sum_{t=1}^T \widetilde{\nabla}_t^T(x_{t+1}-x_{t})+\sum_{t=1}^T\widetilde{\nabla}_t^T(u_{t+1}-u_{t})\\
    &\leq G\left(\sum_{t=1}^T \|x_{t+1}-x_t\|+\sum_{t=1}^T\|u_{t+1}-u_t\|\right)\\
    &\leq G\left(\sum_{t\in[T]} 3G\epsilon+P\right)
    =O(\epsilon G^2T + GP).
\end{align*}
where we simply use the definition of $N_t$ and rearrange terms. Therefore, we have
\begin{align*}
    2\sum_{t=1}^T \widetilde{\nabla}_t^T (x_t-u_t)&\leq \sum_{t=1}^T \frac{\|x_{t+1}-x_t\|^2-\|x_t-u_t\|^2}{\epsilon}+O\left(\eps'G^2T + \frac{ \eps' GP}{\eps}\right)\\
    &\leq \sum_{t=1}^T \frac{\|x_{t+1}-x_t\|^2-\|x_t-u_t\|^2}{\epsilon}+O\left(\eps G^2T + GP\right),
\end{align*}
where the extra contribution is from the last term we just bounded. Here, we use the fact that $\eps'\leq \eps$. At this point, the sum is bounded exactly as in Theorem 10.1.1 of \cite{Hazan-OCO-book} for an overall bound of
\begin{align*}
    \sum_{t=1}^T \nabla f_t(x_t)^T(x_t - u_t)\leq  O\left(\frac{D^2}{\epsilon}P+\epsilon G^2T + GP\right).
\end{align*}
Moreover, it is easy to see that $GP$ is bounded by the other terms (up to constants) for any choice of $\epsilon$. This is because the minimum is (up to constants) $DG\sqrt{PT}\gtrsim D^{1/2}GP\geq GP$. Here, we use the fact that $P\leq 1+2DT$ and the assumption $D\geq 1$.

\section{Omitted Proofs from \Cref{sec:regret}: Individual Guarantees}

\subsection{Stochastic environment}
\label{apx:proof stochastic}

We need to show that the perfect pacing multiple $\nu^*$ approximately optimizes the objective over all multipliers, for (a) value maximization and (b) second-price auctions. (While part (b) is implicit in \citep{BalseiroGur19}, we provide a proof here for the sake of completeness.)

\begin{lemma}\label[lemma]{lm:regret-stochastic-bench}
In the setting of \Cref{cor:regret-stochastic-gen}, assume either that (a) the objective is value-maximization: $\uniObj[t]=V_t$, or that (b) the auction rule is second-price. Then
\begin{align}
\uniObjFixed(\nu^*) \geq \sup_{\mu\in[0,\overmu]} \uniObjFixed(\mu) - O(\sqrt{T}).
\end{align}
\end{lemma}

The rest of this subsection proves this Lemma.

%\begin{lemma}
%In the setting of \Cref{cor:regret-stochastic-gen}, assume either that (a) the objective is value-maximization: $\uniObj[t]=V_t$, or that (b) the auction rule is second-price. Then
%\begin{align}
%\uniObjFixed(\nu^*)\geq \sup_{\mu\in[0,\overmu]} \uniObjFixed(\mu) -O(\sqrt{T}).
%\end{align}
%\end{lemma}

%\ascomment{This is a new version, the old version is below, pls remove eventually.}

%\begin{proposition}[Optimality in Stochastic Environments]
%\label{prop:stochasticopt-old}
%Suppose that \Cref{assn:regularity2}  holds \BLcomment{No need for assumptions 5.3 here? Then we can use Thm 5.5 to claim the optimality?} for a sequence of auctions satisfying MBB such that $V_t\equiv V$ and $Z_t\equiv Z$ for some fixed functions $V,Z$. For $\mu\in [0,\overline{\mu}]$, \BLnew{let} $Y_{\mu}$ to be the expected value of the bidding strategy that paces with multiplier $\mu$ for each time step until running out of budget. Then $
%    \sup_{\mu\in [0,\overline{\mu}]} Y_{\mu} \leq T\cdot V(\mu^*) + \frac{\overline{v}^2}{\rho}$.  \BLcomment{Something is not very smooth for me when reading this Prop in terms of what was shown. I tried to add a sentence}. \BLnew{Hence, combining with Theorem \ref{thm:regmain}, Algorithm 1  with step size  $\epsilon=1/\sqrt{T}$ has no regret compared to the best fixed pacing multiplier. }
%\end{proposition}

Fix the mixing parameter $\gamma\in [0,1]$ in the unified objective. Denote $U_t=U$, $V_t = V$, and $Z_t=Z$ (this being a stochastic environment, there's no dependence on the round $t$). Fix some pacing multiplier $\mu\in [0,\overline{\mu}]$. We will repeatedly use the fact that $\lambda U+(1-\lambda)V = V - \lambda Z$ by definition, and the same identity holds pointwise as well.

Define the random process $M_k = \sum_{t=1}^k \left[x_tv_t-\lambda z_t - (V(\mu)-\lambda Z(\mu))\right],
$ where $x_t$ and $v_t$ are the (random) allocation and valuations at each time $t$ when bidding with multiplier $\mu$ and $z_t$ is the expenditure. Notice that $M_k$ is a martingale by construction and $\mathbb{E}[M_{1}]=0$. Define $\tau_{\mu}$ to be first time that such a bidding algorithm runs out of money at the end of the period, or $T$ if this does not occur. Because $M_k$ is a martingale, by the Optional Stopping Theorem,
\begin{equation}
\label{eq:Mmart}
     \uniObjFixed(\mu)\leq \mathbb{E}\left[\sum_{t=1}^{\tau_{\mu}} (x_tv_t-\lambda z_t)\right]= \mathbb{E}\left[\sum_{t=1}^{\tau_{\mu}} (V(\mu)-\lambda Z(\mu))\right] = \mathbb{E}[\tau_{\mu}]\cdot (V(\mu)-\lambda Z(\mu)),
\end{equation}
where the inequality occurs because the stopping rule counts the (non-negative) value or utility obtained on the period where the agent may (strictly) exceed the budget, and then using Wald's identity in the last step as the function $V(\mu)$ is constant over time by assumption so independent of $\tau_{\mu}$. %\BLnew{If I understand correctly I think you mean to define $M_{k} = \sum x_{t}v_{t}$ (also below for $z_{t}$) and then use Wald's first identity to deduce the equality. But this means that you assume Independence.. do we assume it somewhere?}

We now derive an upper bound on the expectation of the stopping time. On the one hand, we certainly have $\mathbb{E}[\tau_{\mu}]\leq T$ as $\tau _{\mu} \leq T$ by definition. An analogous martingale argument with $N_k = \sum_{t=1}^k \left[z_t-Z(\mu)\right]$ (where $z_t$ is the (random) expenditure at time $t$) similarly implies that
\begin{equation}
\label{eq:Nmart}
    \mathbb{E}[\tau_{\mu}]\cdot Z(\mu) =  \mathbb{E}\left[\sum_{t=1}^{\tau_{\mu}}z_t\right]\leq B+\overline{v},
\end{equation}
where the extra term arises in the analysis because the agent may spend $\overline{v}$ on the final round before she exceeds her budget. Combining \Cref{eq:Mmart} and \Cref{eq:Nmart} implies that $\uniObjFixed(\mu)\leq (\gamma U(\mu)+(1-\gamma) V(\mu))\cdot \min\left\{T,\frac{B+\overline{v}}{Z(\mu)}\right\}$.

Suppose first that $\mu\geq \nu^*$. In this case, it is clear that $V(\mu)\leq V(\mu^*)$ by monotonicity of the auction. This implies that when $\gamma=0$, we have
\begin{equation*}
\uniObjFixed(\mu)\leq V(\mu)\cdot T\leq V(\nu^*)\cdot T,
\end{equation*}
proving part (a).

For part (b), assume further that the auction is second-price. Then we have the inequality $U(\mu)\leq U(\nu^*)$ since utilities are non-increasing in the multiplier for any fixed valuation and competing bid so long as the pacing multipliers are nonnegative. It follows that for second-price auctions, for any $\gamma\in [0,1]$
\begin{equation*}
    \uniObjFixed(\mu)\leq (\lambda U(\mu)+(1-\lambda)V(\mu))\cdot T\leq (\lambda U(\nu^*)+(1-\lambda)V(\nu^*))\cdot T
\end{equation*}

If instead $\mu<\mu^*$, note that by definition this implies that $Z(\mu^*)=\rho$. For any $\gamma\in [0,1]$ and any MBB auction (not necessarily second-price), we then have:
\begin{align*}
    \uniObjFixed(\mu)&\leq (B+\overline{v})\frac{V(\mu)-\gamma Z(\mu)}{Z(\mu)}\\
    &= (B+\overline{v})\left(\frac{V(\mu)}{Z(\mu)}-\lambda\right)\\
    &\leq (B+\overline{v}) \left(\frac{V(\nu^*)}{Z(\nu^*)}-\lambda\right)\\
    &=\frac{B+\overline{v}}{\rho}(V(\nu^*)-\lambda Z(\mu^*))\\
    &\leq T\cdot (\lambda U(\nu^*)+(1-\lambda)V(\nu^*))+\overline{v}^2/\rho.
\end{align*}
where the inequality is the monotone bang-for-buck property (see the proof of Theorem \ref{thm:regmain}) and then using the fact that $B/\rho=T$.

In either case, a similar Wald argument then shows that
\begin{equation*}
    \uniObjFixed(\nu^*)\geq \left(\gamma U(\nu^*)+(1-\gamma)V(\nu^*)\right)\cdot \mathbb{E}[\tau^*-1]
\end{equation*}
where $\tau^*$ is the stopping time of the strategy that paces using $\nu^*$. However, it is easy to see that the expected stopping time is at least $T-O(\sqrt{T})$ by considering the fluctuations of the super-martingale $z_{t}-\rho$ since $\mathbb{E}[z_t]:=Z(\nu^*)\leq \rho$ (using e.g. the Azuma-Hoeffding inequality and integrating the tail). The desired inequality follows, noting that we treat the other parameters as absolute constants by assumption.

\subsection{Artificial objective $H_t$}
\label{app:H_t}

\begin{lemma}\label[lemma]{lem:H}
The function $H_t$ is convex and is $(\overline{v}+\rho)$-Lipschitz.
\end{lemma}

\begin{proof}
The convexity is immediate from the fundamental theorem of calculus and monotonicity assumption on $Z_t$, as the expected expenditure is a weakly decreasing function of the multiplier. To show that the Lipschitz condition holds note that
\begin{equation*}
    \vert H_t(y)-H_t(x)\vert\leq \rho \vert y-x\vert + \left \vert \int_x^y Z_t(s)\mathrm{d}s\right\vert\leq (\rho+ Z_t(0)) \vert y-x\vert.
\end{equation*}
The proof of the Lipschitz condition follows from the fact the expected expenditure function is at most the expected valuation.
\end{proof}

\subsection{Proof of \Cref{lem:lipint}}

The conclusion is trivial if $\lambda=0$, so we assume $\lambda>0$. Moreover, note that we always have $R\geq 0$, and by considering the function $-f(-y)$ if necessary, we may assume without loss of generality that $x\geq 0$ and thus the same for $f(x)$.

For a contradiction, suppose that $f(x)> \sqrt{2\lambda R}$. Then by the assumption that $f$ is $\lambda$-Lipschitz, it follows that $f(y)> \sqrt{2\lambda R} - \lambda(x-y)$ for all $y\in [x-\sqrt{2R/\lambda},x]$ (note $f(y)>0$ on this region so this region is contained in $[0,x]$, where $f(y)\geq 0$). Hence, we have
\begin{equation*}
    \int_0^x f(y)\mathrm{d}y\geq \int_{x-\sqrt{2R/\lambda}}^x f(y)\mathrm{d}y> \int_{x-\sqrt{2R/\lambda}}^x \left[\sqrt{2\lambda R} - \lambda(x-y)\right]\mathrm{d}y.
\end{equation*}
%\BLnew{Note the area of integration above changed inequality 2 and 3.. It should be the same (equal to the latter one)?}
This latter integral gives the area of a right triangle with height $\sqrt{2\lambda R}$ and base $\sqrt{2R/\lambda}$, which is clearly equal to $R$. This contradicts the assumption that $R=\int_0^x f(y)\mathrm{d}y$, proving the lemma.

\section{Omitted Proofs from \Cref{sec:other}: Other Pacing Algorithms}
\label{app:other}

%\jgedit{In the rest of this section we complete the proof of \Cref{thm:other}.} We first establish \Cref{thm:other-general} by verifying that the general condition of \Cref{def:other-general} suffices for liquid welfare bounds:

%\begin{theorem}[\Cref{thm:other-general}, restated]
%Fix any core auction and any distribution $F$ over agent value profiles. Suppose that each agent $k$ uses \Cref{alg:bg} with a \emph{$c$-event-feasible update}, for some $c\in(0,1)$ (possibly a different algorithm for different agents). Write $\boldsym{x}$ for the corresponding allocation sequence rule. Then for any allocation rule $\boldsym{y} \colon [0,\overline{v}]^{n} \to X$ we have
%\begin{align}\label{eq:thm:other-general}
%\asedit{W(\boldsym{x},F)}
%\geq (1-c)\cdot \frac{\overline{W}(\boldsym{y},F)}{2} - %O\rbr{n\overline{v}\sqrt{T\log(\overline{v}nT)}}.
%\end{align}
%\end{theorem}

%In the context of \Cref{thm:main.new}, if we define $\mE_{k,t}=\{\mu_{k,t}=0\}$, we exactly recover the crucial inequality in \eqref{eq:group.sum.new} that was obtained using the epoch decomposition in \Cref{lem:stronger}. In this case, the generalized pacing condition implies we can take $c=0$, and the conclusion thus recovers our original bound. The above argument simply abstracts the properties we require, and we will later show using a similar epoch analysis that each of the previous algorithms in fact satisfy this condition with small $c>0$.

\subsection{The algorithms do not run out of budget too early}

We state and prove an analogue of \Cref{lem:bg1} for \PacingOGD, \PacingOMD, and \PacingOFTRL, showing that the algorithms do not run out of budget too early. For this result, \PacingOMD and \PacingOFTRL can use arbitrary next-gradient predictors $M_{k,t}\leq \overline{v}$.  The only change compared to \Cref{lem:bg1} is the
slightly stronger condition on the problem parameters:
    $\overline{\mu}\geq 2\overline{v}/\rho_k+1$.

\begin{lemma}\label[lemma]{lem:other-runout}
Fix agent $k$ in a core auction with some (possibly adaptive, randomized, adversarially generated) set of valuations and competing bids. Suppose the agent  uses \PacingOGD, \PacingOMD, or \PacingOFTRL (the latter two algorithms with arbitrary next-gradient predictors satisfying $\vert M_{k,t}\vert \leq \overline{v}$). Let $\tau_k$ be the algorithm's stopping time. Assume all valuations are at most $\overline{v}$, and the parameters satisfy
    $\overline{\mu}\geq 2\overline{v}/\rho_k+1$
and
    $\epsilon_k\, \overline{v}\leq 1$.
Then $
    T-\tau_k\leq \frac{\overline{\mu}}{\epsilon_k \rho_k}+\frac{\overline{v}}{\rho_k}$
almost surely.
\end{lemma}

The analysis is modified from Proposition 2 of \citet{Balseiro-BestOfMany-Opre}.

\begin{proof}
Following the proof of \citep[Proposition 2]{Balseiro-BestOfMany-Opre}, it suffices to show that under the stated conditions on $\varepsilon_k$ and $\overline{\mu}$, either the pacing multipliers or the auxiliary sequences (depending on the setting) remain strictly below $\overline{\mu}$ at all times. In each case, the remainder of the argument to bound the stopping time then becomes identical to that in \citet{Balseiro-BestOfMany-Opre}.

\xhdr{\PacingOMD.}
We first show this for \PacingOMD with respect to the auxiliary sequence $\widehat{\mu}_{k,t}$. Observe that we deterministically have
\begin{equation}
\label{eq:spend_small}
    z_{k,t}\leq \frac{\overline{v}}{1+\mu_{k,t}}\leq \frac{\overline{v}}{1+\widehat{\mu}_{k,t}-\eps_k \vert M_{k,t}\vert}\leq \frac{\overline{v}}{1+\widehat{\mu}_{k,t}-\eps_k \overline{v}}\leq \frac{\overline{v}}{\widehat{\mu}_{k,t}}.
\end{equation}
The first inequality holds since expenditure is at most the bid, which in turn can be bounded using $\widehat{\mu}_{k,t}$ after accounting for the optimistic step. The second inequality holds by the assumption $\eps_k\overline{v}\leq 1$. Note that we also deterministically always have:
\begin{equation}
\label{eq:pacer_ub_spend}
     \widehat{\mu}_{k,t+1}\leq \widehat{\mu}_{k,t}+\eps_k z_{k,t}.
\end{equation}
We will now combine these inequalities to show that $\widehat{\mu}_{k,t}<\overline{v}/\rho_k +1$ for all $t$.

First, suppose that at some time $t$, it holds that
$ \widehat{\mu}_{k,t}\leq \overline{v}/\rho_k$. Then \eqref{eq:pacer_ub_spend} immediately implies that
\begin{equation*}
    \widehat{\mu}_{k,t+1}\leq \widehat{\mu}_{k,t}+\eps_k z_{k,t}\leq \widehat{\mu}_{k,t}+1< \overline{v}/\rho_k+1,
\end{equation*}
as claimed.

Suppose instead that $\overline{v}/\rho_k\leq \widehat{\mu}_{k,t}< \overline{v}/\rho_k +1$. Then by \eqref{eq:spend_small}, it immediately follows that $z_{k,t}\leq \rho_k$.
If $\widehat{\mu}_{k,t+1}>0$, then the recursion satisfies:
\begin{equation*}
    \widehat{\mu}_{k,t+1}\leq \widehat{\mu}_{k,t}+\eps_k(z_{k,t}-\rho_k)\leq \widehat{\mu}_{k,t}< \overline{v}/\rho_k +1.
\end{equation*}
If $\widehat{\mu}_{k,t+1}=0$, then we are done as well, so unconditionally $\widehat{\mu}_{k,t+1}$ remains strictly below this threshold as well in this case.

Therefore, it follows that under our assumption that $\overline{\mu}\geq 2\overline{v}/\rho+1>\overline{v}/\rho+1$, it deterministically holds that $\widehat{\mu}_{k,t}<\overline{\mu}$ at all times $t$, and thus there is never any projection back to this endpoint in this recurrence. As stated above, the bound on the stopping time then follows the exact same analysis as in Proposition 2 of \cite{Balseiro-BestOfMany-Opre}, which is based on writing out the full recurrence for $\widehat{\mu}_{k,t}$ until the stopping time in terms of the expenditure. %Expanding the auxiliary sequence recursion thus shows that until the stopping time $\tau_k$,
%\begin{equation*}
%    \sum_{t=1}^{\tau_k} (z_{k,t}-\rho_k)=B-\tau_k\rho_k\leq \frac{\mu_{k,\tau_k+1}-\mu_1}{\eps_k}\leq \frac{\overline{\mu}}{\eps_k}.
%\end{equation*}
%Dividing by $\rho_k$ shows that
%\begin{equation*}
%    T- \tau_k\leq \frac{\overline{\mu}}{\rho_k\eps_k},
%\end{equation*}
%as claimed.

\xhdr{\PacingOFTRL.}
The argument for \PacingOFTRL is nearly identical upon defining the auxiliary sequence
\begin{equation*}
    \widehat{\mu}_{k,t+1}=P_{[0,\overline{\mu}]}
        \rbr{-\eps_k\rbr{{\textstyle \sum_{s\in[t]}}\; \myGrad[k,s]}}.
\end{equation*}
The casework depending on the value of $\widehat{\mu}_{k,t}$ is the same as for \PacingOMD and establishes that $\widehat{\mu}_{k,t}$ remains strictly below $\overline{\mu}$ under the same assumptions on $\overline{\mu}$ and $\eps_k$.

\xhdr{\PacingOGD.}
For \PacingOGD, we show that the pacing multiplier again remains strictly below $\overline{\mu}$ under the stated assumptions. We argue as follows: observe that if $\mu_{k,t}\geq \frac{2\overline{v}}{\rho_k}-1$, then we must have
\begin{equation*}
z_{k,t}\leq \frac{\overline{v}}{1+\mu_{k,t}}\leq \frac{\rho_k}{2}.
\end{equation*}
Suppose that for some $t$, it holds that $\mu_{k,t}\geq \frac{2\overline{v}}{\rho_k}-1$. We claim that in this case, $\mu_{k,t+1}\leq \mu_{k,t}$. Indeed, either $\mu_{k,t+1}=0$ and we are already done, or otherwise there is no negative projection step and hence the \PacingOGD recursion implies that
\begin{equation*}
    \mu_{k,t+1}\leq \mu_{k,t}+\eps_k(z_{k,t}-\rho_k)+\eps_k'(z_{k,t-1}-\rho_k)\leq \mu_{k,t}- \eps_k\frac{\rho_k}{2}+\eps_k\frac{\rho_k}{2}\leq \mu_{k,t}.
\end{equation*}
In the penultimate step, we use the assumption in \PacingOGD that $\eps'_k\in [-\eps_k/2,0]$.
Thus, it follows that $\mu_{k,t+1}\leq \mu_{k,t}$ unconditionally when $\mu_{k,t}\geq \frac{2\overline{v}}{\rho_k}-1$.

Since \PacingOGD and our assumption $\eps_k\overline{v}\leq 1$ deterministically implies that
\begin{equation*}
    \mu_{k,t+1}\leq \mu_{k,t}+2\eps_k\overline{v}\leq \mu_{k,t}+2,
\end{equation*}
it follows that under \PacingOGD dynamics, $\mu_{k,t}$ can never exceed $\frac{2\overline{v}}{\rho_k}+1$. Indeed, we have shown that the sequence must be non-increasing when above $\frac{2\overline{v}}{\rho_k}-1$ in any iteration, and so can only ever exceed this threshold by $2$ from a single step. Therefore, since we assumed that $\overline{\mu}\geq \frac{2\overline{v}}{\rho_k}+1$, it follows that $\mu_{k,t}<\overline{\mu}$ for all times $t$.
\end{proof}

\subsection{Proof of \Cref{thm:other-general}: Liquid Welfare under Event-Feasible Algorithms}

Under the assumptions, we may follow the proof of \Cref{thm:main.new} with relatively minor modifications. In this case, we simply again define:

\begin{equation*}
    R_k(\boldsym{v}) \triangleq
    \textstyle \sum_{t=1}^T\; \left[\boldsym{1}\{\mE_{k,t}\}y_{k}(\boldsym{v})v_{k,t} + \boldsym{1}\{\mE_{k,t}\}\rho_k\right].
\end{equation*}
Note that the concentration inequality of \Cref{lem:concentration} applies equally well to this quantity since the events $\mE_{k,t}$ are determined by the history through time $t-1$ by definition.

At this point, we may simply follow the proof, again defining $A\subseteq [n]$ as the set of agents such that $\sum_{t=1}^T x_{k,t}v_{k,t}<B_k$. By our assumption, for each such agent $k\in A$, we have the analogous inequality
\begin{equation*}
    \sum_{t=1}^{T} x_{k,t}v_{k,t}\geq \sum_{t=1}^T (x_{k,t}v_{k,t}-z_{k,t})\cdot \mathbf{1}(\mE_{k,t})+(1-c)\sum_{t=1}^T\rho_k \cdot \mathbf{1}(\mE_{t,k}^c),
\end{equation*}
and again upon summing over $k\in A$, we derive exact analogue of \eqref{eq:group.sum.new}:

\begin{equation*}
    \sum_{k \in A}\sum_{t=1}^{T} x_{k,t}v_{k,t} \geq \sum_{t=1}^{T}\sum_{k \in A}[\boldsym{1}\{\mE_{k,t}\}(x_{k,t}v_{k,t} - z_{k,t})] + (1-c)\cdot\sum_{k \in A}\sum_{t = 1}^{T}\boldsym{1}\{\mE_{k,t}^c\}\cdot \rho_k.
\end{equation*}
We may now again appeal to the core auction assumption in the same way: for any $t\in [1,T]$, if we set $S\subseteq A$ to be the set of agents $k$ satisfying $\mE_{k,t}$, we find
\begin{align*}
&\textstyle
\sum_{k \in A}[\boldsym{1}\{\mu_{k,t}=0\}(x_{k,t}v_{k,t} - z_{k,t})]
= \sum_{k \in S}(x_{k,t}v_{k,t} - z_{k,t}) \\
&\qquad\geq \sum_{k \in S}x_{k,t}b_{k,t} - \sum_{k\in S}z_{k,t}\\
&\qquad \geq \sum_{k \in S}y_k(\boldsym{v}_t)b_{k,t} - \sum_{k =1}^{n}z_{k,t}\\
&\qquad \geq (1-c)\sum_{k \in S}y_k(\boldsym{v}_t)v_{k,t} - \sum_{k =1}^{n}z_{k,t}\\
&\qquad= (1-c)\sum_{k \in A}\boldsym{1}\{\mE_{k,t}\}y_k(\boldsym{v}_t)v_{k,t} - \sum_{k=1}^n z_{k,t}.
\end{align*}
Here, we crucially use the fact that $v_{k,t}\geq b_{k,t}$ but that \Cref{def:other-general} imposes a reverse inequality up to a factor $(1-c)$. We thus obtain the analogous inequality  \eqref{eq:RkLB} on the good event that concentration held:
\begin{equation*}
    \sum_{k\in [n]} \mathsf{WEL}_{k,\GPD}(\boldsym{v})\geq (1-c)\cdot\sum_{k\in [n]} R_k(\boldsym{v}) - \sum_{k\in [n]}\sum_{t\in[T]} z_{k,t} - n\overline{v}\sqrt{T\log(\overline{v}nT)}.
\end{equation*}
The rest of the proof extends verbatim up to this factor of $(1-c)$ in the final bound.

\subsection{Proof of \Cref{lem:other-prop}: Our Algorithms are Event-Feasible}

%\begin{lemma}
%    Suppose that each agent $k$ uses any of \PacingOGD, \PacingOMD, and \PacingOFTRL, possibly with different step-sizes satisfying $ \epsilon_k\geq 0$ (and $\vert \eps'_k\vert\leq \eps_k/2$ in the case of \PacingOGD). Further suppose that each $M_{k,t}=O(\overline{v})$. Then the assumptions of \Cref{thm:other-general} hold with $c=O\left(\overline{v}\sqrt{\max_k\epsilon_k/\rho_k}\right)$.
%
%    \jgedit{In each case, the events $\mE_{k,t}$ will be of the form $\left\{\mu_{k,t}'\leq O\left(\overline{v}\sqrt{\epsilon_k/\rho_k}\right)\right\}$ where $\mu'$ is either $\mu$ or $\widehat{\mu}$ in the case of \PacingOMD.}
%\end{lemma}
%\begin{proof}

We need to prove that \PacingOGD, \PacingOMD, and \PacingOFTRL are $O(\overline{v}\sqrt{\eps_k/\rho_k})$-event-feasible.
Fix agent $k$ in what follows and define
\begin{equation*}
    c':=2\overline{v}\sqrt{\eps_k/\rho_k}.
\end{equation*}

\xhdr{\PacingOGD.} Let
        $\mE_{k,t} = \cbr{\mu_{k,t}\leq c'}$.
As in \Cref{lem:stronger}, we will prove the desired conditions hold on each epoch individually, where epochs are defined as before. Let $[t_1,t_2)$ denote a maximal epoch for agent $k$ so that $\mu_{k,t_1}=0$ and $\mu_{k,t_2-1}>0$. Because there are no negative projections on an epoch by assumption, writing out the gradient recurrence and rearranging yields:
    \begin{align*}
    \sum_{t=t_1}^{t_2-2}z_{k,t}&\geq (t_2-t_1-1)\rho_k + \frac{\eps_k'}{\eps_k+\eps_k'}(z_{k,t_2-2}-z_{k,t_1-1})\\
    &\geq (t_2-t_1-1)\rho_k - \overline{v}.
    \end{align*}
    In the last step, we use the assumption that $\vert \eps_k'\vert\leq \eps_k/2$ so that the ratio is bounded by $1$ in absolute value and drop a negative term.

    Say that the epoch is \emph{long} (resp. \emph{short}) iff the length of the interval satisfies:
    \begin{equation*}
        t_2-t_1\geq \frac{2\overline{v}}{\rho_k\,c'}.
    \end{equation*}
    It is immediate from basic algebra and the assumption $\overline{v}\geq \rho_k$ that on any long epoch:
    \begin{equation}
    \label{eq:long_epoch}
        \sum_{t=t_1}^{t_2-2}z_{k,t}\geq (1-c')(t_2-t_1)\rho_k.
    \end{equation}
    If an interval is short, then since each step of \PacingOGD can increase the pacer by at most $2\eps_k \overline{v}$ by assumption and the number of steps is at most $\frac{2\overline{v}}{\rho_k\,c'}$, the pacing multiplier can be at most
    \begin{equation*}
       2\eps_k \overline{v}\cdot \frac{2\overline{v}}{\rho_k\,c'}=\frac{4\eps_k \overline{v}^2}{\rho_k\,c'}=2\overline{v}\sqrt{\eps_k/\rho_k}=c'.
    \end{equation*}

    Therefore, the entire epoch will satisfy $\mE_{k,t}$ by construction. This means that for any short epoch,
    \begin{equation*}
        \sum_{t=t_1}^{t_2-1} x_{k,t}v_{k,t}\geq \sum_{t=t_1}^{t_2-1} (x_{k,t}v_{k,t}-z_{k,t})=\sum_{t=t_1}^{t_2-1} (x_{k,t}v_{k,t}-z_{k,t})\cdot \mathbf{1}(\mE_{k,t})+\sum_{t=t_1}^{t_2-1}\rho_k\cdot \mathbf{1}(\mE_{k,t}^c)
    \end{equation*}
    If instead the epoch is long, then we can instead bound using \eqref{eq:long_epoch}:
    \begin{align*}
        \sum_{t=t_1}^{t_2-1}x_{k,t}v_{k,t} &= \sum_{t=t_1}^{t_2-1}(x_{k,t}v_{k,t}-z_{k,t})+\sum_{t=t_1}^{t_2-1}z_{k,t}\\
        &\geq \sum_{t=t_1}^{t_2-1}(x_{k,t}v_{k,t}-z_{k,t})\cdot\mathbf{1}(\mE_{k,t})+(1-c')(t_2-t_1)\rho_k\\
        &\geq \sum_{t=1}^T (x_{k,t}v_{k,t}-z_{k,t})\cdot \mathbf{1}(\mE_{k,t})+(1-c')\sum_{t=t_1}^{t_2-1}\rho_k \cdot\mathbf{1}(\mE_{k,t}^c).
    \end{align*}
    The first inequality holds from nonnegativity of quasilinear utility under any pacing algorithm in an individually rational auction. Finally, note that on $\mE_{k,t}$,
    \begin{equation*}
        \frac{1}{1+\mu_{k,t}}\geq \frac{1}{1+c'}\geq 1-c'.
    \end{equation*}
    Doing this for each agent $k$ and taking the minimum of $\eps_k/\rho_k$ yields the desired conditions with $c=c'$.

\xhdr{\PacingOMD.} We can reduce to the previous analysis. Since the auxiliary sequence $\widehat{\mu}_{k,t}$ follows \PacingOGD (with $\eps_k'=0$), defining the events
    \begin{equation*}
        \mE_{k,t} = \cbr{\widehat{\mu}_{k,t}\leq c'},
    \end{equation*}
    the same lower bound on valuations holds by an identical analysis.
    Since $\vert \mu_{k,t}-\widehat{\mu}_{k,t}\vert\leq \eps_k\vert M_k\vert \leq \eps_k\overline{v}$ by assumption, the event $\mE_{k,t}$ implies that
    \begin{equation*}
        \mu_{k,t}\leq 2\overline{v}\sqrt{\eps_k/\rho_k}+\eps_k\overline{v}\leq 4\overline{v}\sqrt{\eps_k/\rho_k},
    \end{equation*}
    and so we satisfy the desired properties after taking $c = 2c'$.

\xhdr{\PacingOFTRL.} We use a similar analysis. Again, we prove the desired inequality on each epoch $[t_1,t_2)$ separately defined in terms of the $\mu_{k,t}$. The key observation is that by definition of an epoch:
    \begin{align*}
        \mu_{k,t_1}=0&\implies \eps_k\sum_{s=1}^{t_1-1}(z_{k,s}-\rho_k)+\eps_k(z_{k,t_1-1}-\rho_k)\leq 0\\
        \mu_{k,t_2-1}>0&\implies \eps_k\sum_{s=1}^{t_2-2}(z_{k,s}-\rho_k)+\eps_k(z_{k,t_2-2}-\rho_k)>0,
    \end{align*}
    so subtracting the first inequality from the latter yields
    \begin{equation*}
        \eps_k\sum_{s=t_1}^{t_2-2}(z_{k,s}-\rho_k)+ \eps_k (z_{k,t_2-2}-z_{k,t_1-1})> 0.
    \end{equation*}
    After rearrangement and dropping a negative term, we again find that
    \begin{equation*}
        \sum_{s=t_1}^{t_2-2}z_s> (t_2-t_1-1)\cdot \rho_k - \overline{v}.
    \end{equation*}
    But this is identical to the lower bound we derived for \PacingOGD, and so we can follow the same exact analysis by defining the events
        $\mE_{k,t}=\cbr{\mu_{k,t}\leq c'}$
    to obtain the desired conclusion.
%\end{proof}

%\input{app-other-general}

\section{Additional Numerical Simulation Plots}
\label{app:expts}

Our empirical estimation of regret rate in Section~\ref{sec:expts} builds on a hypothesis that regret accumulates at a rate of $\Theta(T^\alpha)$ for some $\alpha \in [0,1]$.  To provide additional evidence for this hypothesis, we illustrate traces of regret evolution for 100 randomly-sampled campaigns from among those that exceed our participation threshold of at least $\theta = 1000$ auction instances.  For each campaign, we trace out the accumulation of regret as the time horizon (i.e., the number of auction instances) is amplified.  Our hypothesis corresponds to linear evolution of these traces when plotted in log-log scale.  See Figure~\ref{fig:regret}.

\begin{figure}[h]
    \centering
    \begin{tabular}{ccc}
        \includegraphics[width=0.4\textwidth]{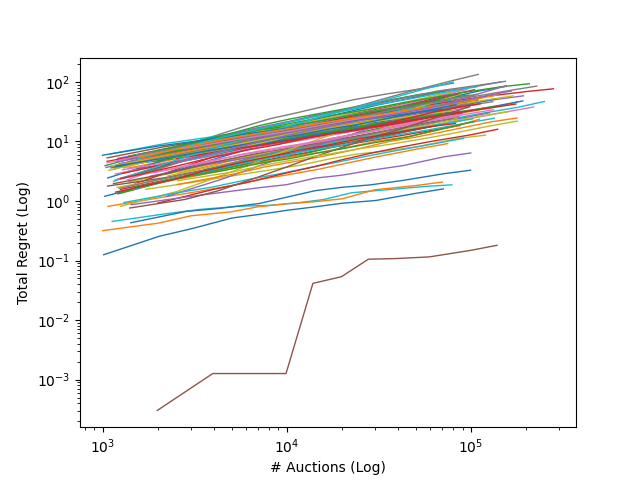}
        &
        \includegraphics[width=0.4\textwidth]{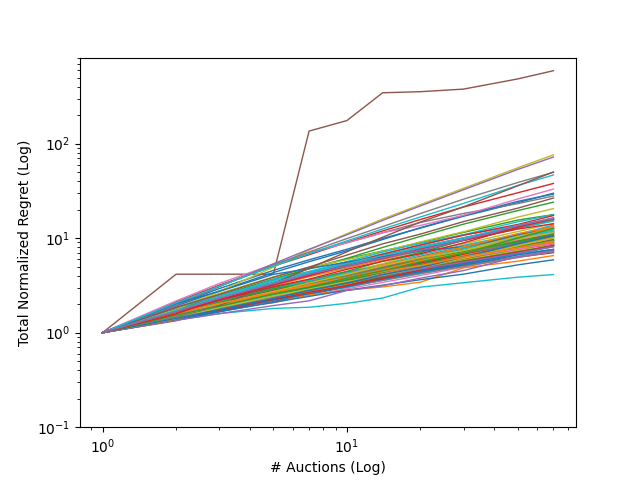}\\
        (a) utility-maximization in SPA & ... on normalized data\\

        \includegraphics[width=0.4\textwidth]{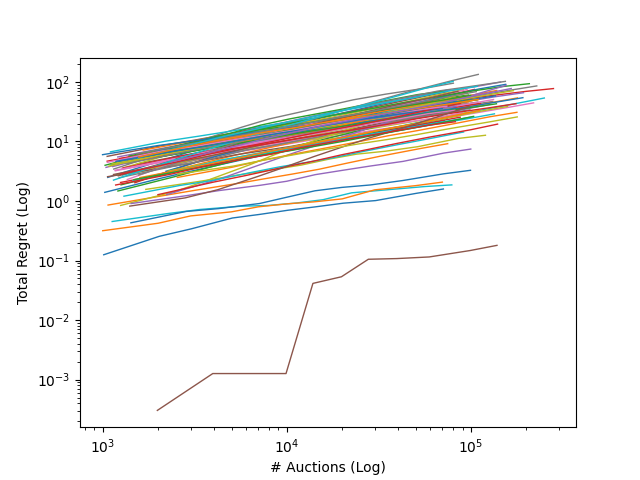}
        &
        \includegraphics[width=0.4\textwidth]{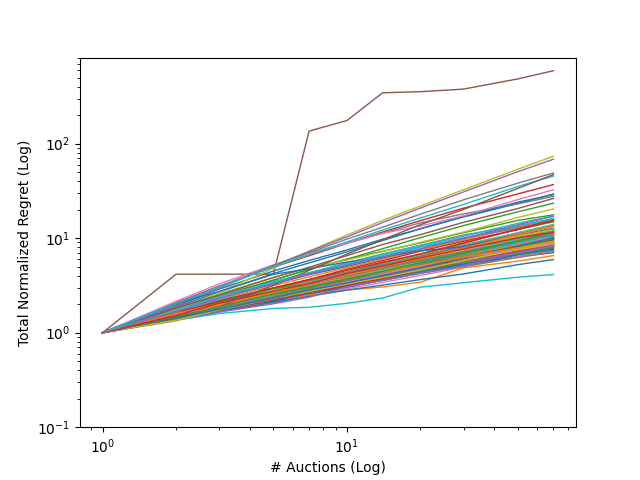}\\
        (b) value-maximization in SPA & ... on normalized data\\

        \includegraphics[width=0.4\textwidth]{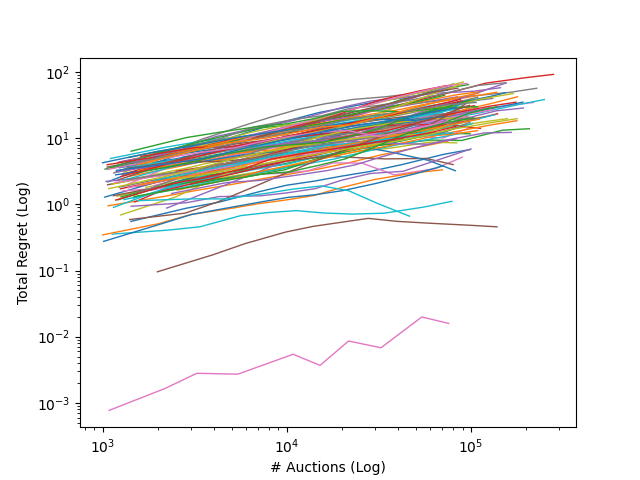} 
        &
        \includegraphics[width=0.4\textwidth]{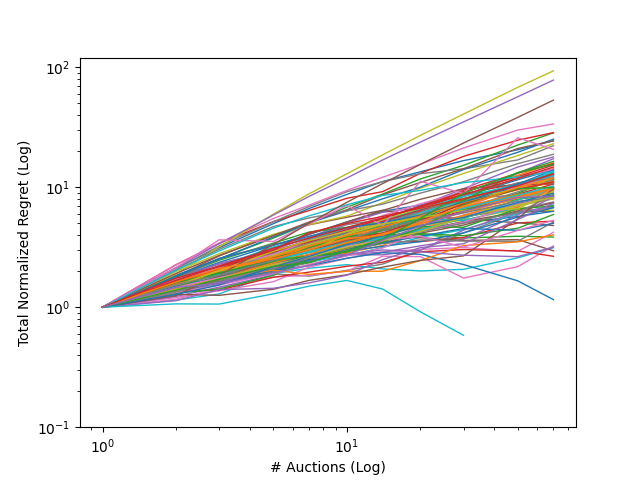}\\
        (c) value-maximization in FPA & ... on normalized data
%        \includegraphics[width=0.3\textwidth]{plot_traces_util.png}
%        &
%        \includegraphics[width=0.3\textwidth]{plot_traces_spa.png}
%        &
%        \includegraphics[width=0.3\textwidth]{plot_traces_fpa.png} \\
%        (a) & (b) & (c) \\
%        \includegraphics[width=0.3\textwidth]{plot_traces_normalized_util.png}
%        &
%        \includegraphics[width=0.3\textwidth]{plot_traces_normalized_spa.png}
%        &
%        \includegraphics[width=0.3\textwidth]{plot_traces_normalized_fpa.png} \\
%        (d) & (e) & (f) \\
    \end{tabular}
    \caption{Illustration of estimated regret in repeated auction simulation.  For each of 100 randomly-sampled campaigns, we trace out the evolution of regret as the time horizon (the number of auction instances) is amplified, illustrated in log-log scale.  Each line corresponds to a single campaign, plotted for different choices of the amplification parameter K (the number of iterations through the dataset).  Results are shown for (a) utility maximization in second-price auctions, (b) value maximization in second-price auctions, and (c) value maximization in first-price auctions. Instances with negative regret are excluded.  In the right column we show the same data but normalized so that all traces begin at (1,1).}
    \label{fig:regret}
\end{figure}

%\bjledit{Add visualization of data traces, normalized to start at same point, for both FPA and SPA (SGD only)}

\end{document}